\documentclass[journal,comsoc]{IEEEtran}

\ifCLASSINFOpdf
\else
\fi
\usepackage{graphicx}
\usepackage[numbers,sort&compress]{natbib}
\usepackage{mathtools}
\usepackage{amsmath}
\usepackage{amssymb}
\usepackage{amsthm}
\newtheorem{theorem}{Theorem}
\newtheorem{corollary}{Corollary}
\newtheorem{lemma}{Lemma}
\theoremstyle{remark}
\newtheorem{remark}{Remark}
\usepackage{capt-of} 
\usepackage{algorithm}
\usepackage{algorithmic}
\usepackage{subfigure}
\usepackage{diagbox}

\usepackage[table]{xcolor} 
\usepackage{colortbl}      
\usepackage{booktabs}      
\definecolor{lightblue}{rgb}{0.85,0.92,0.98} 

\makeatletter
\def\@IEEEBIOskipN{1.5\baselineskip}
\makeatother

\begin{document}

%
\title{Fractional-order SIS-UAU subdiffusion model for rumor propagation}

\author{Qin~Li,~\IEEEmembership{Member,~IEEE,}
        Jiexi~Tang,
        J\"{u}rgen~Kurths,
        Matja\v{z}~Perc,~\IEEEmembership{Member,~IEEE,}
        and~Minyu~Feng,~\IEEEmembership{Senior~Member,~IEEE,}%
\thanks{This work was supported by the Natural Science Foundation of Chongqing under Grant No.~CSTB2025YITP-QCRCX0007, the Fundamental Research Funds for the Central Universities under Grant Nos.~XJ2026004101 and SWU-KT26010, and the Slovenian Research and Innovation Agency under Grant No.~P1-0403 (corresponding authors: Matja\v{z} Perc and Minyu Feng).}%
\thanks{Qin Li is with the Business College, Southwest University, Chongqing 402460, China.}%
\thanks{Jiexi Tang is with the School of Journalism and Communication, Beijing Normal University, Beijing 100875, China.}%
\thanks{J\"{u}rgen Kurths is with the Potsdam Institute for Climate Impact Research, 14437 Potsdam, Germany, and also with the Department of Physics, Humboldt University of Berlin, 10117 Berlin, Germany.}%
\thanks{Matja\v{z} Perc is with the Faculty of Natural Sciences and Mathematics, University of Maribor, Koro\v{s}ka cesta 160, Maribor, 2000, Slovenia, Community Healthcare Center Dr. Adolf Drolc Maribor, Ulica talcev 9, Maribor, 2000, Slovenia, Department of Physics, Kyung Hee University, 26 Kyungheedae-ro, Seoul, 02447, Dongdaemun-gu, Republic of Korea, and University College, Korea University, 145 Anam-ro, Seoul, 02841, Seongbuk-gu, Republic of Korea (email: matjaz.perc@gmail.com).}%
\thanks{Minyu Feng is with the College of Artificial Intelligence, Southwest University, Chongqing 400715, China (email: myfeng@swu.edu.cn).}}


\markboth{IEEE TRANSACTIONS ON SYSTEMS, MAN, AND CYBERNETICS: SYSTEMS}%
{Li \MakeLowercase{\textit{et al.}}: \MakeLowercase{Fractional-Order SIS-UAU Subdiffusion Model for Rumor Propagation}}
%



\IEEEtitleabstractindextext{%
\begin{abstract}
Understanding information and rumor diffusion in complex 
social systems is essential for risk management. Classical SIS-UAU (Susceptible-Infected-Susceptible / Unaware-Aware-Unaware)  
models rely on Markovian assumptions and thus fail to capture the 
memory effects and subdiffusive behaviors commonly observed in real 
social interactions. To address this limitation, we propose a fractional-order SIS-UAU model based on Caputo derivatives, embedding non-Markovian memory into a duplex awareness–rumor spreading process. Methodologically, we derive the corresponding mean-field equations and establish existence, uniqueness, positivity, boundedness, and a Matignon-type stability criterion. The basic reproduction number is obtained, revealing how awareness dynamics reshape the effective rumor threshold. Numerically, we develop a block-adaptive FFT (Fast Fourier Transform) solver to efficiently handle long-horizon fractional dynamics, enabling large-scale simulations. Our results show that fractional memory significantly slows the transient evolution of rumors without altering the steady state. Furthermore, infection-layer parameters dominate long-term rumor prevalence, awareness transmission substantially suppresses rumor levels, and multilayer network topology—especially degree heterogeneity—critically influences propagation outcomes. Collectively, these findings underscore the necessity of integrating memory effects, awareness dynamics, and network structure for realistic rumor modeling.
\end{abstract}

}

\maketitle

\IEEEdisplaynontitleabstractindextext

%
\IEEEpeerreviewmaketitle

\section{Introduction}\label{sec:intro}
\small
\IEEEPARstart{T}{he} study of rumor propagation and information diffusion on complex networks has attracted sustained attention in recent years due to its profound impact on public opinion formation, social stability, and online communication dynamics. Rumor spreading processes exhibit rich nonlinear behaviors shaped by human cognitive responses, network topology, and time-dependent interactions, motivating the development of increasingly sophisticated mathematical models. Early studies explored the formation of rumor patterns through reaction-diffusion frameworks and examined how time delay, media correction, and self-correction mechanisms influence the emergence of spatial heterogeneity in rumor dynamics \cite{li2024turing}. Compartmental modeling approaches have further elucidated the fundamental spreading mechanisms of rumors in both homogeneous and heterogeneous complex networks, revealing threshold phenomena, stability properties, and bifurcation behaviors \cite{li2019dynamical, zhu2020delay}. In addition, research on cyber-rumor propagation over mobile and social platforms has incorporated behavioral features, such as memory effects, rumor denial, and user heterogeneity, leading to more realistic representations of online propagation processes \cite{liu2019modeling}. More recent efforts have further incorporated complex social dynamics, such as the game-theoretic competition between rumors and anti-rumors \cite{xiao2022diffusion}, community-aware suppression mechanisms \cite{gao2023multi}, and evolutionary game models with optimal control \cite{li2025rumor}. Further bridging behavior and dynamics, studies on time-dependent infection rates modulated by self-isolation awareness also highlight how adaptive individual responses fundamentally alter spreading patterns \cite{yuan2024impact}. Recent developments have also emphasized the importance of designing adaptive suppression strategies that incorporate dynamic environmental feedback and multi-round decision processes, advancing the practical management of rumor proliferation in social networks \cite{he2025adaptive}. Beyond rumor-specific studies, broader advances in the theory of signal propagation on complex networks have highlighted the central role of network structure and nonlinear node interactions in shaping diffusion processes, further motivating the need to investigate rumor spreading within multilayer and dynamical network environments \cite{yuan2023stability}.

In recent years, substantial progress has been made in understanding coupled information–risk dynamics on complex and multilayer networks. Early multiplex epidemic–awareness frameworks revealed that information diffusion on a virtual layer can fundamentally reshape epidemic thresholds and alter macroscopic outbreak patterns on the physical contact layer \cite{granell2013dynamical}. This framework has been extensively refined to investigate diverse coupling mechanisms, including the competitive effects of positive and negative preventive information \cite{wang2020epidemic}, the spatiotemporal interplay within metapopulation networks \cite{gao2022epidemic}, the coevolution of vaccination opinions and awareness \cite{dai2019coevolution}, and the impact of individual layer preference in temporal multiplex networks \cite{li2021epidemic}. Subsequent works have refined this framework by incorporating effects of both local and global information \cite{zhang2022effect}, and modeling the unaware-aware-unaware (UAU) diffusion process on structured networks \cite{li2022limited}. Subsequent studies have shown that node heterogeneity, temporal network evolution, and the placement of informed individuals can significantly modulate global propagation behavior and critical transitions \cite{liu2020impact}. More recent works have further distinguished multiple information types and incorporated indirect-contact pathways, demonstrating that both constructive and detrimental signals, together with higher-order contacts, exerted nontrivial influence on the final prevalence and system-level stability \cite{li2024multi}. Additionally, physical-layer cues in cyber–physical systems have been shown to affect awareness density and thereby raise effective epidemic thresholds in applied environments \cite{yuan2025impacts}. These developments form part of a broader theoretical landscape of network signal propagation. Parallel to these theoretical advances, data-driven approaches combining epidemiological models with artificial intelligence have also emerged for epidemic forecasting, highlighting the role of public awareness in shaping transmission trends \cite{zheng2020predicting}.

Parallel to these advances, a rich body of research has emphasized that many real spreading processes—biological or informational—exhibit memory, non-Markovian waiting times, and subdiffusive transients, motivating the use of fractional-order operators \cite{metzler2000random}. In particular, memory has been identified as a core mechanism in rumor and information diffusion: the cumulative effects of memory shape rumor dynamics on social networks \cite{zhang2015rumor}, and the bursty, heavy-tailed timing of human actions underlying these processes has been documented empirically \cite{barabasi2005origin}. Correspondingly, fractional-order rumor propagation models have been developed, e.g., with memory effects \cite{gao2022fractional}, with the Mittag--Leffler kernel \cite{singh2019new}, incorporating media reports \cite{yue2024dynamics}, and with refutation mechanisms on complex networks \cite{zhu2025dynamic}. These fractional-order rumor models, however, are formulated as single-layer compartmental systems in which the fractional memory acts on one rumor-related compartment and the bidirectional coupling between awareness formation and risk transmission is absent. By contrast, the present work embeds fractional memory in a two-layer UAU-SIS coevolution framework, where the history-dependent dynamics of the awareness layer feeds back into the risk layer and infection, in turn, forces an immediate awareness update; it is this bidirectional awareness--risk coupling, together with the complete threshold and stability theory of the fractional system, that constitutes the main departure from the existing fractional rumor literature. Fractional epidemic models have been shown to capture long-term memory effects, anomalous diffusion, and heavy-tailed infection or recovery times more accurately than their integer-order counterparts. Such models have been applied to a variety of settings, ranging from generic fractional epidemic dynamics \cite{ahmed2007fractional, pan2021fractional} and fractional-order COVID-19 models with improved data-fitting performance and optimal control strategies \cite{rajagopal2020fractional, hu2025modeling, hu2025dynamic} to study complexity, chaos, and bifurcation structures in fractional SAIR/SIR systems \cite{wang2025complexity, akrami2020hopf}. Complementary approaches also exist, such as employing stochastic queueing models to incorporate individual protective behaviors, offering alternative perspectives on capturing heterogeneity in epidemic dynamics \cite{li2021protection}. Further efforts include fractional discrete-time models based on Caputo difference operators \cite{nabil2024impact} and positivity-preserving numerical schemes for fractional SIS models \cite{hoang2020dynamics}. The utility of fractional operators has been demonstrated in modeling real-world outbreaks, such as capturing memory effects in Monkeypox transmission using fractional neural networks \cite{igiri2025monkeypox}, and in analyzing malware propagation dynamics in wireless sensor networks \cite{srivastava2023generalized}. More recently, fractional-order approaches have been applied to a range of infectious diseases, including measles modeling using Caputo--Fabrizio operators~\cite{nisar2025assessment}, malaria transmission via piecewise Caputo fractional operators~\cite{farman2025dynamics}, and neurological disorder modeling with fractal--fractional operators~\cite{farman2025modeling}. Beyond epidemiological settings, fractional-order dynamical analysis and control techniques have also been successfully applied to financial chaotic systems \cite{johansyah2022dynamical} and supply chain network dynamics \cite{johansyah2024global}, further illustrating the versatility of fractional-order modeling across diverse complex systems. Together, these works demonstrate the empirical relevance and modeling flexibility of fractional calculus in capturing memory-embedded spreading processes. Furthermore, non-Markovian stochastic simulations have demonstrated that non-exponential interevent time distributions can drastically alter epidemic spreading behavior on complex networks \cite{boguna2014simulating}, while percolation-based analyses of social distancing reveal that memory-dependent behavioral interventions fundamentally reshape propagation outcomes \cite{valdez2013social}. Recent data-driven approaches combining compartmental models with deep neural networks further underscore the practical relevance of memory-aware epidemic forecasting \cite{cheng2025leveraging}. While high-resolution rumor propagation datasets with network-level interaction records and awareness-state annotations are currently scarce, the growing availability of online social media data and digital contact-tracing records is rapidly closing this gap, making empirical calibration and quantitative validation of memory-aware rumor models an increasingly feasible and important direction.

Despite these advances, significant gaps remain and limit the applicability of existing fractional-order approaches to awareness–epidemic coevolution in particular the following ones: (1) most fractional epidemic models focus on single-layer biological contagion and do not incorporate a bidirectional coupling between awareness diffusion and risk transmission that underlies modern information–risk systems \cite{ahmed2007fractional, rajagopal2020fractional, pan2021fractional, wang2025complexity, akrami2020hopf, nabil2024impact, hoang2020dynamics}. As a result, memory-induced effects such as delayed awareness adoption, non-Markovian behavioral responses, and history-dependent risk perception—central to real information–risk coevolution—remain insufficiently understood. (2) Although heterogeneity, indirect contacts, and additional information sources have been shown to strongly shape epidemic outcomes in multilayer settings \cite{liu2020impact, li2024multi, yuan2025impacts}, current fractional-order models rarely integrate these components within a unified framework. (3) The theoretical understanding of stability, bifurcation, and threshold behavior in fractional, multiplex, and behaviorally adaptive systems remains limited, particularly when nonlocal temporal operators interact with multilayer network topology. (4) Scalable numerical solvers tailored to fractional, multilayer, and node-level stochastic processes are still underdeveloped, hindering rigorous calibration and large-scale simulation necessary for policy-relevant predictions \cite{he2022fractional, hu2024stability, minhos2025mathematical}. These limitations highlight the need for a comprehensive fractional-order framework capable of capturing memory, heterogeneity, and multilayer coupling in awareness–risk coevolution.

Motivated by these limitations of existing models that either neglect memory effects or overlook bidirectional coupling between awareness formation and risk diffusion, this study develops a new fractional-order UAU-SIS spreading framework to investigate how memory effects and awareness--risk coupling jointly influence rumor propagation in multiplex networks. Rather than simply introducing another variant of the classical SIS-UAU model, the main focus of this work is to characterize the theoretical properties of the fractional-order system, clarify the role of the fractional derivative in describing non-Markovian spreading dynamics, and examine how the fractional order, awareness mechanism, and network structure affect rumor suppression and transient evolution. Building upon the classical awareness-epidemic coevolution paradigm, the main contributions of this work are summarized as follows.

(1) We construct a novel multiplex network model that couples an information layer (governing awareness updating) with a physical layer (regulating risk propagation), and embed Caputo fractional derivatives into its dynamics to capture the non-Markovian, history-dependent behavioral responses of informed populations.

(2) The resulting fractional-order model simultaneously incorporates several realistic key features: multilayer structural interactions, heterogeneous awareness sensitivity, fractional temporal memory, and feedback-driven state transitions, thereby extending traditional UAU-SIS dynamics to a generalized setting.

(3) We establish a complete theoretical framework, which includes deriving the fractional mean-field equations, proving the system’s positivity and boundedness, determining its equilibrium conditions, and rigorously characterizing both rumor-free and rumor-endemic equilibria via a reproduction-number-based criterion. This provides the first theoretical description of stability and threshold behavior for fractional awareness–risk coevolution.

(4) We develop an adaptive FFT-based numerical solver for the resulting Caputo fractional-order differential equations and perform comprehensive simulations to demonstrate how memory strength, awareness responsiveness, and cross-layer coupling reshape outbreak profiles, delay phase transitions, and modify critical thresholds.

Collectively, these advancements provide a unified, memory-aware analytical framework for understanding and predicting coupled information–risk dynamics in high-knowledge social systems. To the best of our knowledge, this is the first study that embeds fractional-order memory in a two-layer awareness–risk coevolution framework and provides a complete threshold, local-stability, and global-stability theory for the resulting multiplex system, together with a mechanism-based stochastic validation of the mean-field closure. 

The rest of this paper is organized as follows. In Section~\ref{sec:modeling}, we develop a fractional-order extension of the classical SIS-UAU model to capture memory effects and establish its theoretical foundations. Section~\ref{sec:simulation} presents a comprehensive numerical study, where we examine the impact of fractional order, infection dynamics, awareness feedback, and network topology on rumor spreading. Section~\ref{sec:conclusion} concludes the paper by summarizing our main findings and suggesting potential avenues for future research. Technical details and proofs of all theoretical results are provided in the Supplementary Material.

\section{Model formulation and theoretical analysis}\label{sec:modeling}
To realistically capture the coevolution of information and risk in settings with bursty communication and memory effects---where individual decision-making exhibits memory effects and subdiffusive dynamics---we extend the classical awareness-epidemic paradigm to a fractional-order framework. This section serves as the analytical core of our work, structured in two main parts. First, we present the Fractional-Order SIS-UAU Model. We begin by situating the model within the context of informed, socially active populations and then detail its construction: starting from the classical two-layer UAU-SIS network model, we derive its mean-field formulation and subsequently generalize it by replacing the ordinary time derivative with a Caputo fractional derivative of order $\alpha\in(0,1)$. This key step formally incorporates the non-Markovian, historically weighted dynamics into the system. Second, we conduct a complete theoretical analysis of the Fractional-Order System. We establish the well-posedness of the model by proving the existence, uniqueness, positivity, and boundedness of its solutions. We then derive the basic reproduction number $\mathcal{R}_0$ 
and perform a comprehensive stability analysis, providing rigorous criteria for the local asymptotic stability of both the rumor-free equilibrium (RFE) and the rumor-endemic equilibrium. This systematic formulation and analysis lay the essential mathematical foundation for understanding the interplay between network structure, awareness diffusion, risk propagation, and memory in the subsequent numerical investigations.

\subsection{Fractional-Order SIS-UAU Modeling}\label{sec:fractional_modeling}
Our model is formulated within settings with bursty communication and memory effects, such as universities, research institutes, or high-technology companies. In such settings, individuals frequently exchange information and collaborate to jointly address the challenges posed by risk and crisis propagation. Due to the distinctive characteristics of individuals in these environments, the information diffusion process exhibits \textit{memory effects}: individuals make judgments and decisions based on past experiences and previously received information. Moreover, the diffusion process demonstrates a \textit{subdiffusive} nature, i.e., the information spreads more slowly than the exponential growth predicted by traditional models. A schematic representation of the entire model is shown in Fig.~\ref{fig:fractional}.

\begin{figure*}[htbp]
  \centering
  \subfigure[]{
        \includegraphics[width=0.20574683\linewidth]{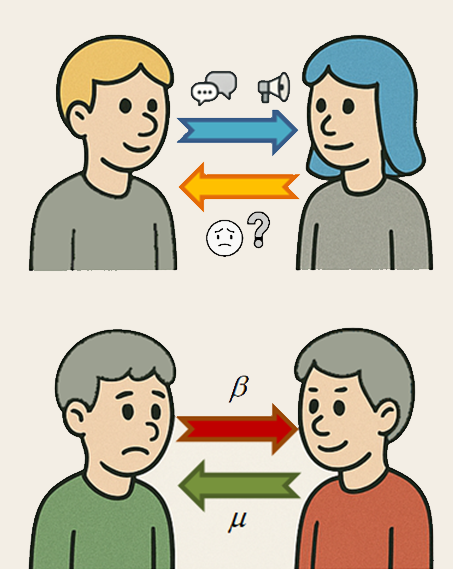}
        \label{fig:fractional_1}
    }
  \subfigure[]{
        \includegraphics[width=0.19530052\linewidth]{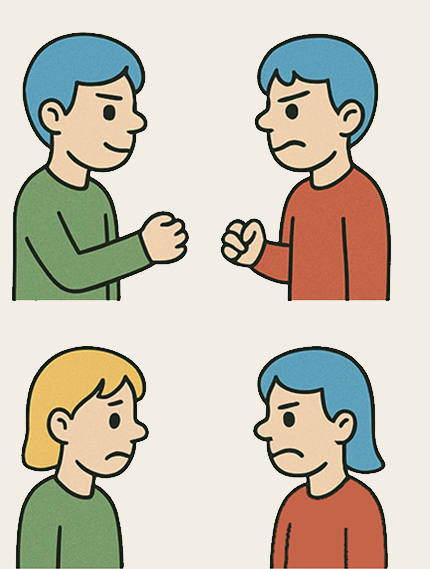}
        \label{fig:fractional_2}
    }
  \subfigure[]{
        \includegraphics[width=0.1918\linewidth]{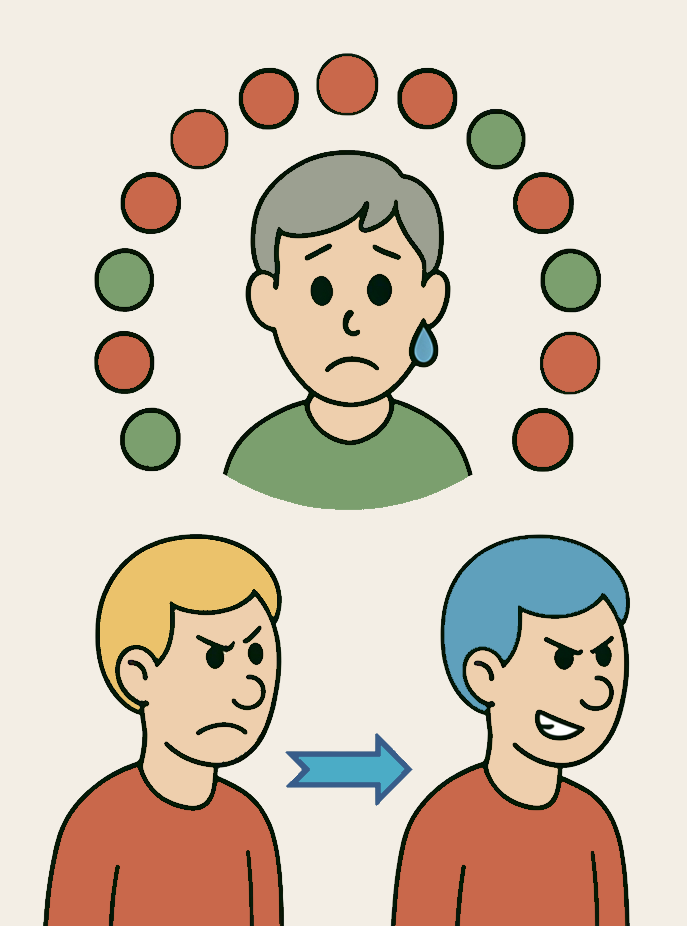}
        \label{fig:fractional_3}
    }
  \subfigure[]{
        \includegraphics[width=0.19620889\linewidth]{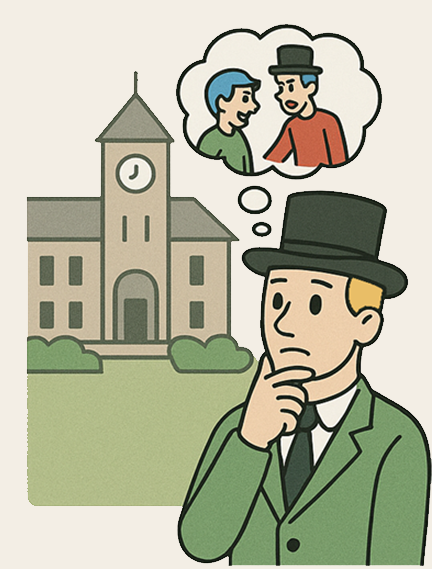}
        \label{fig:fractional_4}
    }
    \vspace{-3mm}
  \caption{
    \textbf{Fractional-order SIS-UAU rumor diffusion model.} 
    The diagram illustrates the fundamental mechanisms of the fractional-order SIS-UAU model. 
    Individuals with yellow and blue hair represent the unaware and aware states, respectively; 
    green and red clothing indicate the susceptible and infected states, while gray is used to suppress visually the attribute that remains unchanged in a given sub-panel, directing attention to the active transition mechanism. (a) Information diffusion occurs through local interactions in the 
    consciousness layer, such as face-to-face conversations or online communications. Unaware individuals become aware through these local contacts, while 
    aware individuals may spontaneously forget and revert to the unaware state. Simultaneously, 
    in the risk layer, rumor propagation follows contact-based transmission, where susceptible 
    and infected individuals transition between states at certain rates through interactions 
    with infected neighbors. (b) Aware individuals exhibit stronger resistance to rumor infection 
    compared to unaware individuals. (c) The infection process involves each infected neighbor 
    interacting with susceptible individuals at a certain rate to spread rumors. Higher local 
    infection density leads to shorter actual infection times for susceptible individuals. 
    Additionally, an unaware individual, upon infection, immediately becomes aware and exhibits an instantaneous transition to the aware state.
    (d) In settings with bursty communication and memory effects, information spreading exhibits memory effects and subdiffusive behavior, which are captured by the fractional-order derivative formulation.
      }
  \label{fig:fractional}
  \vspace{-2mm}
\end{figure*}

To realistically capture the coevolution of information and risk in settings with bursty communication and memory effects, we extend the classical awareness-epidemic paradigm to a fractional-order framework. Rumor propagation in such systems is often characterized by non-Markovian temporal dependence, where individuals' current decisions are influenced not only by the present network state but also by past interactions and previously received information. The classical integer-order SIS-UAU model is able to describe the basic coupling between awareness diffusion and risk transmission, but it cannot explicitly account for memory effects and subdiffusive spreading behavior. By introducing a Caputo fractional derivative of order $\alpha \in (0,1)$, the proposed model provides a more flexible and realistic description of rumor propagation while preserving the analytical structure needed for theoretical analysis and numerical simulation. We first outline the structure and dynamics of the underlying classical model, which forms the basis for our fractional-order extension.

The classical SIS-UAU model evolves on a duplex network with two layers, each containing $N$ nodes in a one-to-one correspondence. The connectivity patterns in both layers may differ. The upper layer represents the information layer, with the adjacency matrix $\mathbf{A}=\{\mathbf{a}_{ij}\}_{N\times N}$, while the lower layer represents the risk-spreading layer, with the adjacency matrix $\mathbf{B}=\{\mathbf{b}_{ij}\}_{N\times N}$.

The proposed model is built on a duplex network with one-to-one correspondence between the two layers, where the information layer governs awareness diffusion and the risk layer governs rumor transmission. Compared with the classical SIS-UAU model, the present framework incorporates the following key features, which collectively provide a tractable yet more realistic description of rumor propagation in multiplex networks.
\begin{itemize}
    \item \textbf{Duplex network structure.} The system evolves on two coupled layers with adjacency matrices $\mathbf{A}=\{a_{ij}\}$ (information layer) and $\mathbf{B}=\{b_{ij}\}$ (risk layer). Nodes in the two layers correspond one-to-one, but the connectivity patterns may differ, allowing independent characterization of the awareness-spreading and risk-transmission topologies.
    \item \textbf{Reduced state space with instantaneous post-infection awareness.} Each node possesses two binary attributes---susceptible or infected ($S/I$) on the risk layer, and unaware or aware ($U/A$) on the information layer---yielding four possible joint states $\{US, AS, UI, AI\}$ in principle. However, we adopt the limit $\kappa\to\infty$, meaning that an infected node becomes aware immediately upon exposure to risk. Consequently, the $UI$ state is transient and the effective dynamical system operates on the reduced space $\{US, AS, AI\}$. Although a finite $\kappa$ would be more realistic, it would re-introduce the $UI$ state as a genuine dynamical variable, enlarging the state space to four dimensions and requiring the full $4\times4$ generator matrix---substantially complicating the theoretical analysis. The $\kappa\to\infty$ limit is therefore adopted as a deliberate simplification that preserves analytical tractability while capturing the essential bidirectional coupling.
    \item \textbf{Continuous-time Poisson transition rules.} All state changes are governed by independent Poisson processes. A $US$ node becomes $AS$ at rate $\lambda$ per aware neighbor; both $US$ and $AS$ nodes become infected ($\to AI$) upon contact with infected neighbors, with an infection rate that depends on awareness; $AS$ nodes revert to $US$ at forgetting rate $\delta$; and $AI$ nodes recover to $AS$ at rate $\mu$. These elementary rules define a well-posed microscopic stochastic process for the classical integer-order setting.
    \item \textbf{Awareness-dependent infection susceptibility.} Awareness reduces an individual's vulnerability to rumor infection. Formally, the infection rate satisfies $\beta^A = \gamma\beta^U$ with $0<\gamma<1$, so that aware susceptible individuals ($AS$) are infected at a lower rate than unaware ones ($US$), capturing the protective effect of corrective information.
    \item \textbf{Mean-field closure.} To obtain a tractable system of dynamical equations, the stochastic neighbor-averaged quantities $m_i^A(t)$, $m_i^I(t)$, $\mathbf{X}_i^A(t)$, and $\mathbf{X}_i^I(t)$ are replaced by their population-level expectations $M_i^A(t)$ and $M_i^I(t)$ via a standard heterogeneous-network mean-field approximation. This closure neglects higher-order correlations but preserves the leading-order effect of network topology on the coevolutionary dynamics.
    \item \textbf{Fractional-order temporal memory.} The ordinary time derivative in the closed mean-field system is replaced by a Caputo fractional derivative ${}^C_0\mathcal{D}_t^\alpha$ of order $\alpha\in(0,1)$. This substitution introduces non-Markovian memory effects and subdiffusive dynamics at the macroscopic level, allowing the spreading process to depend on the entire history of the system rather than only on the instantaneous state. Importantly, the Caputo derivative is applied to the already-closed mean-field equations rather than being introduced through a modification of the underlying microscopic interaction rules; as a result, the present model studies how memory effects, once incorporated at the continuum level, modify the transient behavior, equilibrium conditions, and local stability properties of the coevolutionary process.
\end{itemize}
The main symbols used in the model are summarized in Table~\ref{tab:symbols}.

\begin{table}[htbp]
\centering
\renewcommand{\arraystretch}{1.05}
\setlength{\tabcolsep}{6pt}
\caption{Symbols and definitions of the extended model}
\label{tab:symbols}
\begin{tabular}{ll}
\toprule
\textbf{Symbol} & \textbf{Definition} \\
\midrule
$N$ & Total number of nodes in the network \\
$\mathbf{A}$ & Adjacency matrix of the information layer $(a_{ij})_{N\times N}$ \\
$\mathbf{B}$ & Adjacency matrix of the risk layer $(b_{ij})_{N\times N}$ \\
$p^{XY}_i(t)$ & Probability that node $i$ is in state $XY$ at time $t$ \\
$\lambda$ & Awareness spreading rate for the $US \to AS$ transition \\
$\kappa$ & Instantaneous awareness transition rate for $UI \to AI$ \\
$\delta$ & Forgetting rate for the $A \to U$ transition \\
$\mu$ & Recovery rate of an infected node to the susceptible state\\
$\beta_{i}^{A/U}$ & Risk transmission rate in aware/unaware susceptible state \\
$m_i^I(t)$ & Proportion of infected (risk-adopted) neighbors of node $i$ \\
$m_i^A(t)$ & Proportion of aware neighbors of node $i$ at time $t$ \\
$\mathbf{X}_i^I(t)$ & Indicator of the infected state for node $i$ at time $t$\\
$\mathbf{X}_i^A(t)$ & Indicator of the awareness state for node $i$ at time $t$\\
$M_i^I(t)$ & Mean-field approximation of $i$'s infected neighbors ratio\\
$M_i^A(t)$ & Mean-field approximation of $i$'s aware neighbors ratio\\
$\gamma$ & Relative transmission strength between $A$ and $U$ individuals \\
$\alpha$ & Order of the fractional derivative \\
\bottomrule
\end{tabular}
\end{table}

\begin{figure}[htbp]
    \centering
    \vspace{-1mm}
    \includegraphics[width=0.8\linewidth]{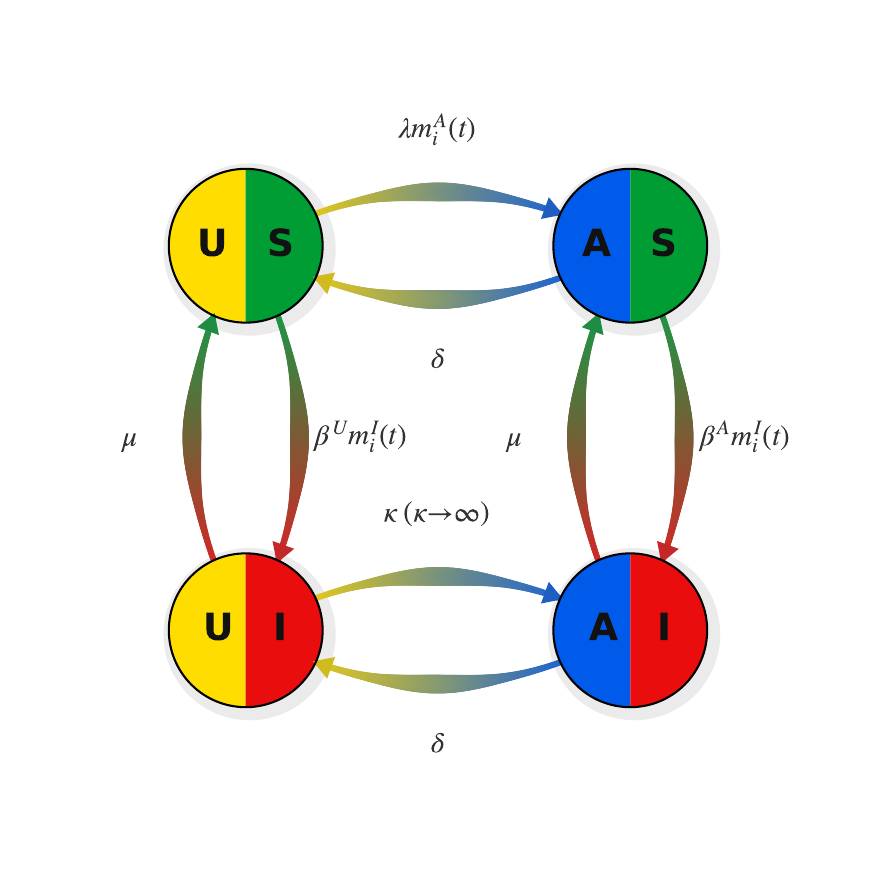}
    \vspace{-6mm}
      \caption{\textbf{Classical continuous-time SIS-UAU model.} This schematic diagram illustrates the structure of the classical SIS-UAU model on a two-layer multiplex network. Each node is represented by a circle vertically divided into two halves: the \emph{left} half corresponds to the information layer and the \emph{right} half to the risk-spreading layer. Yellow and blue indicate the \textbf{unaware} (\(U\)) and \textbf{aware} (\(A\)) states on the information layer, respectively, while green and red represent the \textbf{susceptible} (\(S\)) and \textbf{infected} (\(I\)) states on the risk layer. The four possible joint states are \(US\) (unaware–susceptible), \(AS\) (aware–susceptible), \(UI\) (unaware–infected), and \(AI\) (aware–infected). Arrows indicate the possible transitions between these states, together with their corresponding rates. Individuals can switch between unaware and aware states through awareness spreading (\(\lambda m_i^A(t)\)) and forgetting (\(\delta\)). Meanwhile, susceptible individuals can become infected through contact with infected neighbors, at rates \(\beta^U m_i^I(t)\) for unaware and \(\beta^A m_i^I(t)\) for aware individuals, and recover at rate \(\mu\). Additionally, infected unaware individuals immediately become aware, modeled by an instantaneous transition \(U I \to A I\) with \(\kappa \to \infty\). This figure highlights how awareness diffusion and risk spreading evolve on separate layers but are dynamically coupled through the information-induced modification of infection.}
    \label{fig:transition}
    \vspace{0mm}
\end{figure}

The possible transitions between the four joint states are comprehensively captured by the state-transition diagram shown in Fig.~\ref{fig:transition}. Such a diagram is a directed graph, where nodes represent discrete system states and directed edges indicate possible transitions between them, with edge weights corresponding to the transition rates. In this diagram, it can be observed that at any given time, an individual undergoes a change in either their awareness state or their risk state, but not both simultaneously. Crucially, for transitions of the same type, the specific state in the other layer can lead to different transition rates. For instance, a susceptible individual becomes $I$ at rate $\beta^U m_i^I(t)$ if unaware, but at a reduced rate $\beta^A m_i^I(t)$ if aware ($\beta^A = \gamma \beta^U, 0<\gamma<1$). This explicitly illustrates the dynamic coupling between the two processes: risk perception (awareness state) modifies the risk of infection, while the infection event itself ($UI \to AI$) forces an immediate update of the awareness state.

Based on the transition diagram, we construct the infinitesimal generator matrix:
\begin{small}
  \begin{align}
  \mathcal{Q}^+_i(t) &:=
  \begin{pmatrix}
  Q^+_{11} & \lambda m_i^A(t) & \beta^U m_i^I(t) & 0 \\
  \delta & Q^+_{22} & 0 & \beta^A m_i^I(t) \\
  \mu & 0 & -(\mu+\kappa) & \kappa \\
  0 & \mu & \delta & -(\mu+\delta)
  \end{pmatrix}, 
  \end{align}
\end{small}
where
\begin{small}
\begin{align}
  \begin{cases}
Q^+_{11} &= -(\lambda m_i^{A+}(t)+\beta^U m_i^{I+}(t)), \\
Q^+_{22} &= -(\delta+\beta^{A} m_i^{I+}(t)).
\end{cases}
\end{align}
\end{small}
In the infinitesimal generator matrix $\mathcal{Q}^+_i(t)$, the diagonal elements $Q^+_{11}$ and $Q^+_{22}$ are defined by the outgoing rates. The off-diagonal transitions are governed by the awareness spreading rate $\lambda$, the forgetting rate $\delta$, the recovery rate $\mu$, and the immediate awareness gain rate $\kappa$. The infection rates for unaware and aware susceptible individuals are $\beta^U$ and $\beta^A$ respectively, with $\beta^A = \gamma \beta^U$, where $0 < \gamma < 1$. These transitions are modulated by the proportion of aware neighbors $m_i^{A+}(t)$ and infected neighbors $m_i^{I+}(t)$.
In this derivation, we adopt the following notational convention: variables with a superscript plus sign (e.g., $\mathcal{P}^+$, $\mathcal{Q}^+_i$, $m_i^{I+}$) correspond to quantities defined on the full state space ${US, AS, UI, AI}$, which includes the transient $UI$ state. After taking the limit $\kappa\to\infty$, these reduce to their asterisk‑free counterparts ($\mathcal{P}$, $\mathcal{Q}_i$, $m_i^I$, etc.) defined on the effective state space ${US, AS, AI}$. This distinction keeps the two stages of the derivation clearly separated.

We denote the state probability distribution of node $i$ at time $t$ over the four possible states (including transient ones) as
\begin{small}
\begin{align}
  \mathcal{P}^+_i(t):=\begin{pmatrix}
    p_i^{US}(t) & p_i^{AS}(t) & p_i^{UI}(t) & p_i^{AI}(t)
  \end{pmatrix},
\end{align}
\end{small}
and represent the joint probability distribution of all nodes in matrix form as
\begin{align}
  \mathcal{P}^+(t):=\begin{pmatrix}
    \mathcal{P}^+_1(t) & \mathcal{P}^+_2(t) & \dots & \mathcal{P}^+_N(t)
  \end{pmatrix}^T.
\end{align}
In this way, we obtain the Kolmogorov forward equation:
\begin{align}
  \frac{\text d \mathcal P^+(t)}{\text d t}=\mathcal{P}^+(t)\bigoplus_{i=1}^N{\mathcal{Q}^+_i(t)}.
\label{eq:kolmogorov}
\end{align}
The Kolmogorov forward equation is adopted as the rigorous starting point of our model derivation: it encodes the microscopic Markov transition rules (summarized by the generator $\mathcal{Q}^+_i(t)$) into an exact evolution equation for the state probability distribution, thereby guaranteeing that the mean-field equations derived subsequently inherit the assumed propagation mechanism rather than being introduced phenomenologically. This distinguishes the present framework from purely heuristic rate-equation rumor models. We also note that the Caputo fractional derivative, as a distributed-memory operator with a power-law kernel, inherently accounts for history-dependent and delayed effects in a continuous sense, whereas explicit discrete time-delay terms \cite{zhu2020delay, yuan2023stability} and saturating functional responses \cite{hoang2020dynamics}---while valuable in other rumor models---would break the exact correspondence between the microscopic Poisson rules and the mean-field formulation, and are therefore left as extensions of the present framework.
Here, the symbol $\bigoplus_{i=1}^N$ denotes the direct sum of matrices, constructing a block-diagonal matrix $\text{diag}(\mathcal{Q}^+_1(t), \dots, \mathcal{Q}^+_N(t))$, where each block $\mathcal{Q}^+_i(t)$ governs the local state transitions of node $i$, and off-diagonal blocks are zero.
This Kolmogorov forward equation synthesizes the local transition rules into the core governing dynamics for the system's state probability distribution $\mathcal{P}^+(t)$. The block-diagonal structure of $\bigoplus_{i=1}^N \mathcal{Q}^+_i(t)$ signifies that the global dynamics emerges from the aggregate of all nodes' independent local transitions, with the network coupling implicitly encoded within the definitions of each $\mathcal{Q}^+_i(t)$.

However, since we employ indicator variables in the model, a mean-field approximation is required to close the system. Specifically, we approximate the indicator variables as follows:
\begin{small}
  \begin{align}
    \begin{cases}
      m^{A+}_i(t)\approx M_i^{A+}(t)=\sum_{j=1}^{N}\mathbf{a}_{ij}p_j^{A}(t)=\mathbf{e}_i^T\mathbf{A}\mathcal{P}^+\begin{pmatrix}0,1,0,1\end{pmatrix}^T,\\[0.6em]
      m^{I+}_i(t)\approx M_i^{I+}(t)=\sum_{j=1}^{N}\mathbf{b}_{ij}p_j^{I}(t)=\mathbf{e}_i^T\mathbf{B}\mathcal{P}^+\begin{pmatrix}0,0,1,1\end{pmatrix}^T.
    \end{cases}
    \label{eq:mean_field}
  \end{align}
\end{small}
\noindent where $\mathbf{e}_i$ is the $i$-th standard basis vector in $\mathbb{R}^N$, and $p_j^{A}(t)=p_j^{AS}(t)+p_j^{AI}(t)$ and $p_j^{I}(t)=p_j^{AI}(t)$ denote the marginal probabilities that node $j$ is aware and infected at time $t$, respectively.
Substituting the mean-field approximations (\ref{eq:mean_field}) into the Kolmogorov forward equation (\ref{eq:kolmogorov}) and applying the transient-state assumption, we obtain the following relationship:
\begin{align}
    \frac{\text dp^{UI}_i(t)}{\text dt}=p^{UI}_i(t)=0.
\end{align}
By substituting and eliminating the $UI$ state, we yield the closed $N$-entangled mean-field equations for the continuous-time UAU-SIS model:
\begin{align}
  \frac{\text d}{\text d t}\mathcal{P}(t)=\mathcal{P}(t)\bigoplus_{i=1}^N\mathcal{Q}_i(t),
  \label{eq:classical}
\end{align}
where $\mathcal{P}_i(t)$ denotes the state probability distribution of node $i$ after eliminating the $UI$ state:
\begin{align}
  \mathcal{P}_i(t):=\begin{pmatrix}
    p_i^{US}(t) & p_i^{AS}(t) & p_i^{AI}(t)
  \end{pmatrix}.
\end{align}
The matrix $\mathcal{P}(t)$ is then defined as
\begin{align}
  \mathcal{P}(t):=\begin{pmatrix}
    \mathcal{P}_1(t) & \mathcal{P}_2(t) & \dots & \mathcal{P}_N(t)
  \end{pmatrix}^T,
\end{align}
and $\mathcal{Q}_i(t)$ represents the infinitesimal generator matrix after eliminating the $UI$ state:
\begin{align}
\mathcal{Q}_i(t) &:=
\begin{pmatrix}
Q_{11}(t) & \lambda M_i^A(t) & \beta^U M_i^I(t) \\
\delta & -(\delta + \beta^A M_i^I(t)) & \beta^A M_i^I(t) \\
0 & \mu & -\mu
\end{pmatrix} ,
\end{align}
where
\begin{align}
  \begin{cases}
Q_{11}(t) &= -(\lambda M_i^A(t) + \beta^U M_i^I(t)), \\[2pt]
  M_i^A(t) &= \mathbf{e}_i^T\mathbf{A}\mathcal{P}\begin{pmatrix}0,1,1\end{pmatrix}^T, \\
  M_i^I(t) &= \mathbf{e}_i^T\mathbf{B}\mathcal{P}\begin{pmatrix}0,0,1\end{pmatrix}^T.
\end{cases}
\end{align}
Eq.~\eqref{eq:classical} describes the evolution of state probabilities in the classical Markovian setting, where changes depend only on the current state.

The above expressions constitute the standard mean-field closure. The fractional-order extension is built at the macroscopic level—the Caputo derivative replaces the ordinary derivative in the already-closed equations, rather than modifying the underlying microscopic rules. Hence no agent-based stochastic process exists from which the fractional equations can be recovered, making direct Monte Carlo simulation infeasible. The reason is that a mechanism-based (e.g., Gillespie-type) Monte Carlo simulation requires a complete node-level event prescription: at every instant each node must possess well-defined transition rates from which the next event can be sampled, exactly as the Poisson rules above provide for the integer-order model. The Caputo operator, by contrast, is a macroscopic, history-weighted operator applied to the aggregate equations; it does not determine a unique per-node waiting-time or residual-time law, so any attempt to construct a node-level simulator with memory would require additional modeling assumptions and would simulate a different model rather than the proposed one. The model instead investigates the macroscopic consequences of memory effects on transient dynamics, equilibrium, and stability, and the consistency of the fractional extension is assessed through the convergence test as $\alpha\to1$ (Section~12.1 of the Supplementary Material).

The infinitesimal generator matrix \(\mathcal{Q}_i(t)\) can be decomposed into three distinct components, offering a clear interpretation of the modeled dynamics:
\begin{align}
\mathcal{Q}_i(t) =&
\underbrace{\begin{pmatrix}
0 & 0 & 0 \\
\delta & -\delta & 0 \\
0 & \mu & -\mu
\end{pmatrix}}_{\text{Intrinsic transitions}} 
+ M_i^A(t) \underbrace{\begin{pmatrix}
-\lambda & \lambda & 0 \\
0 & 0 & 0 \\
0 & 0 & 0
\end{pmatrix}}_{\text{Awareness spreading}} \notag\\
&+ M_i^I(t) \underbrace{\begin{pmatrix}
-\beta^U & 0 & \beta^U \\
0 & -\beta^A & \beta^A \\
0 & 0 & 0
\end{pmatrix}}_{\text{Risk transmission}} .
\label{eq:Qi_decomposition}
\end{align}
In Eq.~\eqref{eq:Qi_decomposition}, the three components respectively govern intrinsic transitions (forgetting and recovery), awareness spreading, and risk transmission---the latter two weighted by the local neighbor densities \(M_i^A(t)\) and \(M_i^I(t)\). This additive structure captures the core coupling: awareness inhibits risk (\(\beta^A < \beta^U\)), while infection forces an immediate awareness update, both mediated through the network structure.

To incorporate the memory effects and subdiffusive behavior characteristic of information spreading in settings with bursty communication and memory effects, we generalize this framework by replacing the ordinary time derivative with a Caputo fractional derivative of order $\alpha \in (0,1)$. The Caputo fractional derivative is defined as
\begin{align}
{}^C_0\mathcal{D}_t^\alpha q(t)=\frac{1}{\Gamma(1-\alpha)}\int_0^t \frac{q'(\tau)}{(t-\tau)^\alpha}\, \mathrm{d}\tau,
\end{align}
where $\Gamma(\cdot)$ denotes the Gamma function \cite{podlubny1998fractional, kilbas2006theory}.

By replacing the time derivative in the classical SIS-UAU model with the Caputo fractional derivative, we obtain the fractional-order SIS-UAU model:

\begin{align}
  {}^C_0\mathcal{D}^\alpha_t\mathcal{P}(t)=\mathcal{P}(t)\bigoplus_{i=1}^N\mathcal{Q}_i(t),
  \label{eq:main}
\end{align}
where $\alpha \in (0,1)$; the classical model is recovered at $\alpha = 1$. The component-wise form of Eq.~\eqref{eq:main} is:

\begin{align}
  \begin{cases}
    {}^C_0\mathcal{D}^\alpha_tp_i^{US}(t)=&-\left(\lambda M_i^A(t)+\beta^UM_i^I(t)\right)p_i^{US}(t)+\delta p_i^{AS}(t),\\
    {}^C_0\mathcal{D}^\alpha_tp_i^{AS}(t)=&\left(\lambda M_i^A(t)\right)p_i^{US}(t)-\left(\delta+\beta^AM_i^I(t)\right)p_i^{AS}(t)\\
    &+\mu p_i^{AI}(t),\\
    {}^C_0\mathcal{D}^\alpha_tp_i^{AI}(t)=&\left(\beta^UM_i^I(t)\right)p_i^{US}(t)+\left(\beta^AM_i^I(t)\right)p_i^{AS}(t)\\&
    -\mu p_i^{AI}(t).
  \end{cases}
  \label{eq:main_component}
\end{align}

Note that $M_i^I(t)$ in Eq.~\eqref{eq:main_component} remains an instantaneous function of the current state probabilities; however, because $\mathcal P(t)$ evolves under the Caputo fractional derivative, $M_i^I(t)$ inherits history dependence indirectly through the fractional-order dynamics.

By the definition of the Caputo fractional derivative, Eq.~\eqref{eq:main} implicitly requires that
\begin{align}
\mathcal{P} \in C^{\lceil \alpha \rceil}\bigl(\mathbb{R}+, \mathbb{R}^{N\times 3}\bigr)
= C^1\bigl(\mathbb{R}+, \mathbb{R}^{N\times 3}\bigr), \quad \text{a.e.}
\end{align}
For convenience, in the subsequent analysis we directly assume that 
\begin{align}
  \mathcal{P} \in C^1\bigl(\mathbb{R}_+, \mathbb{R}^{N\times 3}\bigr).
\end{align}
This regularity assumption formalizes the intuitive notion of smooth evolution and ensures the well-posedness of the subsequent existence, uniqueness, and stability analyses.

\subsection{Theoretical Analysis of the Fractional-Order SIS-UAU Model}\label{sec:theoretical_analysis}

In this subsection, we analyze the theoretical properties of the fractional-order SIS-UAU model (\ref{eq:main}), including the existence and uniqueness of solutions, positivity and boundedness, the basic reproduction number, as well as the local asymptotic stability of both the rumor-free equilibrium and the rumor-endemic equilibrium. 
We now present the existence and uniqueness of solutions to system~\eqref{eq:main} under given initial conditions. In what follows, $\Omega$ denotes the feasible state space
\begin{align}
  \Omega :=\Big\{\mathcal P \in \mathbb{R}^{N\times 3} \ \big|\ p_i^{US},\,p_i^{AS},\,p_i^{AI} \ge 0,\\\notag
  \ \ p_i^{US}+p_i^{AS}+p_i^{AI}=1,\ \forall i\in\mathbb Z[1,N]\Big\},
\end{align}
and $L$ denotes the Lipschitz constant of the vector field $F(\mathcal P):=\mathcal P\bigoplus_{i=1}^N\mathcal Q_i(\mathcal P)$ on $\Omega$.
\begin{theorem}\label{theorem:existence_uniqueness}
  (Existence and uniqueness of the solution)
Consider the fractional-order SIS-UAU model defined by Eq.~\eqref{eq:main} with the initial condition $\mathcal{P}(0) = \mathcal{P}_0 \in \Omega$. Then, there exists a unique solution $\mathcal{P}(t)$ defined on $[0, T]$ for any $T \in \left[0,\left(\frac{\Gamma(1+\alpha)}{L}\right)^{1/\alpha}\right)$.
\end{theorem}
\noindent \textit{Proof.}
The detailed proof is provided in the Supplementary Material, Section~3.

Next, we present that the first-order continuous solution of system~\eqref{eq:main} remains positive and bounded under the initial conditions.
\begin{theorem}\label{theorem:positivity_boundedness}
  (Positivity and boundedness of the solution)
Consider the fractional-order SIS-UAU model defined by Eq.~\eqref{eq:main} with the initial condition $\mathcal{P}(0) = \mathcal{P}_0 \in \Omega$. Then, the solution $\mathcal{P}(t)$ satisfies $p_i^{US}(t), p_i^{AS}(t), p_i^{AI}(t) \geq 0$ and $p_i^{US}(t) + p_i^{AS}(t) + p_i^{AI}(t) = 1$ for all $t \geq 0$ and for all $i\in \mathbb{Z}[1,N]$.
\end{theorem}
\noindent \textit{Proof.}
The detailed proof is provided in the Supplementary Material, Section~4.

We then analyze the equilibria and their stability for system~\eqref{eq:main}. Therefore, we decompose the state matrix $\mathcal P$ column-wise as
\begin{align}
  \mathcal P = \begin{pmatrix} \mathcal P^{US}, \ \mathcal P^{AS}, \ \mathcal P^{AI} \end{pmatrix}.
\end{align}

We first discuss the rumor-free equilibrium (RFE) of the system. In this case, all nodes are in the susceptible state with respect to the risk-spreading process, i.e., $p^{AI}_i = 0$ and $M^I_i = 0$ for all $i \in \mathbb{Z}[1,N]$. Under this condition, the system reduces to a single-layer awareness spreading model. Its equilibrium
\begin{align}
  \mathcal P^{*} = \begin{pmatrix} 1-\mathcal P^{AS*}, \ \mathcal P^{AS*}, \ 0 \end{pmatrix}
\end{align}
satisfies
\begin{align}
  \frac{\delta}{\lambda}\frac{\mathcal P^{AS*}}{1 - \mathcal P^{AS*}} = \mathbf{A} \mathcal P^{AS*},
\end{align}
where $\mathcal P^{AS*}$ denotes the distribution of the awareness states at the rumor-free equilibrium. Moreover,
\begin{align}
  \frac{\mathcal P^{AS*}}{1 - \mathcal P^{AS*}}=
  \begin{pmatrix}
  \frac{p_1^{AS*}}{p_1^{US*}}, \ \frac{p_2^{AS*}}{p_2^{US*}}, \ \ldots, \ \frac{p_N^{AS*}}{p_N^{US*}}
  \end{pmatrix}^T.
\end{align}
This equation indicates that, at equilibrium, the relationship between each node’s awareness acquisition rate and its forgetting rate is influenced by the underlying network topology.

We define the basic reproduction number as
\begin{align}
  \label{eq:R0_def}
  \mathcal R_0 := \frac{1}{\mu}\rho\bigl(\mathrm{diag}(p^{US*}_i\beta^U + p^{AS*}_i\beta^A)\mathbf B\bigr),
\end{align}
where $\rho(\cdot)$ denotes the spectral radius of a matrix and $\mathrm{diag}(\cdot)$ is the diagonal matrix operator \cite{pastor2015epidemic}.

For the stability results, we have the following theorem. The stability analysis follows the standard framework for Caputo fractional-order systems: equilibria are determined by the algebraic steady-state conditions (since the Caputo derivative of a constant vanishes), and local asymptotic stability is characterized by the Matignon-type eigenvalue criterion, as established in the fractional stability theory \cite{li2010stability, deng2007stability, brandibur2021stability}.
\begin{theorem}\label{theorem:stability}
  (Local asymptotic stability of the rumor-free equilibrium)
Consider the fractional-order SIS-UAU model defined by Eq.~\eqref{eq:main} with the initial condition $\mathcal{P}(0)=\mathcal{P}_0\in\Omega$. 
The rumor-free equilibrium (RFE)
$\mathcal{P}^{*} = \begin{pmatrix} \mathcal P^{US*}, & \mathcal P^{AS*}, & \mathbf 0 \end{pmatrix}$
is locally asymptotically stable if every eigenvalue $\xi$ of the matrix
\begin{small}
\begin{align}
  \mathcal J \;:=\; \operatorname{diag}(p^{US*}_i\beta^U + p^{AS*}_i\beta^A)\,\mathbf B \;-\; \mu\,\mathbf I_N
\end{align}
satisfies
\begin{align}
  \big|\arg(\xi)\big| \;>\; \frac{\alpha\pi}{2}.
\end{align}
\end{small}
\end{theorem}

\noindent \textit{Proof.}
The detailed proof is provided in the Supplementary Material, Section~5.

We now derive a convenient corollary from Theorem~\ref{theorem:stability} which gives a simple, practically verifiable sufficient condition for the local asymptotic stability of the RFE.

\begin{corollary}\label{corollary:R0_less_1}
If the basic reproduction number defined in \eqref{eq:R0_def} satisfies
\begin{align}\label{eq:R0_lt_1}
  \mathcal R_0 \;<\; 1,
\end{align}
then the rumor-free equilibrium $\mathcal P^*$ is locally asymptotically stable for the fractional-order system \eqref{eq:main} (for every $\alpha\in(0,1]$).
\end{corollary}

\noindent \textit{Proof.}
The detailed proof is provided in the Supplementary Material, Section~6.

We now turn our attention to the existence of a nontrivial equilibrium where the rumor persists in the network.
Let the rumor-endemic equilibrium be denoted by
\begin{align}\label{eq:endemic_equilibrium_def}
  \mathcal P^{**} = 
  \begin{pmatrix}
    \mathcal P^{US**}, &
    \mathcal P^{AS**}, &
    \mathcal P^{AI**}
  \end{pmatrix},
\end{align}
where $\mathcal P^{AI**}\!>\!0$ represents the nonzero steady-state infection probabilities of all nodes.

At equilibrium, ${}^C_0\mathcal D_t^{\alpha}\mathcal P(t)=0$, so that each node satisfies
\begin{align}\label{eq:endemic_equilibrium_equation}
  \mathcal P_i^{**}\,\mathcal Q_i(\mathcal P^{**}) = 0,
  \qquad i\in \mathbb{Z}[1,N].
\end{align}
By expanding \eqref{eq:endemic_equilibrium_equation} for the infection, component we obtain
\begin{align}\label{eq:endemic_balance_I}
  \mu\,p_i^{AI**} 
  = 
  \beta^U M_i^{I**} p_i^{US**}
  + 
  \beta^A M_i^{I**} p_i^{AS**},
\end{align}
where $M_i^{I**}=\sum_{j}b_{ij}p_j^{AI**}$ is the total effective infectious pressure received by node $i$ at equilibrium.
\begin{remark}
The equilibrium points of the fractional-order system are independent of the fractional order $\alpha$. Indeed, at equilibrium, the Caputo derivative vanishes and the steady-state solutions are determined entirely by the algebraic equations obtained from
${}^C_0\mathcal D_t^{\alpha}\mathcal P(t)=0$.
Therefore, changing $\alpha$ affects only the transient dynamics and local stability properties, rather than the equilibrium locations themselves.
\end{remark}
Combining the three steady-state relations implied by \eqref{eq:endemic_equilibrium_equation} gives the following nonlinear system:
\begin{small}
\begin{align}\label{eq:endemic_system}
\left\{
\begin{aligned}
  &\big(\lambda M_i^{A**}+\beta^U M_i^{I**}\big)p_i^{US**}
  = 
  \delta p_i^{AS**}, \\[4pt]
  &\beta^U M_i^{I**}p_i^{US**} + \beta^A M_i^{I**}p_i^{AS**}
  =
  \mu p_i^{AI**}, \\[4pt]
  &p_i^{US**}+p_i^{AS**}+p_i^{AI**}=1.
\end{aligned}
\right.
\end{align}
\end{small}

The first equation in \eqref{eq:endemic_system} expresses the awareness balance at equilibrium, while the second one describes the rumor-spreading balance. Eliminating $p_i^{AS**}$ and $p_i^{US**}$ in favor of $p_i^{AI**}$, we get the self-consistent relation
\begin{align}\label{eq:self_consistent}
p_i^{AI**} := \frac{M_i^{I**}\big(\beta^U\delta+\beta^A(\lambda M_i^{A**}+\beta^U M_i^{I**})\big)}
{\splitfrac{\mu\big(\delta+\lambda M_i^{A**}+\beta^U M_i^{I**}\big)}
           {+ M_i^{I**}\big(\beta^U\delta+\beta^A(\lambda M_i^{A**}+\beta^U M_i^{I**})\big)}}.
\end{align}

To motivate the threshold, we linearize \eqref{eq:self_consistent} around the rumor-free equilibrium, where $p_i^{AI**}\approx 0$ and hence $M_i^{I**}\approx \sum_j b_{ij}p_j^{AI**}$. This yields the approximate relation
\begin{align}\label{eq:linearized_endemic}
  p_i^{AI**}
  \;\approx\;
  \frac{1}{\mu}
  \big(p_i^{US*}\beta^U + p_i^{AS*}\beta^A\big)
  \sum_{j=1}^N b_{ij}p_j^{AI**}.
\end{align}

In vector form, this reads
\begin{align}\label{eq:vector_linearized_endemic}
  \mathbf p^{AI**}
  \;\approx\;
  \frac{1}{\mu}\,
  \operatorname{diag}(p^{US*}_i\beta^U + p^{AS*}_i\beta^A)\,
  \mathbf B\,\mathbf p^{AI**}
  \;=\;
  \mathcal K\,\mathbf p^{AI**},
\end{align}
where
\begin{align}\label{eq:K_def}
  \mathcal K
  :=
  \frac{1}{\mu}\,
  \operatorname{diag}(p^{US*}_i\beta^U + p^{AS*}_i\beta^A)\,\mathbf B
\end{align}
is the corresponding next-generation matrix.

Relation \eqref{eq:vector_linearized_endemic} is only the first-order approximation of \eqref{eq:self_consistent} at the rumor-free equilibrium; it identifies $\mathcal R_0$ as the threshold candidate. The rigorous analysis of the nonlinear relation \eqref{eq:self_consistent} is given in the Supplementary Material, Section~7. Hence we obtain the following existence result.

\begin{theorem}\label{theorem:endemic_existence}
(Existence and uniqueness of the rumor-endemic equilibrium)  
Assume that $\mathbf A$ and $\mathbf B$ are nonnegative zero-diagonal irreducible matrices and that $0<\beta^A<\beta^U$. Then, for the fractional-order UAU-SIS system \eqref{eq:main}, there exists a unique rumor-endemic equilibrium $\mathcal P^{**}$ with $\mathcal P^{AI**}\!>\!0$ if and only if the basic reproduction number satisfies
\begin{align}\label{eq:R0_gt_1}
  \mathcal R_0 = \rho(\mathcal K) > 1.
\end{align}
Conversely, if $\mathcal R_0 \le 1$, the only equilibrium is the rumor-free equilibrium $\mathcal P^*$.
\end{theorem}
\noindent \textit{Proof.}
The detailed proof is provided in the Supplementary Material, Section~7.

\begin{remark}
Theorem~\ref{theorem:endemic_existence} characterizes the threshold behavior of the system as a transcritical bifurcation: the rumor-free equilibrium loses stability and a unique rumor-endemic equilibrium emerges exactly at $\mathcal{R}_0=1$. The bifurcation diagram reported in Section~12.2 of the Supplementary Material is the numerical trace of this analytically predicted threshold.
\end{remark}

\begin{remark}
The basic reproduction number $\mathcal R_0$ is obtained from the equilibrium condition and the associated next-generation matrix, and therefore its expression does not explicitly depend on the fractional order $\alpha$. In the proposed fractional-order UAU-SIS system, the role of $\alpha$ is reflected in the transient evolution and the local stability properties of the equilibria, rather than in the algebraic threshold condition for the existence of the rumor-endemic equilibrium. In particular, the fractional order may change the convergence rate and stability region through the Matignon-type criterion, while the threshold condition $\mathcal R_0>1$ remains the same.
\end{remark}

The local behavior of the fractional-order system~\eqref{eq:main} near the rumor-endemic equilibrium
$\mathcal P^{**}=(\mathcal P^{US**},\mathcal P^{AS**},\mathcal P^{AI**})$ is governed by the Jacobian matrix
$\mathcal J^{*}$ obtained by linearizing the componentwise dynamics~\eqref{eq:main_component} at $\mathcal P^{**}$.
As detailed in Section~1.5 of the Supplementary Material, $\mathcal J^{*}$ admits the compact block representation
\begin{align}\label{eq:Jstar_compact_text}
  \mathcal J^{*}
  = \sum_{m=1}^{3}\sum_{n=1}^{3} E_{mn}\otimes B_{mn},
\end{align}
where $E_{mn}$ are the canonical basis matrices in $\mathbb R^{3\times3}$ and the $N\times N$ blocks $B_{mn}$ are
expressed in terms of the fundamental matrices $\{D_u, D_a, D_{M^A}, D_{M^I}, \mathbf A, \mathbf B, \mathbf I_N\}$.
The explicit expressions of the blocks are collected in Section~1.5 of the Supplementary Material, where the
local stability of $\mathcal P^{**}$ is established.

\begin{theorem}\label{theorem:endemic_stability_compact}
  (Local asymptotic stability of the rumor-endemic equilibrium)
Let $\mathcal P^{**}=(\mathcal P^{US**},\mathcal P^{AS**},\mathcal P^{AI**})$ be a rumor-endemic equilibrium of system~\eqref{eq:main} with $\mathcal P^{AI**}\ge0$. 
If all eigenvalues $\xi$ of the Jacobian matrix $\mathcal J^{*}$ defined in~\eqref{eq:Jstar_compact_text} satisfy
\begin{align}\label{eq:matignon_endemic_compact}
  \big|\arg(\xi)\big| > \frac{\alpha\pi}{2},
\end{align}
then the rumor-endemic equilibrium $\mathcal P^{**}$ is locally asymptotically stable.
\end{theorem}

\noindent \textit{Proof.}
The detailed proof is provided in the Supplementary Material, Section~8.

\begin{remark}
  The fractional order $\alpha$ affects the local stability boundary of the system through the Matignon-type condition. Specifically, an equilibrium is locally asymptotically stable only if all eigenvalues of the corresponding Jacobian matrix satisfy $|\arg(\xi)|>\alpha\pi/2$. Therefore, the stability region is a sector in the complex plane whose boundary depends explicitly on $\alpha$. As $\alpha$ increases, the threshold angle $\alpha\pi/2$ becomes larger, and the stability region becomes more restrictive. In contrast, the equilibrium locations themselves are determined by the algebraic steady-state equations and do not depend on $\alpha$. Hence, the fractional order influences the stability boundary and convergence behavior, but not the existence or position of the equilibria.
\end{remark}

\begin{remark}
Regarding global asymptotic stability, we establish the following rigorous result for the rumor-free equilibrium under sufficient conditions on both layers of the multiplex system.
\end{remark}

\begin{theorem}\label{theorem:global_RFE}
(Global asymptotic stability of the rumor-free equilibrium)
Let $\mathbf A$ and $\mathbf B$ be nonnegative, zero-diagonal, irreducible adjacency matrices and assume $0<\beta^A<\beta^U$. Define $\bar{\mathcal R}_0 := \beta^U \rho(\mathbf B)/\mu$, where $\rho(\cdot)$ denotes the spectral radius. If
\begin{align}
  \bar{\mathcal R}_0 < 1
  \qquad\text{and}\qquad
  \lambda\rho(\mathbf A) < \delta,
\end{align}
then for every $\alpha \in (0,1)$ and every initial condition $\mathcal P(0) \in \Omega$, the solution of system~\eqref{eq:main} converges to the rumor-free equilibrium $\mathcal P^{*} = (\mathbf 1, \mathbf 0, \mathbf 0)$; i.e., $\mathcal P^{*}$ is globally asymptotically stable.
\end{theorem}

\noindent \textit{Proof.}
The detailed proof is provided in the Supplementary Material, Section~11.2.

\begin{remark}
Four points concerning the scope of Theorem~\ref{theorem:global_RFE} are worth emphasizing. (i) The condition $\bar{\mathcal R}_0<1$ alone already guarantees rumor extinction, i.e., $\mathcal P^{AI}(t)\to \mathbf 0$ for every initial condition; the additional condition $\lambda\rho(\mathbf A)<\delta$ makes the awareness layer itself subcritical, so that the awareness density also decays and the full state converges to $\mathcal P^{*}=(\mathbf 1,\mathbf 0,\mathbf 0)$. (ii) When the awareness layer is supercritical ($\lambda\rho(\mathbf A)>\delta$), the rumor-free equilibrium carries a nontrivial awareness distribution $\mathcal P^{AS*}>0$. A rigorous proof that $\mathcal P^{AS}(t)\to \mathcal P^{AS*}$ after the rumor has died out would require an asymptotic-autonomy (limit-system) principle for Caputo fractional systems, which is an open technical problem because the Caputo flow lacks the semigroup property on which the classical integer-order argument \cite{lajmanovich1976deterministic} relies; we verify this convergence numerically in the Supplementary Material (Section~11.4). (iii) Since $\mathrm{diag}(p^{US*}_i\beta^U + p^{AS*}_i\beta^A) \le \beta^U\mathbf I_N$ componentwise, $\bar{\mathcal R}_0<1$ implies $\mathcal R_0<1$, so the theorem is a sufficient---not necessary---condition; on the BA network of Section~III-C, $\bar{\mathcal R}_0=1$ at $\beta^U\approx0.126$ whereas $\mathcal R_0=1$ at $\beta^U\approx0.289$. The sharp threshold $\mathcal R_0<1$ for the rumor-free equilibrium remains an open problem; the single-layer fractional network SIS analogue in the disease-free case is established in \cite{liu2019fractional}. (iv) For the rumor-endemic regime $\mathcal R_0>1$, Theorem~\ref{theorem:endemic_global_stability} below establishes global asymptotic stability of the rumor-endemic equilibrium under a verifiable quadratic dissipativity condition of Lajmanovich--Yorke type; for the BA network used in this paper, evaluation of the ratio of the left-hand side of \eqref{eq:dissipativity_main} to $\sum_i(\Delta x_i^2+\Delta y_i^2)$ over $1.4\times10^5$ sampled feasible states, combined with local maximization of the ratio, yields a maximum value of approximately $-0.19$ with $d_i=e_i=1$ (attained on the boundary face $p^{AI}\equiv0$), providing numerical evidence in support of the dissipativity condition. Since finite sampling cannot verify the inequality over the entire continuous feasible region, Theorem~7 remains a conditional criterion whose applicability to a given network is to be established analytically or by exhaustive verification. The multi-initial-condition numerical experiment reported in the Supplementary Material provides direct evidence of the global convergence behavior discussed here.
\end{remark}

\begin{theorem}\label{theorem:endemic_global_stability}
(Global asymptotic stability of the rumor-endemic equilibrium)
Let $\mathbf A$ and $\mathbf B$ be nonnegative, zero-diagonal, irreducible, symmetric adjacency matrices, let $0<\beta^A<\beta^U$ and $\mathcal R_0>1$, and let $\mathcal P^{**}=(p^{US**},p^{AS**},p^{AI**})$ be the unique rumor-endemic equilibrium of system~\eqref{eq:main} with all components positive. For $\mathcal P\in\Omega$ write $\Delta x_i=p_i^{AI}-p_i^{AI**}$ for the deviation of the infected fraction and $\Delta y_i=(p_i^{AS}+p_i^{AI})-(p_i^{AS**}+p_i^{AI**})$ for the deviation of the total aware fraction, and define
\begin{align}\label{eq:GH_def}
  G_i(\mathcal P)&=M_i^I\big(\beta^U p_i^{US}+\beta^A p_i^{AS}\big)-\mu p_i^{AI},\\
  H_i(\mathcal P)&=\lambda M_i^A p_i^{US}-\delta p_i^{AS}+\beta^U M_i^I p_i^{US},
\end{align}
where $H_i$ is the right-hand side of the equation satisfied by the total aware fraction $p_i^{AS}+p_i^{AI}$, obtained by adding the second and third equations of \eqref{eq:main_component}. If there exist positive numbers $d_i$, $e_i$ and a constant $c>0$ such that, for every $\mathcal P\in\Omega$,
\begin{align}\label{eq:dissipativity_main}
  2\sum_{i=1}^{N}d_i\,\Delta x_i\,G_i(\mathcal P)
  +2\sum_{i=1}^{N}e_i\,\Delta y_i\,H_i(\mathcal P)
  \le -c\Big(\sum_{i=1}^{N}\Delta x_i^2+\sum_{i=1}^{N}\Delta y_i^2\Big),
\end{align}
then, for every $\alpha\in(0,1)$ and every initial condition in $\Omega$, the solution of system~\eqref{eq:main} converges to $\mathcal P^{**}$; i.e., $\mathcal P^{**}$ is globally asymptotically stable.
\end{theorem}

\noindent \textit{Proof.}
The detailed proof is provided in the Supplementary Material, Section~11.3.

\section{Numerical simulations}\label{sec:simulation}
In this section, we perform a series of numerical simulations to validate the analytical results and illustrate the dynamic properties of the proposed fractional-order UAU-SIS model. Four groups of simulations are conducted to systematically examine: (i)~the consistency between the integer-order mean-field model and stochastic simulations, (ii)~the influence of the fractional derivative order~$\alpha$ on rumor propagation dynamics and steady states, (iii)~the impact of the consciousness layer parameters and topology on rumor spreading, and (iv)~the sensitivity of the coupled information-rumor dynamics to network topology; the last two groups are reported in Section~12 of the Supplementary Material and summarized at the end of this section. These simulation groups play complementary roles: the first one provides a benchmark for the classical integer-order setting, the second one highlights the dynamical effect of the fractional order, and the last two investigate how awareness mechanisms and multilayer structure shape the propagation outcomes. Before presenting these simulations, we first describe the numerical methods used for solving both integer-order and fractional-order systems. To further clarify the reliability of the numerical results, we note that the proposed fractional-order model is supported by rigorous theoretical analysis in Section~II, where the existence, uniqueness, positivity, boundedness, and local stability of the system are established. In Section~III, the numerical simulations are performed using the corresponding numerical solvers to illustrate the dynamical behavior predicted by the model and to compare it with the analytical conclusions. Together, these theoretical and numerical components provide evidence for the accuracy, convergence behavior, and validity of the proposed fractional-order UAU-SIS framework.

\subsection{Numerical solvers}\label{sec:numerical_methods}

For the integer-order system ($\alpha=1$), two simulation approaches are employed:
\begin{enumerate}
  \item The \emph{Gillespie stochastic simulation algorithm} (SSA), which provides an event-driven realization of the microscopic stochastic process, is used to capture the statistical evolution of the coupled awareness–infection dynamics on networks.
  \item The \emph{individual-based mean-field approximation} (IBMFA), which directly integrates the deterministic ODE system corresponding to Eq.~\eqref{eq:main}, thereby describing the average behavior of large populations.
\end{enumerate}
The integer-order model can therefore be interpreted as two mutually coupled SIS-type processes on a multiplex network, and its trajectories from IBMFA are expected to match the ensemble-averaged results of multiple Gillespie runs.

For the fractional-order system ($0<\alpha<1$), we numerically integrate the Caputo FODEs using a \emph{Block-Adaptive FFT (Fast Fourier Transform) Solver}, which accelerates the Caputo convolution via FFT and adaptively adjusts the block size. The complete pseudocode of the solver, the L1 discretization scheme, the predictor–corrector damping strategy, and the dynamic block-size adaptation rules are provided in the Supplementary Material. The systems simulated below are the mean-field image of the Kolmogorov forward equation (Eq.~\eqref{eq:kolmogorov}): the transfer generator $\mathcal{Q}_i(t)$ entering the simulated equations, including the awareness-dependent infection rates and the immediate awareness update upon infection, is inherited directly from the KFE through the mean-field closure.

Although the fractional-order model in Eq.~\eqref{eq:main} is theoretically well-defined for any $\alpha\in(0,1)$, the convergence performance of the numerical solver degrades when $\alpha$ becomes very small. Therefore, in all subsequent fractional-order simulations, we restrict $\alpha$ to values $\geq 0.5$, which ensures reliable numerical integration, while still capturing a wide range of memory effects. Furthermore, $\alpha\in(0,1)$ corresponds to the subdiffusive, memory-retaining regime that is the physical target of this work, where information spreads more slowly than the exponential prediction of classical models. For $\alpha>1$, the generalized Caputo definition based on the fractional part of $\alpha$ would be required, corresponding to a superdiffusive regime; since this regime lies outside the scope of the present study, we restrict our analysis to the physically relevant interval $\alpha\in(0,1)$.

\subsection{Validation of Integer-Order Model and Effects of Fractional Dynamics}\label{sec:exp1}

In this first simulation, we validate the proposed UAU-SIS modeling framework by comparing the deterministic individual-based mean-field approximation (IBMFA) with stochastic Monte Carlo simulations for the integer-order system ($\alpha=1$), and subsequently investigate the impact of fractional dynamics ($0.5\leq\alpha<1$) on the spreading process.

The parameter values used throughout our numerical study are selected within the ranges commonly adopted in the awareness-epidemic coevolution literature \cite{granell2013dynamical, wang2020epidemic}, with $\mu=1.0$ serving as the time-normalization scale so that the relevant quantities are the rate ratios $\beta^U/\mu$, $\lambda/\mu$, and $\delta/\mu$. As an illustrative interpretation, if we assume one unit of time to be one day, then $\mu=1$ corresponds to a mean infectious period of one day, $\beta^U$ is the expected number of transmissions from an infected node to a given susceptible neighbor per day, and $\delta=0.5$ means that an aware individual forgets, on average, every two days. Under this interpretation, the infection rates $\beta^U\in\{0.3,0.6,1.5\}$ used below range from subcritical to strongly endemic transmission and are comparable with the per-contact transmission rates of rapidly spreading respiratory infections or of rumor propagation in dense online communities, whereas values such as $\beta^U=15$ or $150$ would imply tens to hundreds of transmissions per infectious period and are not representative of human interaction processes. For each simulation, the parameter regime is consistent with the corresponding theoretical conditions: whenever $\beta^U$ lies above the threshold value (i.e., $\mathcal{R}_0>1$), the system converges to the rumor-endemic equilibrium whose local stability is guaranteed by Theorem~5; whenever $\beta^U$ lies below the threshold ($\mathcal{R}_0<1$), the system converges to the rumor-free equilibrium, consistent with Theorem~3 and Corollary~1. The convergence of the numerical trajectories to the analytically predicted equilibria thus provides numerical evidence that the chosen parameters are consistent with the conditions underlying the main theoretical conclusions.

We first examine the consistency between the IBMFA and stochastic simulations using three distinct parameter sets designed to highlight different dynamical regimes:
\begin{itemize}
    \item \textbf{Baseline}: Standard parameters ($\lambda=0.3$, $\beta^U=0.6$, $\gamma=1/3$, $\delta=0.5$, $\mu=1.0$)
    \item \textbf{High Infection}: Reduced awareness spreading ($\lambda=0.1$) with an enhanced infection rate ($\beta^U=1.5$)
    \item \textbf{High Awareness}: Strong awareness spreading ($\lambda=1.0$) with a reduced infection rate ($\beta^U=0.3$) and a lower forgetting rate ($\delta=0.2$)
\end{itemize}

All simulations are conducted on a duplex Barabási-Albert network with $N=1000$ nodes, where the awareness layer is generated with $m_A=4$ and the risk layer with $m_B=2$. Initial conditions are set with $2\%$ of nodes in the $AI$ state and $3\%$ in the $AS$ state (i.e., $5\%$ aware in total), distributed randomly in the network. The results are shown in Fig.~\ref{fig:exp1_integer}.

\begin{figure*}[!ht]
    \centering
    \subfigure[State dynamics comparison between IBMFA and Monte Carlo simulations]{
        \includegraphics[width=0.8\linewidth]{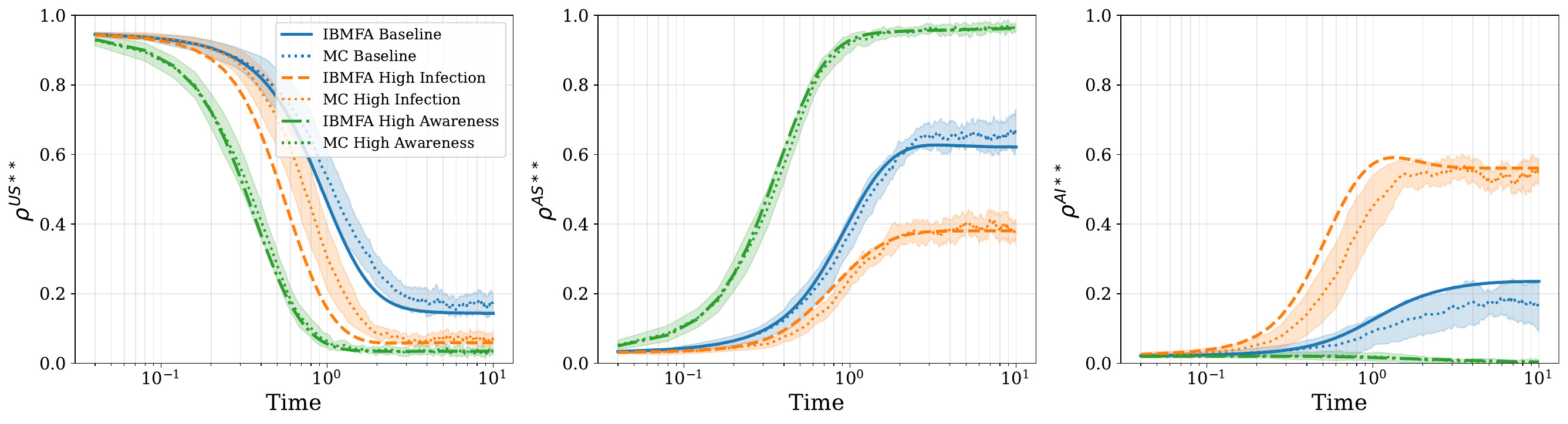}
        \label{fig:exp1a}
    }
    \subfigure[Quantitative analysis of final state densities and parameters]{
        \includegraphics[width=0.8\linewidth]{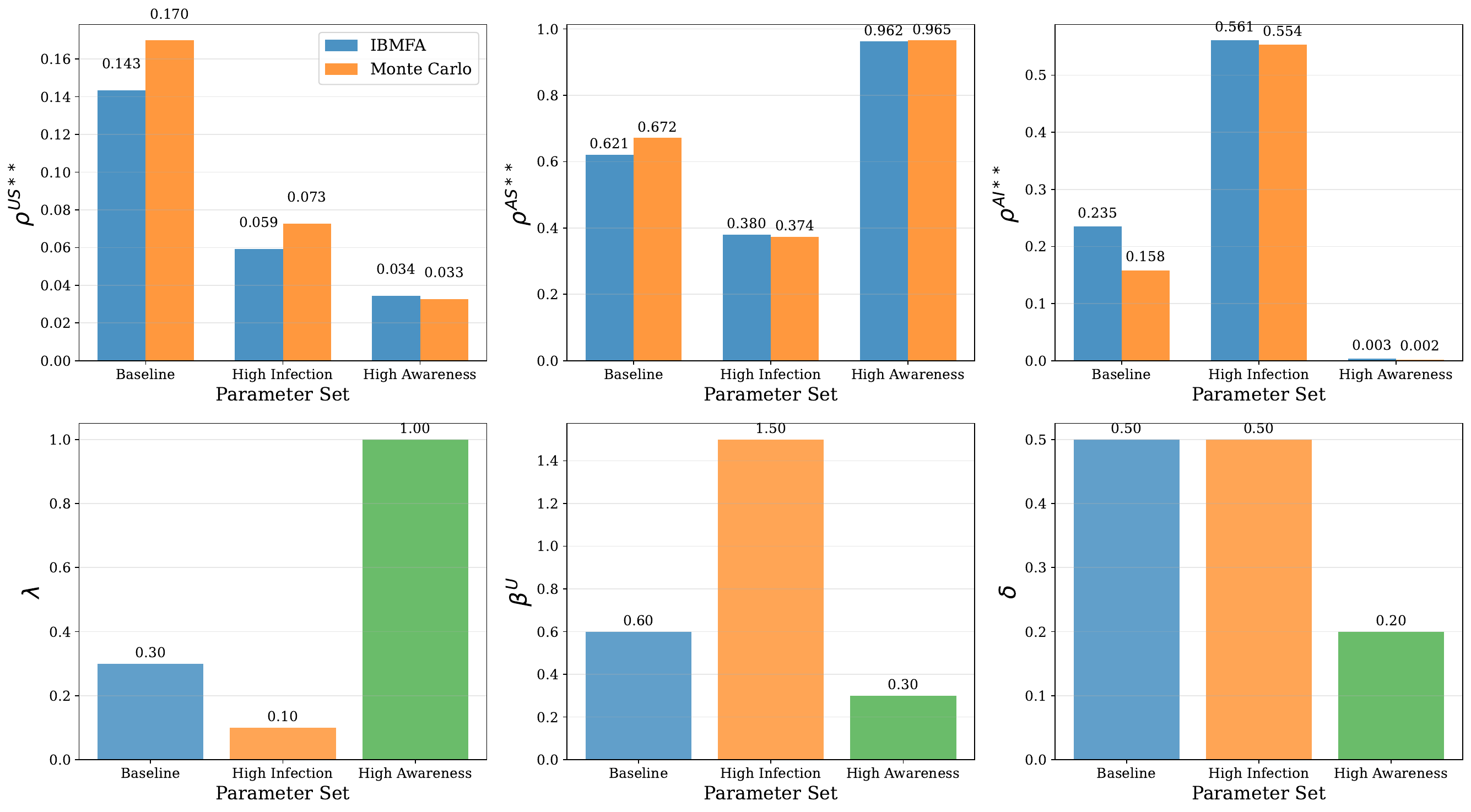}
        \label{fig:exp1b}
    }
    \vspace{-3mm}
    \caption{
    \textbf{Validation of integer-order UAU-SIS model through comparison of IBMFA and Monte Carlo simulations.}
    (a) displays the temporal evolution of state probabilities $p^{US}$, $p^{AS}$, and $p^{AI}$ for three parameter configurations using logarithmic time axis. Solid lines represent IBMFA solutions, while dashed lines with shaded regions show Monte Carlo means $\pm$ one standard deviation from 15 independent runs. 
    (b) provides quantitative analysis of final state densities (top row) and corresponding parameter values (bottom row). The Baseline configuration ($\lambda=0.3$, $\beta^U=0.6$, $\delta=0.5$) serves as reference, while High Infection ($\lambda=0.1$, $\beta^U=1.5$) and High Awareness ($\lambda=1.0$, $\beta^U=0.3$, $\delta=0.2$) configurations demonstrate contrasting dynamical behaviors.
    }
    \vspace{-1mm}
    \label{fig:exp1_integer}
\end{figure*}

Fig.~\ref{fig:exp1a} demonstrates an excellent agreement between the IBMFA solutions and Monte Carlo simulations across all three parameter configurations, validating the accuracy of our mean-field approximation. The logarithmic time axis effectively reveals both the rapid early-stage dynamics and the approach to the steady-state equilibrium. 

Several key observations emerge from the comparative analysis. In the High Infection configuration (orange curves), the reduced awareness spreading rate ($\lambda=0.1$) combined with an enhanced infection rate ($\beta^U=1.5$) leads to a rapid depletion of the unaware susceptible population ($p^{US}$) and sustains elevated infection levels ($p^{AI}$). Conversely, the High Awareness configuration (green curves) exhibits a strong suppression of infection through an enhanced awareness spreading ($\lambda=1.0$) and a reduced infection susceptibility ($\beta^U=0.3$), resulting in a higher fraction of aware susceptible individuals ($p^{AS}$).

The quantitative analysis in Fig.~\ref{fig:exp1b} reveals that both modeling approaches converge to nearly identical steady-state values, with maximum discrepancies of less than $0.02$ across all states and parameter configurations. This remarkable consistency between the deterministic IBMFA and stochastic Monte Carlo simulations provides strong empirical support for the theoretical foundations developed in Section~\ref{sec:modeling}.

We now extend our investigation to the fractional-order UAU-SIS model by solving Eq.~\eqref{eq:main_component} for $\alpha \in \{0.7, 0.8, 0.9\}$ using the same three parameter configurations. The Block-Adaptive FFT Solver is employed to efficiently handle the long-memory effects inherent in Caputo fractional derivatives. The results are shown in Fig.~\ref{fig:exp1_fractional}.

\begin{figure*}[!ht]
    \centering
    \subfigure[Fractional dynamics across parameter configurations]{
        \includegraphics[width=0.8\linewidth]{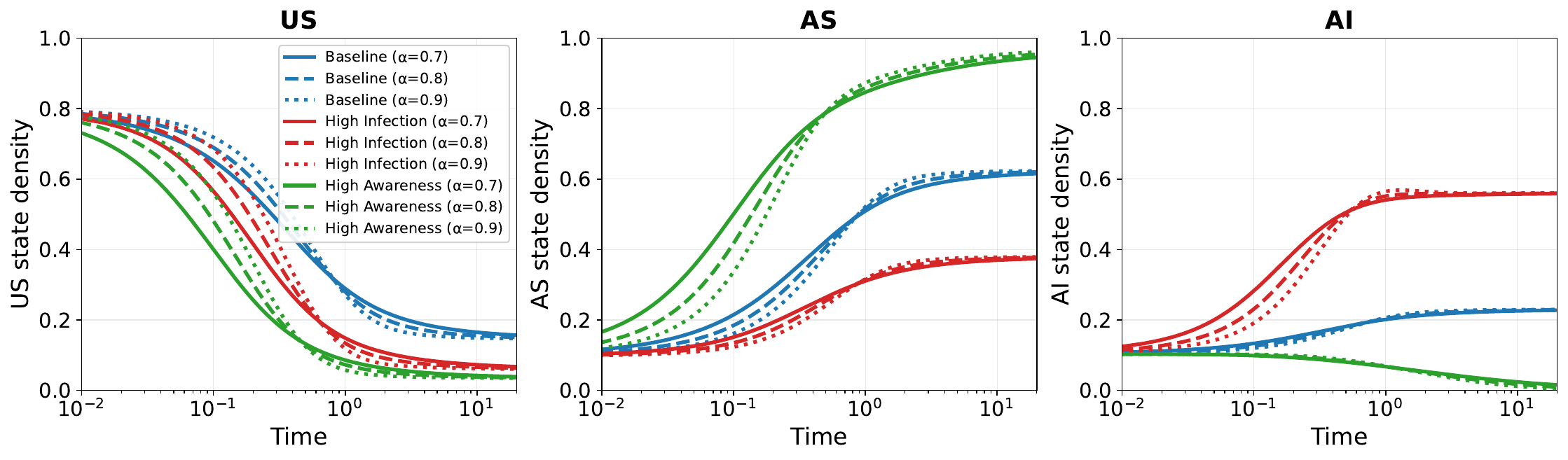}
        \label{fig:exp1c}
    }
    \vspace{-3mm}
    \subfigure[Memory effects in baseline configuration]{
        \includegraphics[width=0.6\linewidth]{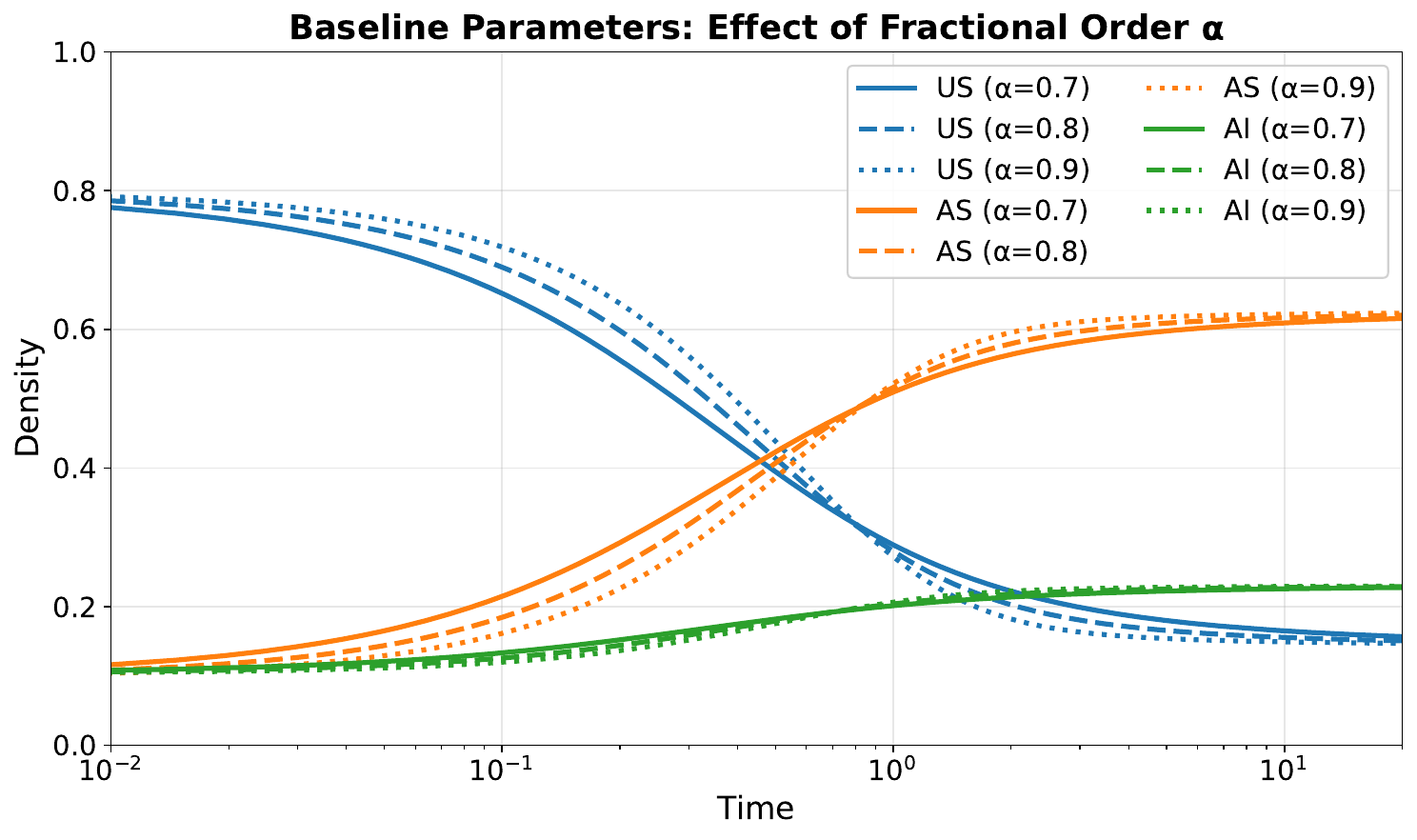}
        \label{fig:exp1d}
    }
    \vspace{-1mm}
    \caption{
    \textbf{Impact of fractional order $\alpha$ on UAU-SIS dynamics.}
    (a) illustrates the temporal evolution of state probabilities for three parameter configurations with varying fractional orders $\alpha \in \{0.7, 0.8, 0.9\}$. Each subplot corresponds to a different state ($p^{US}$, $p^{AS}$, $p^{AI}$), with line styles distinguishing between $\alpha$ values.
    (b) focuses on the Baseline parameter configuration, comparing the dynamics for different $\alpha$ values to highlight memory effects. The subdiffusive nature of fractional dynamics manifests as delayed response and slower convergence to steady states, with more pronounced effects for smaller $\alpha$ values.
    }
    \vspace{-1mm}
    \label{fig:exp1_fractional}
\end{figure*}

Fig.~\ref{fig:exp1c} reveals that while the qualitative behavior of the fractional-order system follows similar patterns to the integer-order case, significant quantitative differences emerge during the transient phase. The memory effects introduced by the Caputo derivative produce a characteristic subdiffusive behavior, particularly evident in the delayed peak infection times and a slower convergence to equilibrium.

Notably, the steady-state values remain largely unaffected by the fractional order $\alpha$, as demonstrated in Fig.~\ref{fig:exp1d}. This observation suggests that while memory effects substantially influence the temporal evolution and early-stage dynamics, they do not alter the long-term equilibrium properties of the system. The persistence of rumor prevalence despite memory-induced delays underscores the fundamental limitation of information-based interventions in completely suppressing rumor propagation once the propagation threshold is exceeded.

The observed subdiffusive behavior becomes more pronounced with decreasing $\alpha$, consistent with the theoretical expectation that smaller fractional orders correspond to stronger memory effects and slower diffusion processes. This memory-dependent temporal scaling provides a mechanistic explanation for the prolonged rumor durations often observed in real-world information diffusion processes with heavy-tailed inter-event times.

To further validate the consistency of the fractional-order formulation, we perform a convergence test as the fractional order approaches the integer limit ($\alpha \to 1$). As detailed in Section~12.1 of the Supplementary Material, the fractional curves gradually approach the integer-order reference as $\alpha$ increases, and the root-mean-square error decreases monotonically to zero at $\alpha=1$, confirming that the memory effects introduced by the Caputo derivative produce a smooth interpolation between the Markovian and strongly non-Markovian regimes.

The simulation results yield four main insights. First, the close agreement between the IBMFA and the Monte Carlo simulations validates the continuous-time formulation of the UAU-SIS model and the underlying mean-field closure. Second, the contrasting dynamics under the different parameter configurations highlight the role of awareness spreading: effective information dissemination substantially reduces the rumor prevalence, although complete suppression remains challenging once the basic reproduction number exceeds unity. Third, memory effects reshape the transient dynamics without changing the steady-state properties, so that interventions should not rely on memory effects alone to suppress rumors in the long run. Fourth, the fractional-order formulation recovers the integer-order dynamics as $\alpha\to1$ (Section~12.1 of the Supplementary Material), so that $\alpha$ provides a well-controlled interpolation between the Markovian and the subdiffusive regimes.

\subsection{Impact of Fractional Order and Infection Rate}\label{sec:exp2}
To systematically investigate the role of memory effects and infection transmission in the coupled UAU-SIS process, 
we conduct a comprehensive parameter scan examining both the fractional order $\alpha$ 
and the infection rate in the unaware state $\beta^U$. This simulation aims to characterize 
the individual and combined effects of these two parameters on the system's transient dynamics 
and steady-state behavior.

We fix all other parameters to their baseline values ($N=1000$, $\mu=1.0$, $\delta=0.5$, 
$\lambda=0.3$, $\gamma=1/3$, $m_{A}=4$, $m_{B}=2$, 
$T=20.0$), and systematically vary $\alpha$ across the set 
$\{0.95,0.9,0.8,0.7,0.6,0.5\}$ and $\beta^U$ across $\{0.3,0.4,0.5,0.6,0.8,1.0\}$. 
The initial condition is configured with $80\%$ of nodes in the $US$ state, while the 
remaining $20\%$ are equally distributed between $AS$ and $AI$ states. All simulations 
are performed using our adaptive block solver for the fractional-order IBMFA equations. The simulation results are presented in Fig.~\ref{fig:exp2_combined}.

\begin{figure}[htbp]
    \centering
    \subfigure[Steady-state infection density]{
        \includegraphics[width=0.6\linewidth]{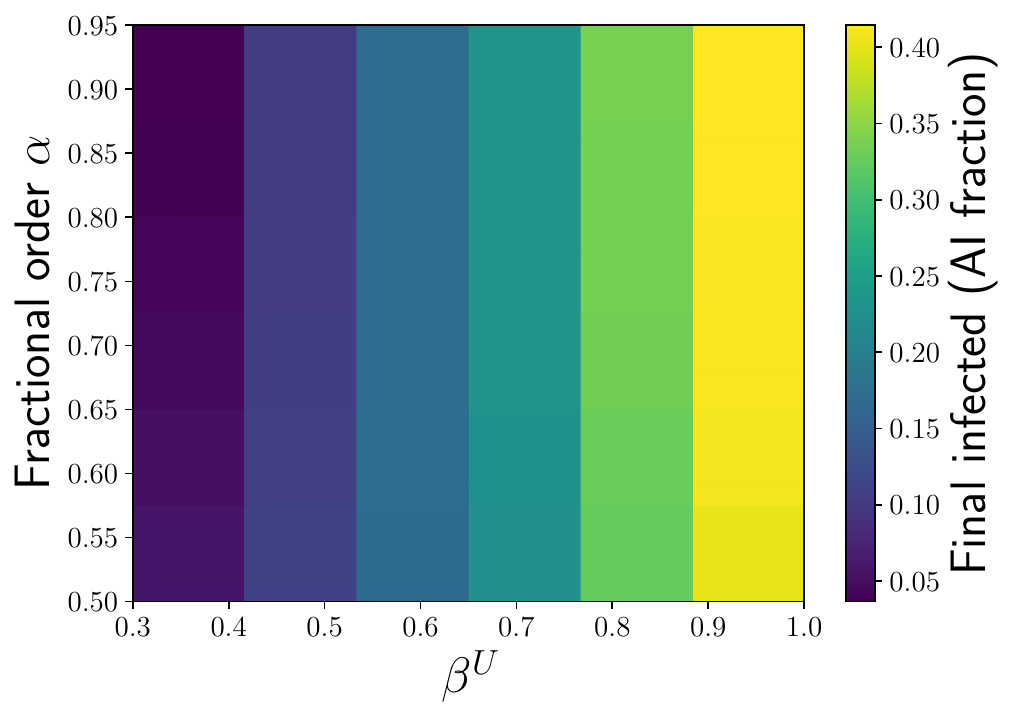}
        \label{fig:exp2heat}
    }
    \vspace{-3mm}
    
    \subfigure[$\beta^U = 0.3$]{
        \includegraphics[width=0.32\linewidth]{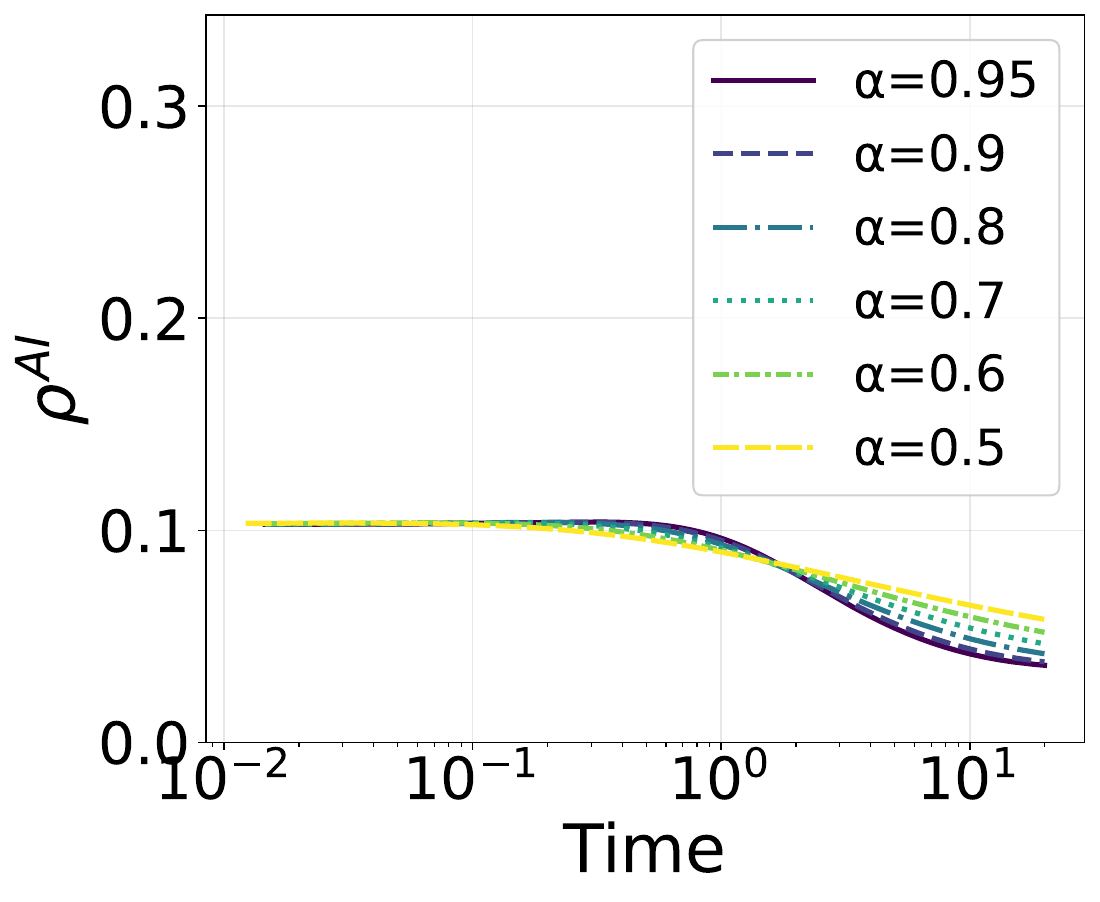}
        \label{fig:traj03}
    }
    \hspace{-5mm}
    \subfigure[$\beta^U = 0.4$]{
        \includegraphics[width=0.32\linewidth]{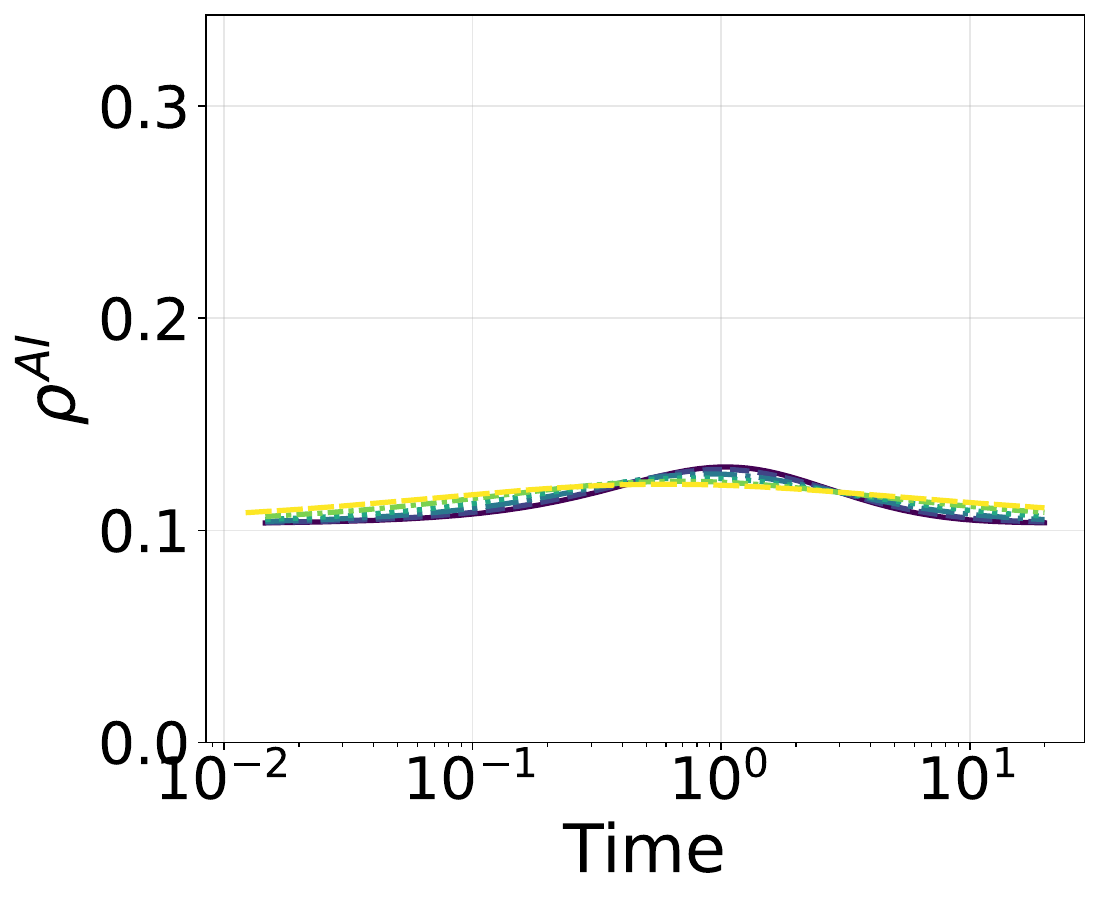}
        \label{fig:traj04}
    }
    \hspace{-5mm}
    \subfigure[$\beta^U = 0.5$]{
        \includegraphics[width=0.32\linewidth]{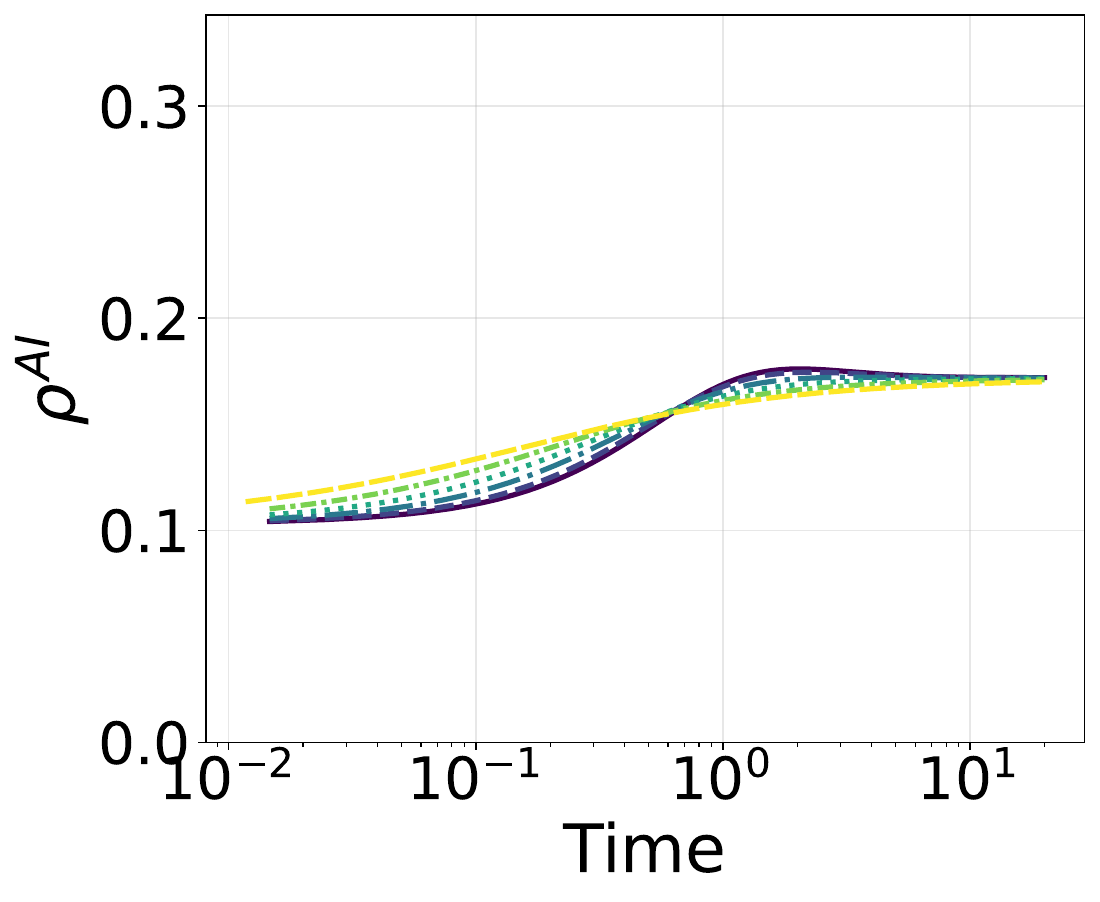}
        \label{fig:traj05}
    }
    \vspace{-3mm}
    
    \subfigure[$\beta^U = 0.6$]{
        \includegraphics[width=0.32\linewidth]{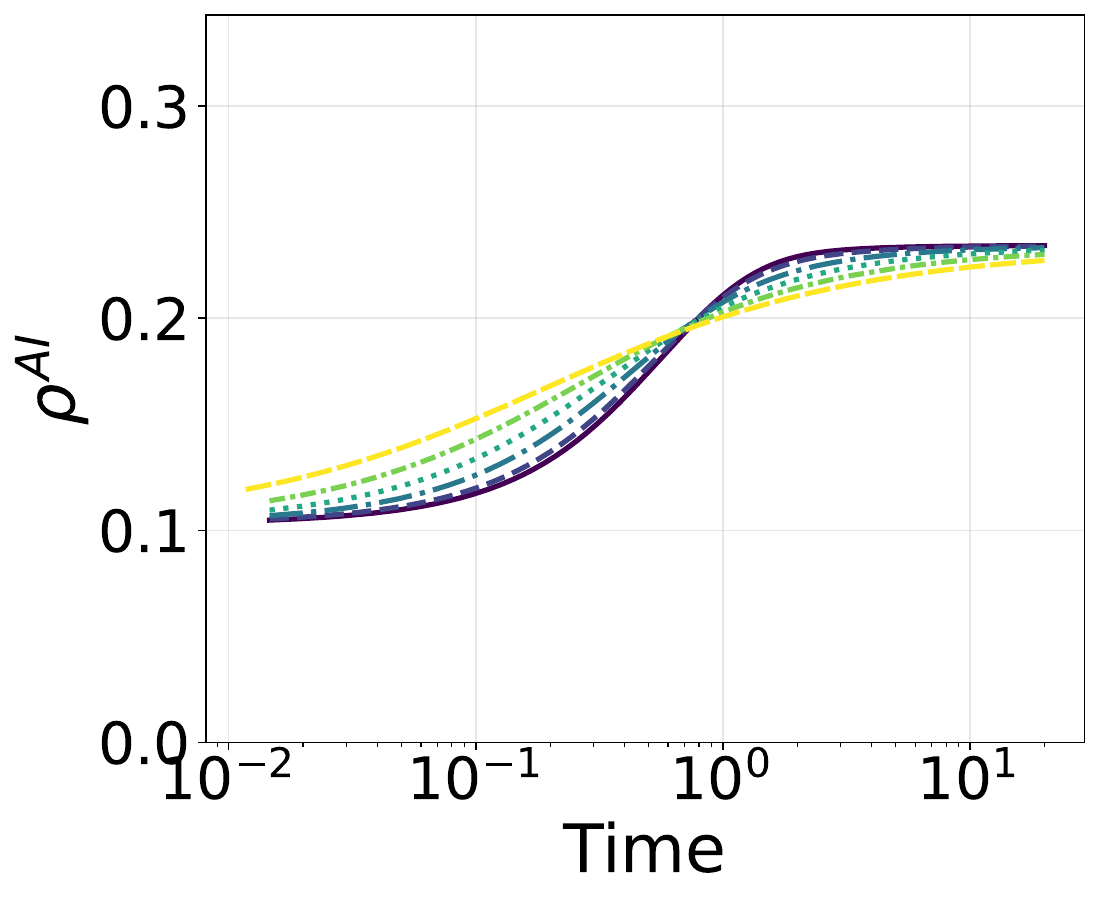}
        \label{fig:traj06}
    }
    \hspace{-5mm}
    \subfigure[$\beta^U = 0.8$]{
        \includegraphics[width=0.32\linewidth]{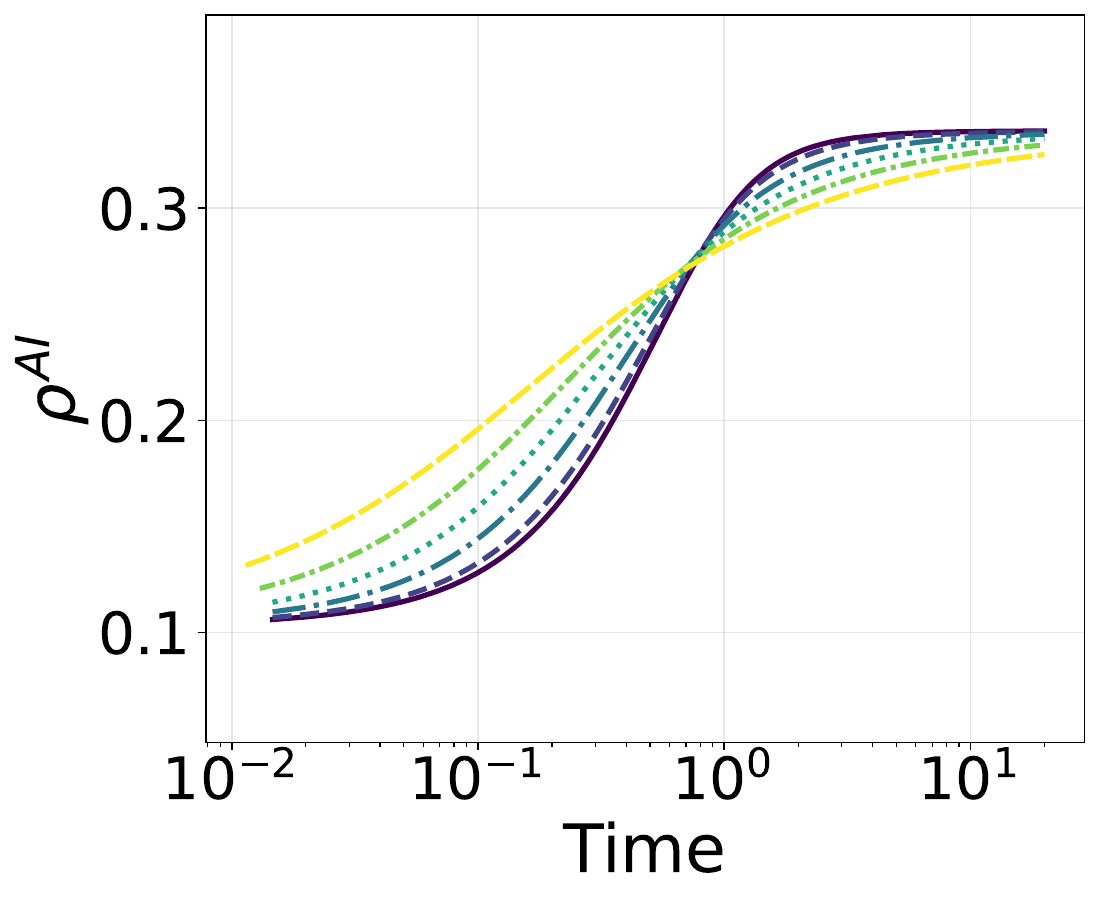}
        \label{fig:traj08}
    }
    \hspace{-5mm}
    \subfigure[$\beta^U = 1.0$]{
        \includegraphics[width=0.32\linewidth]{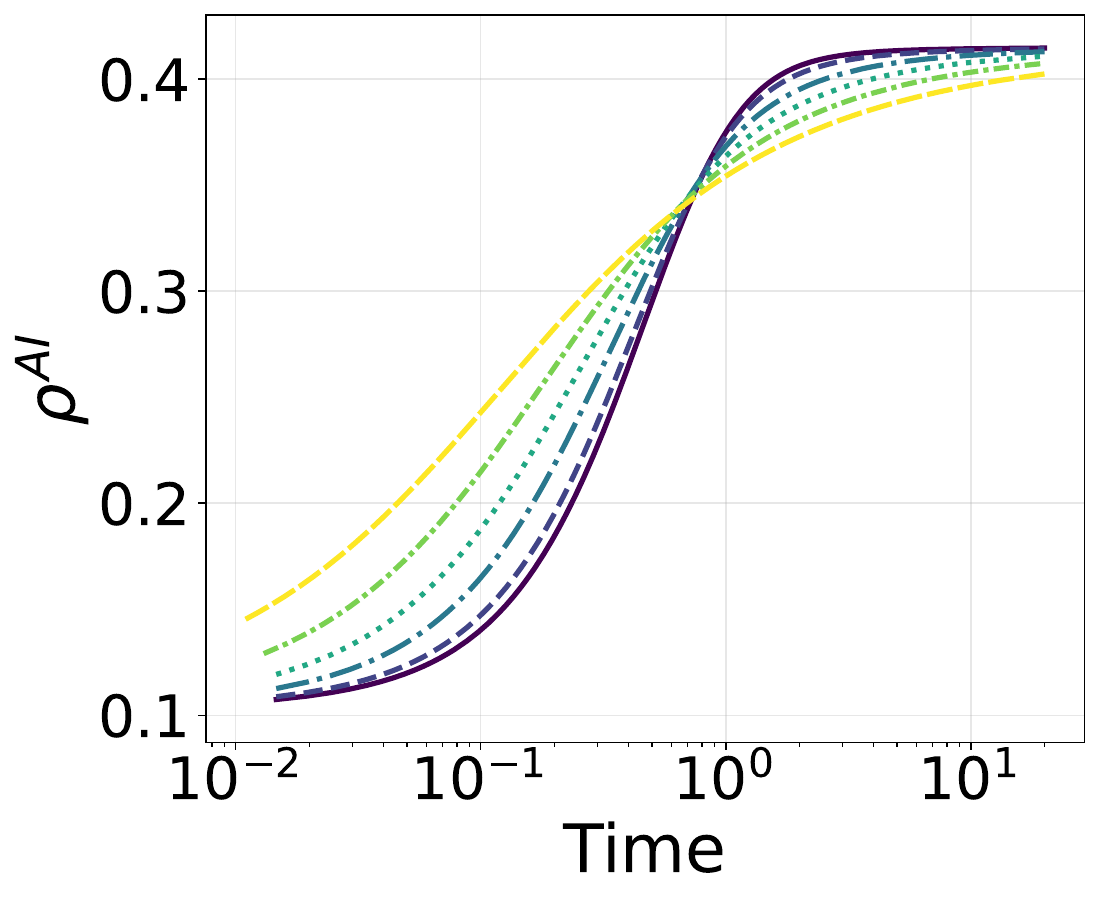}
        \label{fig:traj10}
    }
    \vspace{-3mm}
    \caption{
        \textbf{Steady-state infection density and temporal dynamics as functions of the fractional order $\alpha$ and the infection rate $\beta^U$.}
        (a) Heatmap of the steady-state infection density (AI fraction) over
        $\beta^U \in \{0.3,0.4,0.5,0.6,0.8,1.0\}$ and $\alpha \in \{0.95,0.9,0.8,0.7,0.6,0.5\}$,
        computed with the fractional-order IBMFA solver ($T=20$, $N=1000$, $\mu=1.0$, $\delta=0.5$, $\lambda=0.3$,
        $\gamma=1/3$, $m_A=4$, $m_B=2$, with $80\%$ of the nodes initially in the $US$ state and $10\%$ each in the
        $AS$ and $AI$ states). Higher $\beta^U$ increases the steady-state prevalence, whereas $\alpha$ has little
        effect on the final state. (b)-(g) Temporal dynamics of the infection density for the different $\beta^U$ values,
        each panel showing the trajectories for $\alpha$ from $0.5$ to $0.95$ on a logarithmic time axis. Lower $\alpha$
        (stronger memory) slows the convergence to equilibrium without shifting it, so that the fractional order
        primarily reorganizes the transient dynamics.
    }
    \vspace{-1mm}
    \label{fig:exp2_combined}
\end{figure}
An analysis of the heatmap in Fig.~\ref{fig:exp2heat} reveals a distinct pattern of vertical striping. As anticipated by theory, increasing the infection rate $\beta^U$ significantly elevates the steady-state infection density, confirming that $\beta^U$ is a primary determinant of the system's final rumor prevalence. However, a key observation is that for any fixed value of $\beta^U$, variations in $\alpha$ result in negligible changes to the steady-state prevalence. This indicates that while $\alpha$ profoundly influences the transient dynamics, as will be shown in the subsequent panels, it does not affect the long-term equilibrium of the system.

These observations reveal a fundamental dichotomy of fractional-order rumor models: the memory encoded in $\alpha$ governs the journey of an outbreak, whereas its final destination is dictated by the network topology and the intrinsic reaction rates (such as $\beta^U$ and $\mu$). Memory can accelerate or delay the approach to equilibrium, but it cannot change the equilibrium itself.

The temporal dynamics presented in Figs~\ref{fig:traj03}-\ref{fig:traj10} elucidate the transient behavior of the system across all parameter combinations. These panels are particularly instrumental in analyzing the impact of $\alpha$ on the system's approach to equilibrium. A consistent pattern emerges across all $\beta^U$ values: the infection density trajectories exhibit markedly different temporal profiles depending on $\alpha$, demonstrating that memory effects substantially reshape the transient process.

Specifically, regardless of $\beta^U$, systems with lower $\alpha$ values consistently display more sluggish responses and flatter temporal profiles. This phenomenon directly manifests the \textit{subdiffusive} nature of the diffusion process characterized in Section~\ref{sec:modeling}, where information spreads more slowly than the exponential growth predicted by traditional models. The stronger memory effects (lower $\alpha$) introduce a temporal drag on the system's evolution, causing delayed response to infection pressures and prolonged convergence to steady state.

This decoupling between the transient dynamics and the final outcome has direct implications for rumor management in highly knowledgeable communities. The subdiffusive slowdown associated with strong memory effects opens a longer window for intervention, but it also prolongs the tail of an outbreak and therefore calls for sustained control measures. Memory thus reorganizes the temporal landscape of rumor spreading without altering its final prevalence.

To further corroborate the threshold-governed nature of the coupled awareness--rumor system, we numerically verify the transition from the rumor-free to the rumor-endemic regime. As detailed in Section~12.2 of the Supplementary Material, the bifurcation diagram exhibits a clear threshold transition at $\mathcal R_0\approx1$, in agreement with the transcritical bifurcation predicted by Theorem~\ref{theorem:endemic_existence}, and the fractional-order data points ($\alpha=0.6$ and $\alpha=0.8$) approach the same $\alpha$-independent equilibrium, consistent with Remark~1.

The global convergence behavior underlying the stability analysis of Section~II-B is verified numerically in the Supplementary Material (Section~11.4): starting from twelve markedly different initial conditions (infection fractions from $1\%$ to $30\%$ and awareness fractions from $1\%$ to $60\%$), all trajectories of the integer-order system converge to the same steady-state value $\rho^{AI**}=0.2301$ in the endemic regime ($\beta^U=0.6$, $\mathcal R_0=2.07$, spread below $5\times10^{-6}$) and decay to $\rho^{AI}\approx 0$ in the rumor-free regime ($\beta^U=0.1$, $\mathcal R_0=0.35<1$, down to $10^{-7}$); representative fractional-order trajectories ($\alpha=0.8$) converge to the same equilibria up to the tolerance of the fractional solver over the finite horizon $T=20$.

\subsection{Impact of Awareness Layer and Network Topology}\label{sec:exp3}

We further examine how the awareness layer and the network architecture of the two layers
shape the steady-state rumor prevalence. The complete account of the two experiments is
given in Sections~12.3 and~12.4 of the Supplementary Material, and we summarize their main
findings here.

For the awareness layer, we fix $\alpha=0.8$ and $\beta^U=0.4$ on a multiplex network with
$N=250$ nodes ($m_B=2$) and scan the awareness spreading rate $\lambda\in[0.1,1.0]$, the
forgetting rate $\delta\in[0.1,1.0]$, and the connectivity of the awareness layer
$m_A\in\{2,4,8\}$. The results show a consistent pattern across all connectivity
configurations: slow spreading combined with rapid forgetting of corrective information
(small $\lambda$ and large $\delta$) raises the steady-state rumor prevalence by
$20\%$--$30\%$ relative to the optimal combination (large $\lambda$ and small $\delta$),
whereas denser awareness layers ($m_A=8$) systematically lower the average prevalence with
respect to sparse ones ($m_A=2$). The awareness layer therefore modulates the rumor
dynamics both through its dynamical parameters $(\lambda,\delta)$ and through the topology
of the information channel.

For the network architecture, we compare all nine combinations of Erdős--Rényi (ER),
Barabási--Albert (BA), and regular networks in the two layers, with the average degrees
held at approximately $4$ and $2$ for the awareness and the risk layer, respectively. The
state composition is nearly invariant along each row of the resulting $3\times3$ grid,
which indicates that, under the low awareness spreading rate ($\lambda=0.1$) used in this
experiment, the awareness-layer topology exerts only a limited influence on the final
prevalence. The risk-layer topology, by contrast, plays a dominant role: BA networks yield
the highest infection density, followed by regular and then ER networks. This ordering
follows from degree heterogeneity: the hubs of BA networks are infected early and act as
super-spreaders, the homogeneous random structure of ER networks spreads the rumor more
uniformly and slowly, and regular networks lie in between. The risk-layer topology
additionally modulates the awareness distribution through the infection process itself,
because infected nodes enter the aware state immediately. Taken together, these results
indicate that rumor containment in knowledgeable communities should combine awareness
campaigns with structural interventions on the risk layer.

\section{Conclusions and future work}\label{sec:conclusion}
In this work, we proposed a fractional-order SIS-UAU model on multiplex networks to capture the intertwined dynamics between rumor propagation and awareness diffusion under memory effects. The model distinguishes an awareness layer that disseminates debunking information from a risk layer that governs rumor transmission, and couples both via state-dependent cross-layer interactions. By employing the Caputo fractional derivative, the model incorporates a realistic form of temporal memory that reflects subdiffusive human responses, hesitation, and delayed adoption of awareness information. On the theoretical side, we established several fundamental properties of the system: positivity of the solution components, invariance of the probabilistic simplex, Lipschitz continuity of the system vector field, local existence and uniqueness of solutions, stability conditions for both the rumor-free and endemic equilibria under the Matignon criterion, and a global asymptotic stability result for the rumor-free equilibrium under a sufficient condition on both layers (Theorem~6, with a complete self-contained proof in the Supplementary Material). These results ensure that the fractional-order dynamics are mathematically well-posed and provide rigorous guarantees for stability characterization across different memory regimes.

To efficiently solve the resulting high-dimensional fractional-order dynamical system, we developed a block-adaptive FFT solver tailored for long-horizon Caputo dynamics. With this numerical tool, we systematically investigated the temporal evolution and steady-state behavior of the coupled awareness–rumor process across a broad range of parameters and network structures. Our numerical findings reveal several robust and interpretable phenomena: (1) the fractional order plays a decisive role in shaping transient behavior: strong memory effects markedly slow down the system’s relaxation and induce pronounced subdiffusive transients, yet they do not alter the eventual equilibrium state. (2) The infection-related parameters on the rumor layer dominate the determination of steady-state rumor prevalence, while the fractional order primarily influences the time scale of convergence. (3) The awareness layer has a substantial and non-cosmetic impact on rumor dynamics: stronger awareness transmission, slower fading of awareness, and higher awareness-layer connectivity significantly reduce the long-term rumor burden. (4) The network architecture of the multiplex system heavily influences the outcome. The topology of the rumor layer is the principal determinant of global rumor prevalence—particularly the heterogeneity of degree distributions—while the awareness-layer topology provides a secondary but non-negligible modulation through indirect cross-layer feedback. Together, these results demonstrate that memory effects, awareness-layer dynamics, and multilayer network structures jointly shape the unfolding of rumor propagation, and accurate modeling requires integrating all three dimensions.

This work establishes a theoretical--numerical framework for investigating rumor spreading under the interplay of memory effects, awareness dynamics, and multilayer social connectivity. Several extensions follow naturally. (1) Behavioral heterogeneity could be introduced by allowing individuals to differ in skepticism or sensitivity to corrective information, and higher-order interactions via hypergraphs~\cite{millan2025topology} could capture group-level social reinforcement beyond pairwise contacts. (2) Relaxing the instantaneous awareness assumption ($\kappa\to\infty$) to a finite delay would yield a richer four-state dynamics with more realistic behavioral lags. (3) A systematic multi-kernel analysis---repeating the threshold, stability, and transient analyses under the Caputo--Fabrizio, Atangana--Baleanu, and fractal--fractional operators~\cite{akgul2026harmful,nisar2025assessment,farman2025dynamics}---would clarify how much of the observed memory-induced behavior is intrinsic to the model rather than to the specific kernel. (4) Calibrating the model against empirical or mechanism-based synthetic data, generated by a generalized Gillespie algorithm with heavy-tailed waiting times~\cite{boguna2014simulating}, would enable parameter estimation and quantitative validation. (5) Establishing rigorous global asymptotic stability under the sharp thresholds requires an asymptotic-autonomy (limit-system) principle for Caputo fractional systems, which remains open; Theorem~6 provides a sufficient-condition result, and the multi-initial-condition experiment reported in Section~11.4 of the Supplementary Material offers numerical evidence that the sharp-threshold behavior is globally attracting. These directions would broaden the scope and applicability of the fractional SIS-UAU framework.

\ifCLASSOPTIONcaptionsoff
  \newpage
\fi

\bibliographystyle{unsrtnat}

\scriptsize
\bibliography{cas-refs}

\normalsize


\begin{IEEEbiography}[{\includegraphics[width=1in,height=1.25in,clip,keepaspectratio]{Qin\_Li.png}}]{Qin~Li}
(Member, IEEE) received her Ph.D. degree in management from the School of Public Policy and Administration, Chongqing University, Chongqing, China, in 2025. Since 2025, she has been a lecturer in the Business College, Southwest University, Chongqing, China. Currently, she serves as an Editorial Board Member for Humanities \& Social Sciences Communications and Scientific Reports. Her research interests include Computational Social Science, Evolutionary Games, Epidemiological Models, and Complex Systems.
\end{IEEEbiography}

\begin{IEEEbiography}[{\includegraphics[width=1in,height=1.25in,clip,keepaspectratio]{Jiexi\_Tang.png}}]{Jiexi~Tang}
received the B.E. degree in artificial intelligence from the College of Artificial Intelligence, Southwest University, Chongqing, China, in 2026. He is currently pursuing the M.S. degree in intelligent communication and computational communication with the School of Journalism and Communication, Beijing Normal University, Beijing, China. He has published one journal article in Applied Mathematical Modelling. His research interests include complex networks, information propagation, and computational communication.
\end{IEEEbiography}

\begin{IEEEbiography}[{\includegraphics[width=1in,height=1.25in,clip,keepaspectratio]{Jurgen\_Kurths.png}}]{J\"{u}rgen~Kurths}
received the B.S. degree in mathematics from the University of Rostock, Rostock, Germany, the Ph.D. degree from the Academy of Sciences, German Democratic Republic, Berlin, Germany, in 1983, the Honorary degree from N. I. Lobachevsky State University, Nizhny Novgorod, Russia, in 2008, and the Honorary degree from Saratov State University, Saratov, Russia, in 2012.

From 1994 to 2008, he was a Full Professor with the University of Potsdam, Potsdam, Germany. Since 2008, he has been a Professor in nonlinear dynamics with the Humboldt University of Berlin, Berlin, Germany, and the Chair of the Research Domain Complexity Science with the Potsdam Institute for Climate Impact Research, Potsdam. He has authored more than 700 papers, which are cited more than 60\,000 times (H-index: 111). He is a Highly Cited Researcher in Engineering. His main research interests include synchronization, complex networks, and time series analysis and their applications. He is a member of the Academia Europaea. He is a Fellow of the American Physical Society, the Royal Society of Edinburgh, and the Network Science Society. He was a recipient of the Alexander von Humboldt Research Award from India in 2005 and from Poland in 2021, the Richardson Medal of the European Geophysical Union in 2013, and the Eight Honorary Doctorates. He is the Editor-in-Chief of \textit{Chaos} and on the editorial boards of more than ten journals.
\end{IEEEbiography}

\begin{IEEEbiography}[{\includegraphics[width=1in,height=1.25in,clip,keepaspectratio]{Matjaz\_Perc.png}}]{Matja\v{z}~Perc}
(Member, IEEE) received the Ph.D. degree in physics from the University of Maribor in 2006. He is currently Professor of Physics at the University of Maribor, Staff Researcher at the Community Healthcare Center Dr. Adolf Drolc Maribor, and Adjunct Professor at Kyung Hee University and Korea University. He is a member of Academia Europaea and the European Academy of Sciences and Arts, and among top 1\% most cited physicists according to Clarivate Analytics data. He is also the 2015 recipient of the Young Scientist Award for Socio and Econophysics from the German Physical Society, and the 2017 USERN Laureate. In 2018 he received the Zois Award, which is the highest national research award in Slovenia. In 2019 he became Fellow of the American Physical Society.
\end{IEEEbiography}

\begin{IEEEbiography}[{\includegraphics[width=1in,height=1.25in,clip,keepaspectratio]{Minyu\_Feng.png}}]{Minyu Feng}
(Senior Member, IEEE) received his Ph.D. in Computer Science in 2018 through a joint program between the University of Electronic Science and Technology of China, Chengdu, China, and Humboldt University of Berlin, Berlin, Germany. Since 2019, he has been an Associate Professor at the College of Artificial Intelligence, Southwest University, Chongqing, China.

Dr. Feng has published more than 90 peer-reviewed papers in authoritative journals, including IEEE Transactions on Pattern Analysis and Machine Intelligence, IEEE Transactions on Systems, Man, and Cybernetics: Systems, and IEEE Transactions on Networking. He is a Senior Member of the China Computer Federation (CCF) and the Chinese Association of Automation (CAA). Currently, he serves as a Subject Editor for Applied Mathematical Modelling and an Editorial Advisory Board Member for Chaos. He has also served as a current or former Editorial Board Member of more than 10 journals.

Dr. Feng's research interests include complex systems, evolutionary game theory, computational social science, and mathematical epidemiology.
\end{IEEEbiography}

\clearpage
\setcounter{section}{0}
\setcounter{equation}{0}
\setcounter{figure}{0}
\setcounter{table}{0}
\renewcommand{\thefigure}{S\arabic{figure}}
\renewcommand{\thetable}{S\arabic{table}}

\title{Fractional-order SIS-UAU subdiffusion model for rumor propagation\\(Supplementary Material)}

\author{Qin~Li,~\IEEEmembership{Member,~IEEE,}
        Jiexi~Tang,
        J\"{u}rgen~Kurths,
        Matja\v{z}~Perc,~\IEEEmembership{Member,~IEEE,}
        and~Minyu~Feng,~\IEEEmembership{Senior~Member,~IEEE,}%
\thanks{This work was supported by the Natural Science Foundation of Chongqing under Grant No.~CSTB2025YITP-QCRCX0007, the Fundamental Research Funds for the Central Universities under Grant Nos.~XJ2026004101 and SWU-KT26010, and the Slovenian Research and Innovation Agency under Grant No.~P1-0403 (corresponding authors: Matja\v{z} Perc and Minyu Feng).}%
\thanks{Qin Li is with the Business College, Southwest University, Chongqing 402460, China.}%
\thanks{Jiexi Tang is with the School of Journalism and Communication, Beijing Normal University, Beijing 100875, China.}%
\thanks{J\"{u}rgen Kurths is with the Potsdam Institute for Climate Impact Research, 14437 Potsdam, Germany, and also with the Department of Physics, Humboldt University of Berlin, 10117 Berlin, Germany.}%
\thanks{Matja\v{z} Perc is with the Faculty of Natural Sciences and Mathematics, University of Maribor, Koro\v{s}ka cesta 160, Maribor, 2000, Slovenia, Community Healthcare Center Dr. Adolf Drolc Maribor, Ulica talcev 9, Maribor, 2000, Slovenia, Department of Physics, Kyung Hee University, 26 Kyungheedae-ro, Seoul, 02447, Dongdaemun-gu, Republic of Korea, and University College, Korea University, 145 Anam-ro, Seoul, 02841, Seongbuk-gu, Republic of Korea (email: matjaz.perc@gmail.com).}%
\thanks{Minyu Feng is with the College of Artificial Intelligence, Southwest University, Chongqing 400715, China (email: myfeng@swu.edu.cn).}}

\markboth{IEEE TRANSACTIONS ON SYSTEMS, MAN, AND CYBERNETICS: SYSTEMS}%
{Li \MakeLowercase{\textit{et al.}}: \MakeLowercase{Fractional-Order SIS-UAU Subdiffusion Model for Rumor Propagation}}


\clearpage
\normalsize
\begin{center}
{\Large\bfseries SUPPLEMENTARY MATERIAL}\\[1.5mm]
{\normalsize\bfseries Fractional-order SIS-UAU subdiffusion model for rumor propagation}\\[1.2mm]
{\small Qin Li, Jiexi Tang, J\"{u}rgen Kurths, Matja\v{z} Perc, and Minyu Feng}
\end{center}
\vspace{1mm}
{\small
\noindent\textbf{Abstract---}This supplementary material provides complete proofs for all theoretical results presented in the main manuscript, including lemmas on fractional calculus properties, existence and uniqueness of solutions, positivity and boundedness of the fractional-order SIS-UAU model, stability analysis of both rumor-free and rumor-endemic equilibria, and the basic reproduction number. These detailed derivations are essential for the theoretical completeness but are presented separately to maintain the conciseness of the main manuscript. In addition, Section~12 collects the additional numerical experiments that complement the simulations of the main manuscript, namely the convergence test towards the integer-order limit, the bifurcation diagram of the system, and the analyses of the awareness layer and of the network architecture.
\par\smallskip
\noindent\textbf{Index Terms---}Rumor propagation, 
Fractional-order SIS-UAU model, 
Caputo fractional derivative, 
Multiplex network.
\par}
\vspace{2mm}
\par\vspace*{5ex}


%
\IEEEpeerreviewmaketitle

\IEEEraisesectionheading{\section{Model Summary and Recalled Equations}
\label{sec:model_summary}
}
\IEEEPARstart{F}{or} the convenience of readers, we summarize the fractional-order SIS-UAU model and recall some important equations from the main manuscript that will be referenced in the subsequent proofs.
\subsection{Linear Caputo Fractional-Order System}
We recall the linear single-order Caputo fractional system
\begin{equation}\label{eq:linear-single-order_sm}
{}^{C}\mathcal D_q y(t)=A\,y(t),\qquad q\in(0,1],\;A\in\mathbb{R}^{n\times n},
\end{equation}
with initial condition $y(0)=y_0\in\mathbb{R}^n$, where ${}^{C}\mathcal D_q$
denotes the Caputo derivative applied component-wise.
\subsection{Fractional-Order SIS-UAU Model}
The fractional-order SIS-UAU model studied in the main manuscript is governed by
\begin{align}
{}^C_0\mathcal{D}^\alpha_t\mathcal{P}(t)
=
\mathcal{P}(t)\bigoplus_{i=1}^N\mathcal{Q}_i(t),
\label{eq:main_sm}
\end{align}
where $\mathcal{P}(t)\in\Omega$ denotes the state probability vector and
$\alpha\in(0,1]$ is the fractional order.

The component-wise representation of the fractional-order SIS-UAU model
\eqref{eq:main_sm} is given, for each node $i=1,\dots,N$, by
\begin{align}
  \begin{cases}
    {}^C_0\mathcal{D}^\alpha_t p_i^{US}(t)
    =-\left(\lambda M_i^A(t)+\beta^U M_i^I(t)\right)p_i^{US}(t)
    \\[1mm]
    \quad +\delta p_i^{AS}(t),\\[2mm]
    {}^C_0\mathcal{D}^\alpha_t p_i^{AS}(t)
    =\left(\lambda M_i^A(t)\right)p_i^{US}(t)
    \\[1mm]
    \quad -\left(\delta+\beta^A M_i^I(t)\right)p_i^{AS}(t)
    +\mu p_i^{AI}(t),\\[2mm]
    {}^C_0\mathcal{D}^\alpha_t p_i^{AI}(t)
    =\left(\beta^U M_i^I(t)\right)p_i^{US}(t)
    +\left(\beta^A M_i^I(t)\right)p_i^{AS}(t)\\
    \quad -\mu p_i^{AI}(t),
  \end{cases}
  \label{eq:main_component_sm}
\end{align}
where $M_i^A(t)$ and $M_i^I(t)$ denote the awareness and infection pressures,
respectively.

\subsection{Basic Reproduction Number}
The basic reproduction number is defined as
\begin{align}
\label{eq:R0_def_sm}
\mathcal R_0
:=
\frac{1}{\mu}\,
\rho\!\left(
\operatorname{diag}(p^{US*}_i\beta^U + p^{AS*}_i\beta^A)\mathbf B
\right),
\end{align}
where $\rho(\cdot)$ denotes the spectral radius.
\subsection{Self-Consistent and Linearized Endemic Relations}
At equilibrium, eliminating the awareness variables yields the self-consistent system
\begin{align}
\label{eq:self_consistent_sm}
p_i^{U**}
&=\frac{\delta(1-p_i^{I**})}
{\delta+\lambda M_i^{A**}+\beta^U M_i^{I**}},\\
p_i^{A**}&=1-p_i^{U**},\\
\big(\mu+\beta^A M_i^{I**}\big)p_i^{I**}
&=M_i^{I**}\big(\beta^U p_i^{U**}+\beta^A p_i^{A**}\big).
\end{align}
Its first-order approximation around the rumor-free equilibrium is
\begin{align}
\label{eq:linearized_endemic_sm}
p_i^{AI**}
\approx
\frac{1}{\mu}
\big(p_i^{US*}\beta^U + p_i^{AS*}\beta^A\big)
\sum_{j=1}^N b_{ij}p_j^{AI**},
\end{align}
or, in vector form,
\begin{align}
\label{eq:vector_linearized_endemic_sm}
\mathbf p^{AI**}
\approx
\frac{1}{\mu}\,
\operatorname{diag}(p^{US*}_i\beta^U + p^{AS*}_i\beta^A)\,
\mathbf B\,\mathbf p^{AI**}
=
\mathcal K\,\mathbf p^{AI**},
\end{align}
where $\mathcal K$ denotes the next-generation matrix. The rigorous analysis of the nonlinear relation \eqref{eq:self_consistent_sm} is carried out in Section~\ref{appendix:proof_endemic_existence_sm}.
\subsection{Jacobian Matrix at the Rumor-Endemic Equilibrium}

To analyze the local stability of the rumor-endemic equilibrium
$\mathcal P^{**}=(\mathcal P^{US**},\mathcal P^{AS**},\mathcal P^{AI**})$,
we consider the Jacobian matrix obtained by linearizing the
component-wise system~\eqref{eq:main_component_sm} at $\mathcal P^{**}$.

For notational convenience, we introduce the diagonal matrices
\begin{align}
  D_{u} &:= \mathrm{diag}(\mathcal P^{US**}), &
  D_{a} &:= \mathrm{diag}(\mathcal P^{AS**}), \nonumber\\
  D_{M^A} &:= \mathrm{diag}(\mathbf M^{A**}), &
  D_{M^I} &:= \mathrm{diag}(\mathbf M^{I**}),
\end{align}
where
\begin{align}
  \mathbf M^{A**} := \mathbf A(\mathcal P^{AS**}+\mathcal P^{AI**}),
  \qquad
  \mathbf M^{I**} := \mathbf B\,\mathcal P^{AI**}.
\end{align}

The resulting Jacobian matrix admits the compact block representation
\begin{align}
  \mathcal J^{*}
  = \sum_{m=1}^{3}\sum_{n=1}^{3} E_{mn}\otimes B_{mn},
  \label{eq:Jstar_compact_text_sm}
\end{align}
where $E_{mn}$ denote the canonical basis matrices in $\mathbb R^{3\times3}$,
and the $N\times N$ block matrices $B_{mn}$ are given by
\begin{align}
  \begin{cases}
    B_{11} = -(\lambda D_{M^A} + \beta^U D_{M^I}), \\
    B_{12} = -\lambda D_u\mathbf A + \delta\mathbf I_N, \\
    B_{13} = -\lambda D_u\mathbf A - \beta^U D_u\mathbf B + \mu\mathbf I_N, \\
    B_{21} = \lambda D_{M^A}, \\
    B_{22} = \lambda D_u\mathbf A - (\delta\mathbf I_N + \beta^A D_{M^I}), \\
    B_{23} = \lambda D_u\mathbf A - \beta^A D_a\mathbf B, \\
    B_{31} = \beta^U D_{M^I}, \\
    B_{32} = \beta^A D_{M^I}, \\
    B_{33} = (\beta^U D_u + \beta^A D_a)\mathbf B - \mu\mathbf I_N.
  \end{cases}
  \label{eq:Bmn_compact_text_sm}
\end{align}

This compact formulation highlights that $\mathcal J^{*}$ depends only on
a small set of fundamental matrices and parameters, which facilitates both
theoretical analysis and numerical implementation.

\section{Mathematical Preliminaries}
\label{sec:preliminaries}
This section provides some essential lemmas on fractional calculus and some properties of the SIS-UAU model that will be utilized in the subsequent proofs of the main theoretical results.

\subsection{Lemma on Short-Memory Property of Caputo Derivative}
\label{lemma:short_memory_section}

\begin{lemma}
  \label{lemma:caputo_limit_sm}
Let $p(t) \in C^1([0,\infty))$ be a positive function such that
\begin{align}
p(t) = \mathcal{O}(1), \quad p'(t) = o\left(\frac{1}{t^{1 - \alpha}}\right), \quad \forall \alpha \in (0,1).
\end{align}
Then,
\begin{align}
\lim_{t \to \infty} {}^C_0 \mathcal{D}_t^\alpha p(t) = 0.
\end{align}
\end{lemma}
\begin{proof}
Recall the definition of the Caputo fractional derivative of order $\alpha \in (0,1)$:
\begin{align}
{}^C_0 \mathcal{D}_t^\alpha p(t) = \frac{1}{\Gamma(1 - \alpha)} \int_0^t \frac{p'(\tau)}{(t - \tau)^\alpha} \, d\tau.
\end{align}
We divide the integral domain into two parts:

\begin{align}
\int_0^t = \int_0^{t^\alpha} + \int_{t^\alpha}^t.
\end{align}
Part I: $[0, t^\alpha]$  
By the first mean value theorem for integrals, there exists $\xi \in (0, t^\alpha)$ such that:
\begin{align}
\int_0^{t^\alpha} \frac{p'(\tau)}{(t - \tau)^\alpha} \, d\tau = \frac{1}{(t - \xi)^\alpha} \int_0^{t^\alpha} p'(\tau) \, d\tau = \frac{p(t^\alpha) - p(0)}{(t - \xi)^\alpha}.
\end{align}
Since $p(t)$ is bounded and $t^\alpha \to \infty$ as $t \to \infty$, we have $p(t^\alpha) - p(0) = \mathcal{O}(1)$, and $(t - \xi)^\alpha \to \infty$. Therefore,
\begin{align}
\lim_{t \to \infty} \int_0^{t^\alpha} \frac{p'(\tau)}{(t - \tau)^\alpha} \, d\tau = 0.
\end{align}
Part II: $[t^\alpha, t]$  
We apply the first mean value theorem for integrals in a more refined way. Since $p(t) \in C^1([0, \infty))$, and $\tau^{1 - \alpha} > 0$, the function
\begin{align}
f(\tau) := \frac{p'(\tau)}{\tau^{1 - \alpha}},
\end{align}
is continuous and integrable on $[t^\alpha, t]$. Hence, there exists $\xi \in (t^\alpha, t)$ such that:
\begin{align}
\int_{t^\alpha}^t \frac{p'(\tau)}{(t - \tau)^\alpha} \, d\tau = \xi^{1 - \alpha} p'(\xi) \int_{t^\alpha}^t \frac{1}{\tau^{1 - \alpha} (t - \tau)^\alpha} \, d\tau.
\end{align}
We now estimate the remaining integral. Since both $\tau \in [t^\alpha, t] \subset [0, t]$, we have:

\begin{align}
\int_{t^\alpha}^t \frac{1}{\tau^{1 - \alpha} (t - \tau)^\alpha} \, d\tau \leq \int_0^t \frac{1}{\tau^{1 - \alpha} (t - \tau)^\alpha} \, d\tau.
\end{align}
We recognize the right-hand side as a classical Beta integral:
\begin{align}
\int_0^t \frac{1}{\tau^{1 - \alpha} (t - \tau)^\alpha} \, d\tau &= t^{0} \cdot B(\alpha, 1 - \alpha) \notag\\
&= \frac{\pi}{\sin(\pi \alpha)} = \pi \csc(\pi \alpha).
\end{align}
This is a finite constant for any $\alpha \in (0,1)$. Hence, we obtain:
\begin{align}
\left| \int_{t^\alpha}^t \frac{p'(\tau)}{(t - \tau)^\alpha} \, d\tau \right| \leq \pi \csc(\pi \alpha) \cdot \xi^{1 - \alpha} |p'(\xi)|.
\end{align}
By assumption, $p'(t) = o\left( \dfrac{1}{t^{1 - \alpha}} \right)$, and since $\xi \to \infty$ as $t \to \infty$, we deduce that:
\begin{align}
\xi^{1 - \alpha} |p'(\xi)| \to 0,
\end{align}
which implies:
\begin{align}
\lim_{t \to \infty} \int_{t^\alpha}^t \frac{p'(\tau)}{(t - \tau)^\alpha} \, d\tau = 0.
\end{align}
This completes the proof.

Conclusion:  
Combining both parts, we obtain:
\begin{align}
\lim_{t \to \infty} {}^C_0 \mathcal{D}_t^\alpha p(t) = 0.
\end{align}
\end{proof}

\subsection{Lemma on Lipschitz Continuity}
\label{lemma:lipschitz_section}
\begin{lemma}
    \label{lemma:lipschitz_sm}
The function $F: \mathbb{R}^{N \times 3} \to \mathbb{R}^{N \times 3}$ defined by
\begin{align}
F(\mathcal{P}(t)) := \mathcal{P}(t) \bigoplus_{i=1}^N \mathcal{Q}_i(t),
\end{align}
is Lipschitz continuous on the domain $\Omega := \{ \mathcal{P}\in \mathbb{R}^{N \times 3} | p_i^{US}, p_i^{AS}, p_i^{AI} \geq 0, p_i^{US} + p_i^{AS} + p_i^{AI} = 1, \forall i\in \mathbb{Z}[1,N] \}$. 
That is, there exists a constant $L > 0$ such that for all $\mathcal{P}, \tilde{\mathcal{P}} \in \Omega$,
\begin{align}
\| F(\mathcal{P}) - F(\tilde{\mathcal{P}}) \| \leq L \| \mathcal{P} - \tilde{\mathcal{P}} \|.
\end{align}
\end{lemma}
\begin{proof}
  For matrix/vector norms, we adopt the element-wise maximum norm, that is, for any matrix $\mathbf{X} = \{ x_{ij} \}$ and vector $\mathbf{x} = (x_1, x_2, \ldots, x_n)$, we define:
\begin{align}
\| \mathbf{X} \|_{\max} := \max_{i,j} |x_{ij}| , \quad \| \mathbf{x} \|_\infty := \max_i |x_i|.
\end{align}

For convenience, let:
\begin{align}
  \mathcal{Q}_i(\mathcal P) =&\; \underbrace{\begin{pmatrix} 0&0&0\\ \delta&-\delta&0\\ 0&\mu&-\mu \end{pmatrix}}_{=:Q^{(0)}_i}\notag\\
  &+ \underbrace{\mathbf{e}_i^T\mathbf{A}\mathcal{P}\bigl(0,1,1\bigr)^T}_{=:M_i^A(\mathcal P)}\underbrace{\begin{pmatrix} -\lambda&\lambda&0\\ 0&0&0\\ 0&0&0 \end{pmatrix}}_{=:Q^{(A)}_i}\notag\\
  &+ \underbrace{\mathbf{e}_i^T\mathbf{B}\mathcal{P}\bigl(0,0,1\bigr)^T}_{=:M_i^I(\mathcal P)}\underbrace{\begin{pmatrix} -\beta^U&0&\beta^U\\ 0&-\beta^A&\beta^A\\ 0&0&0 \end{pmatrix}}_{=:Q^{(I)}_i}.
\end{align}

To divide $\mathcal P$ by lines, we only need to prove: $\forall \mathcal{P}, \tilde{\mathcal{P}} \in \Omega,$
\begin{align}
\| \mathcal{P}_i \mathcal{Q}_i(\mathcal{P}) - \tilde{\mathcal{P}}_i \mathcal{Q}_i(\tilde{\mathcal{P}}) \|_{\infty} \leq L_i \| \mathcal{P} - \tilde{\mathcal{P}} \|_{\max}, \forall i\in \mathbb{Z}[1,N].
\end{align}

We have:
\begin{align}
  \mathcal Q_i(\mathcal P) - \mathcal Q_i(\tilde{\mathcal P}) =& (M_i^A(\mathcal P) - M_i^A(\tilde{\mathcal P})) Q^{(A)}_i \notag\\
  &+ (M_i^I(\mathcal P) - M_i^I(\tilde{\mathcal P})) Q^{(I)}_i,
\end{align}
and where:
\begin{align}
  |M_i^A(\mathcal P) - M_i^A(\tilde{\mathcal P})| &\leq \|\mathbf{e}_i^T \mathbf{A}\|_1 \cdot \| \mathcal P - \tilde{\mathcal P} \|_{\max} \cdot \left\| \begin{pmatrix} 0,1,1 \end{pmatrix}^T \right\|_1 \notag\\
  &\leq 2N \cdot \| \mathcal P - \tilde{\mathcal P} \|_{\max}\\
  |M_i^I(\mathcal P) - M_i^I(\tilde{\mathcal P})| &\leq \|\mathbf{e}_i^T \mathbf{B}\|_1 \cdot \| \mathcal P - \tilde{\mathcal P} \|_{\max} \cdot \left\| \begin{pmatrix} 0,0,1 \end{pmatrix}^T \right\|_1 \notag\\
  &\leq N \cdot \| \mathcal P - \tilde{\mathcal P} \|_{\max}.
\end{align}
Here, $\|\mathbf x\|_1$ denotes the vector 1-norm, i.e., the sum of absolute values of its elements.
Therefore, we have:
\begin{align}
  &\| \mathcal{Q}_i(\mathcal{P}) - \mathcal{Q}_i(\tilde{\mathcal{P}}) \|_{\infty} \notag\\
  &\leq \left(\| \mathcal Q^{(A)}_i \|_{\max} \cdot 2N + \| \mathcal Q^{(I)}_i \|_{\max} \cdot N\right) \cdot \| \mathcal P - \tilde{\mathcal P} \|_{\max}\notag\\
  &\leq \left(2N\lambda + N\beta^U\right) \cdot \| \mathcal P - \tilde{\mathcal P} \|_{\max}.
\end{align}

Then, let $F_i(\mathcal P) := \mathcal P_i \mathcal Q_i(\mathcal P)$, we have:
\begin{align}
  &\| F_i(\mathcal P) - F_i(\tilde{\mathcal P}) \|_{\infty} = \| \mathcal P_i \mathcal Q_i(\mathcal P) - \tilde{\mathcal P}_i \mathcal Q_i(\tilde{\mathcal P}) \|_{\infty}\notag\\
  &\leq \| \mathcal P_i (\mathcal Q_i(\mathcal P) - \mathcal Q_i(\tilde{\mathcal P})) \|_{\infty} + \| (\mathcal P_i - \tilde{\mathcal P}_i) \mathcal Q_i(\tilde{\mathcal P}) \|_{\infty} \notag\\
  &\leq \| \mathcal P_i \|_{1} \cdot \| \mathcal Q_i(\mathcal P) - \mathcal Q_i(\tilde{\mathcal P}) \|_{\max} + \| \mathcal P_i - \tilde{\mathcal P}_i \|_{\infty} \cdot \| \mathcal Q_i(\tilde{\mathcal P}) \|_{1}.
\end{align}
Since $\mathcal P \in \Omega$, we have $\| \mathcal P_i \|_1 = p_i^{US} + p_i^{AS} + p_i^{AI} = 1$. Also, we can bound $\| \mathcal Q_i(\tilde{\mathcal P}) \|_1$ as follows:
\begin{align}
  \| \mathcal Q_i(\tilde{\mathcal P}) \|_1 \leq& \| Q^{(0)}_i \|_1 + |M_i^A(\tilde{\mathcal P})| \cdot \| Q^{(A)}_i \|_1 \notag\\
  &+ |M_i^I(\tilde{\mathcal P})| \cdot \| Q^{(I)}_i \|_1 \leq \delta + \mu + 2N \lambda + 2N \beta^U.
\end{align}
Combining the above inequalities, we obtain:
\begin{align}
  \| F_i(\mathcal P) - F_i(\tilde{\mathcal P}) \|_{\infty} \leq L_i \cdot \| \mathcal P - \tilde{\mathcal P} \|_{\max},
\end{align}
where
\begin{align}
L_i := 1 + 4N\lambda + 3N\beta^U + \delta + \mu.
\end{align}
Finally, by taking $L := \max_{1 \leq i \leq N} L_i$, we conclude that $F$ is Lipschitz continuous on $\Omega$.
\end{proof}

\subsection{Matignon-Type Stability Criterion}
\label{matignon_stability_section}
\begin{lemma}\cite{brandibur2021stability}
    \label{lemma:matignon_sm}
\textit{The linear system \eqref{eq:linear-single-order_sm} is}

i. \(\;\mathcal O(t^{-q})\)-\textit{asymptotically stable if and only if}
\[
\sigma(A)\subset S_q=\{\lambda\in\mathbb{C}:\;|\arg(\lambda)|>q\pi/2\}.
\]

ii. \textit{stable if and only if } \(\sigma(A)\subset S_q\) \textit{and the eigenvalues of }\(A\)
\textit{which satisfy } \(|\arg(\lambda)|=q\pi/2\) \textit{have the index 1.}
\end{lemma}

\section{Proof of Theorem \ref{theorem:existence_uniqueness_sm}}
\label{proof:theorem_existence_uniqueness}
\begin{theorem}\label{theorem:existence_uniqueness_sm}
  (Existence and uniqueness of the solution)
Consider the fractional-order SIS-UAU model defined by equation~\eqref{eq:main_sm} with the initial condition $\mathcal{P}(0) = \mathcal{P}_0 \in \Omega$. Then, there exists a unique solution $\mathcal{P}(t)$ defined on $[0, T]$ for any $T \in \left(0,\left(\frac{\Gamma(1+\alpha)}{L}\right)^{1/\alpha}\right)$.
\end{theorem}
\begin{proof}
  We rewrite equation~\eqref{eq:main_sm} in its equivalent integral form:
\begin{align}
\mathcal{P}(t) = \mathcal{P}_0 + \frac{1}{\Gamma(\alpha)} \int_0^t (t - \tau)^{\alpha - 1} F(\mathcal{P}(\tau)) \, d\tau,
\end{align}
where $F(\mathcal{P}(t)) := \mathcal{P}(t) \bigoplus_{i=1}^N \mathcal{Q}_i(t)$.
We define the operator $\mathcal{T}$ on the space $C([0, T], \Omega)$ as follows:
\begin{align}
\mathcal{T}(\mathcal{P})(t) := \mathcal{P}_0 + \frac{1}{\Gamma(\alpha)} \int_0^t (t - \tau)^{\alpha - 1} F(\mathcal{P}(\tau)) \, d\tau.
\end{align}
To prove the existence and uniqueness of the solution, we will show that $\mathcal{T}$ is a contraction mapping on $C([0, T], \Omega)$ for sufficiently small $T$.
Let $\mathcal{P}, \tilde{\mathcal{P}} \in C([0, T], \Omega)$. We have:
\begin{align}
&\| \mathcal{T}(\mathcal{P})(t) - \mathcal{T}(\tilde{\mathcal{P}})(t) \|_{C([0, T], \Omega)} \notag\\
&\leq \frac{1}{\Gamma(\alpha)} \int_0^t (t - \tau)^{\alpha - 1} \| F(\mathcal{P}(\tau)) - F(\tilde{\mathcal{P}}(\tau)) \|_{C([0, T], \Omega)} \, d\tau \notag\\
&\leq \frac{L}{\Gamma(\alpha)} \int_0^t (t - \tau)^{\alpha - 1} \| \mathcal{P}(\tau) - \tilde{\mathcal{P}}(\tau) \|_{C([0, T], \Omega)} \, d\tau \notag\\
&\leq \frac{L}{\Gamma(\alpha)} \| \mathcal{P} - \tilde{\mathcal{P}} \|_{C([0, T], \Omega)} \int_0^t (t - \tau)^{\alpha - 1} \, d\tau \notag\\
&= \frac{L T^\alpha}{\Gamma(1 + \alpha)} \| \mathcal{P} - \tilde{\mathcal{P}} \|_{C([0, T], \Omega)}.
\end{align}
By choosing $T$ such that $\frac{L T^\alpha}{\Gamma(1 + \alpha)} < 1$, we ensure that $\mathcal{T}$ is a contraction mapping. Therefore, by the Banach fixed-point theorem, there exists a unique fixed point $\mathcal{P}(t)$ in $C([0, T], \Omega)$, which is the unique solution to equation~\eqref{eq:main_sm} on the interval $[0, T]$.
\end{proof}

\section{Proof of Theorem \ref{theorem:positivity_boundedness_sm}}
\label{proof:theorem_positivity_boundedness}
\begin{theorem}\label{theorem:positivity_boundedness_sm}
  (Positivity and boundedness of the solution)
Consider the fractional-order SIS-UAU model defined by equation~\eqref{eq:main_sm} with the initial condition $\mathcal{P}(0) = \mathcal{P}_0 \in \Omega$. Then, the solution $\mathcal{P}(t)$ satisfies $p_i^{US}(t), p_i^{AS}(t), p_i^{AI}(t) \geq 0$ and $p_i^{US}(t) + p_i^{AS}(t) + p_i^{AI}(t) = 1$ for all $t \geq 0$ and for all $i\in \mathbb{Z}[1,N]$.
\end{theorem}
\begin{proof}
We first show that the three probabilities sum to one. Adding the three equations in \eqref{eq:main_component_sm} for a fixed $i$, the right-hand side vanishes identically, hence
\[
{}^C_0\mathcal D_t^\alpha\bigl(p_i^{US}(t)+p_i^{AS}(t)+p_i^{AI}(t)\bigr)=0 .
\]
Writing this as an equivalent integral equation gives
\[
\begin{aligned}
p_i^{US}(t)+p_i^{AS}(t)+p_i^{AI}(t)
&= p_i^{US}(0)+p_i^{AS}(0)+p_i^{AI}(0)\\
&\quad +\frac{1}{\Gamma(\alpha)}\int_0^t (t-\tau)^{\alpha-1}\cdot 0\,d\tau,
\end{aligned}
\]
in which the integral term is zero; together with the initial condition $p_i^{US}(0)+p_i^{AS}(0)+p_i^{AI}(0)=1$ this yields
\[
p_i^{US}(t)+p_i^{AS}(t)+p_i^{AI}(t)=1 .
\]

Next we prove that every component of the solution is nonnegative. Let $F(\mathcal P)$ denote the right-hand side of the original system. For $\varepsilon>0$ consider the auxiliary system
\begin{equation}\label{eq:aux_system}
{}^C_0\mathcal D_t^\alpha\mathcal P^{\varepsilon}(t)=F(\mathcal P^{\varepsilon}(t))+\varepsilon\mathbf 1,\qquad
\mathcal P^{\varepsilon}(0)=\mathcal P_0,
\end{equation}
where $\mathbf 1\in\mathbb R^{3N}$ is the all-ones vector. Since $F$ is locally Lipschitz, local existence and uniqueness for the auxiliary system follow exactly as in Theorem~\ref{theorem:existence_uniqueness_sm}.

We claim that every component of $\mathcal P^{\varepsilon}$ is nonnegative on its interval of existence. Suppose that some component first becomes negative, and set
\begin{multline*}
\sigma=\inf\Bigl\{t>0:\ \exists\,i\in\mathbb Z[1,N],\ X\in\{US,AS,AI\}\\
\text{such that}\ p_i^{X,\varepsilon}(t)<0\Bigr\},
\end{multline*}
choosing a component $p_i^{X,\varepsilon}(t)$ which is negative on the first connected component $(\sigma,\rho)$ of the negative set. Then $p_i^{X,\varepsilon}(\sigma)=0$, and all components are nonnegative on $[0,\sigma]$.

If $p_i^{X,\varepsilon}$ is not identically zero on $[0,\sigma]$, then since $p_i^{X,\varepsilon}\in C^1$, $p_i^{X,\varepsilon}\ge0$ on $[0,\sigma]$, and $p_i^{X,\varepsilon}(\sigma)=0$, integration by parts gives
\begin{multline*}
{}^C_0\mathcal D_t^\alpha p_i^{X,\varepsilon}(\sigma)
=\frac{1}{\Gamma(1-\alpha)}\int_0^\sigma(\sigma-s)^{-\alpha}\,\frac{d}{ds}p_i^{X,\varepsilon}(s)\,ds\\
=-\frac{1}{\Gamma(1-\alpha)}
\Bigl[\frac{p_i^{X,\varepsilon}(0)}{\sigma^\alpha}\\
+\alpha\int_0^\sigma p_i^{X,\varepsilon}(s)(\sigma-s)^{-\alpha-1}\,ds\Bigr]<0 .
\end{multline*}
On the other hand, at time $\sigma$ the right-hand side of \eqref{eq:aux_system} is $F_i^X(\mathcal P^{\varepsilon}(\sigma))+\varepsilon$. Since all components are nonnegative at $\sigma$ and the chosen one is zero, for the three possible states we have
\begin{align*}
X=US:\quad &F_i^{US}=\delta p_i^{AS,\varepsilon}\ge0,\\
X=AS:\quad &F_i^{AS}=\lambda M_i^A p_i^{US,\varepsilon}+\mu p_i^{AI,\varepsilon}\ge0,\\
X=AI:\quad &F_i^{AI}=\beta^U M_i^I p_i^{US,\varepsilon}+\beta^A M_i^I p_i^{AS,\varepsilon}\ge0,
\end{align*}
where $M_i^A$ and $M_i^I$ are nonnegative linear combinations of the remaining nonnegative components (the adjacency matrices have no self-loops). Hence the right-hand side is at least $\varepsilon>0$, which contradicts the strict negativity above.

If $p_i^{X,\varepsilon}$ is identically zero on $[0,\sigma]$, then its Caputo derivative vanishes on $(0,\sigma)$, so the right-hand side of \eqref{eq:aux_system} must vanish on $(0,\sigma)$; but the same inspection shows that it is at least $\varepsilon>0$, again a contradiction. Therefore every component of $\mathcal P^{\varepsilon}$ is nonnegative on its interval of existence.

We now give an a priori bound for the auxiliary system and extend it. Adding the three component equations of \eqref{eq:aux_system}, the right-hand sides sum to $3N\varepsilon$, so
\[
{}^C_0\mathcal D_t^\alpha\sum_{i=1}^N\sum_{X\in\{US,AS,AI\}}p_i^{X,\varepsilon}(t)=3N\varepsilon,
\]
and consequently
\[
\sum_{i=1}^N\sum_{X\in\{US,AS,AI\}}p_i^{X,\varepsilon}(t)
=1+\frac{3N\varepsilon}{\Gamma(1+\alpha)}\,t^\alpha .
\]
Together with nonnegativity, every component is uniformly bounded on any fixed interval $[0,T]$; since $F$ is Lipschitz on this compact set, the standard continuation argument ensures that $\mathcal P^{\varepsilon}$ exists on every $[0,T]$, with a bound uniform for $\varepsilon\in(0,\varepsilon_0]$.

Next we pass to the limit $\varepsilon\to0$. Let $\mathcal P$ be the solution of the original system on $[0,T]$. Subtracting the two integral equations and applying the fractional Gr\"onwall inequality in the $\ell^1$ norm yields
\[
\|\mathcal P^{\varepsilon}(t)-\mathcal P(t)\|
\le \frac{3N\varepsilon\,t^\alpha}{\Gamma(1+\alpha)}\,E_\alpha(Lt^\alpha),\qquad t\in[0,T],
\]
where $E_\alpha(z)=\sum_{k=0}^{\infty}z^k/\Gamma(1+\alpha k)$ is the Mittag--Leffler function and $L$ is a Lipschitz constant of $F$ on a common compact set. Thus $\mathcal P^{\varepsilon}\to\mathcal P$ uniformly on $[0,T]$; since every $\mathcal P^{\varepsilon}$ is nonnegative, we get $\mathcal P\ge0$. Combined with the conservation identity this gives $0\le p_i^X\le1$, and $T>0$ is arbitrary.

Finally, the solution of the original system stays in the compact simplex $\Omega$, on which $F$ is uniformly Lipschitz and bounded; the standard continuation argument extends the local solution to $[0,\infty)$. Hence nonnegativity, boundedness, and the conservation identity hold for all $t\ge0$. For the classical case $\alpha=1$, the integration-by-parts formula above is replaced by the elementary argument that a $C^1$ function which first reaches zero from nonnegative values satisfies $p_i^{X\prime}(\sigma)\le0$; the rest of the proof is identical.
\end{proof}

\section{Proof of Theorem \ref{theorem:stability_sm}}
\label{proof:theorem_stability}
\begin{theorem}\label{theorem:stability_sm}
  (Local asymptotic stability of the rumor-free equilibrium)
Consider the fractional-order SIS-UAU model defined by equation~\eqref{eq:main_sm} with the initial condition $\mathcal{P}(0)=\mathcal{P}_0\in\Omega$. 
The rumor-free equilibrium (RFE)
$\mathcal{P}^{*} = \begin{pmatrix} \mathcal P^{US*}, & \mathcal P^{AS*}, & \mathbf 0 \end{pmatrix}$
is locally asymptotically stable if every eigenvalue $\lambda$ of the matrix
\begin{align}
  \mathcal J \;:=\; \operatorname{diag}(p^{US*}_i\beta^U + p^{AS*}_i\beta^A)\,\mathbf B \;-\; \mu\,\mathbf I_N,
\end{align}
satisfies
\begin{align}
  \big|\arg(\lambda)\big| \;>\; \frac{\alpha\pi}{2}.
\end{align}
\end{theorem}
\begin{proof}
We work with the componentwise form of system \eqref{eq:main_sm}. For each node \(i\) the infection equation (the third column of \(\mathcal P\)) reads
\begin{equation}\label{eq:ai_dynamics_sm}
{}^C_0\mathcal D_t^\alpha p_i^{AI}(t)
= \beta^U M_i^I(t)\,p_i^{US}(t) + \beta^A M_i^I(t)\,p_i^{AS}(t) - \mu\, p_i^{AI}(t),
\end{equation}
The awareness and susceptible components satisfy
\begin{align}
{}^C_0\mathcal D_t^\alpha p_i^{US}(t)
=& -\big(\lambda M_i^A(t)+\beta^U M_i^I(t)\big)\,p_i^{US}(t)\notag\\
& + \delta\, p_i^{AS}(t),\label{eq:us_dynamics_sm}\\
{}^C_0\mathcal D_t^\alpha p_i^{AS}(t)
=& \lambda M_i^A(t)\, p_i^{US}(t) -\big(\delta+\beta^A M_i^I(t)\big)\, p_i^{AS}(t)\notag\\
& + \mu\, p_i^{AI}(t),\label{eq:as_dynamics_sm}
\end{align}

Let \(\mathcal P^{*}=(\mathcal P^{US*},\mathcal P^{AS*},\mathcal P^{AI*})\) denote the RFE, so that \(\mathcal P^{AI*}=\mathbf 0\), \(M_i^{I*}=0\) and \(\mathcal P^{US*}=\mathbf 1-\mathcal P^{AS*}\). We linearize the full system at the RFE. To obtain the linearization we compute the first-order partial derivatives of the right-hand side of \eqref{eq:ai_dynamics_sm} with respect to the perturbations in the three components.

Since \(M_i^I(t)\) is linear in the infection variables \(p_j^{AI}\), and since at the RFE \(M_i^{I*}=0\), the derivatives of the infection right-hand side at RFE are:
\begin{itemize}
  \item derivative with respect to \(p_j^{AI}\):
  \begin{align}
    \frac{\partial}{\partial p_j^{AI}}\Big( \beta^U M_i^I\,p_i^{US} + \beta^A M_i^I\,p_i^{AS} - \mu p_i^{AI} \Big)\Bigg|_{\text{RFE}}\notag\\
  = b_{ij}\big(\beta^U p_i^{US*} + \beta^A p_i^{AS*}\big) - \mu\delta_{ij},
  \end{align}
  where \(\delta_{ij}\) is the Kronecker delta.
  \item derivative with respect to \(p_j^{AS}\):
  \begin{align}
    \frac{\partial}{\partial p_j^{AS}}\Big(\beta^U M_i^I\,p_i^{US} + \beta^A M_i^I\,p_i^{AS} - \mu p_i^{AI}\Big)\Bigg|_{\text{RFE}}\notag\\
  = \beta^A \,M_i^{I*}\,\delta_{ij} = 0.
  \end{align}
  \item derivative with respect to \(p_j^{US}\):
  \begin{align}
    \frac{\partial}{\partial p_j^{US}}\Big(\beta^U M_i^I\,p_i^{US} + \beta^A M_i^I\,p_i^{AS} - \mu p_i^{AI}\Big)\Bigg|_{\text{RFE}}\notag\\
  = \beta^U \,M_i^{I*}\,\delta_{ij} = 0.
  \end{align}
\end{itemize}
Thus, at the RFE the infection component's linearization depends only on the infection variables themselves (to first order), and there is no first-order coupling from perturbations in \(p^{US}\) or \(p^{AS}\) into the infection equations.

Collecting the coefficients for all \(i,j=1,\dots,N\), the linearized infection subsystem (written in matrix form for the third column \(\mathcal P^{AI}\)) is
\begin{equation}\label{eq:linear_AI_sm}
{}^C_0\mathcal D_t^\alpha \,\mathcal P^{AI}(t)
= \Big(\operatorname{diag}(p^{US*}_i\beta^U + p^{AS*}_i\beta^A)\,\mathbf B - \mu\mathbf I_N\Big)\,\mathcal P^{AI}(t).
\end{equation}
Therefore the Jacobian block acting on the infection vector is precisely
\begin{align}
  \mathcal J \;=\; \operatorname{diag}(p^{US*}_i\beta^U + p^{AS*}_i\beta^A)\,\mathbf B - \mu\mathbf I_N,
\end{align}
as stated in the theorem.

Eq.~\eqref{eq:linear_AI_sm} is a linear Caputo-type system of the form \({}^C_0D_t^\alpha y(t)=A y(t)\) with \(A=\mathcal J\). By lemma \ref{lemma:matignon_sm},
the trivial solution \(\mathcal P^{AI}\equiv\mathbf0\) of the linearized infection subsystem is (locally) asymptotically stable if every eigenvalue \(\lambda\) of \(\mathcal J\) satisfies
\begin{align}
  \big|\arg(\lambda)\big| > \frac{\alpha\pi}{2}.
\end{align}
Local asymptotic stability of the infection subsystem implies that small infection perturbations decay to zero while the awareness and susceptible components remain near their RFE values; hence the full nonlinear system is locally asymptotically stable at the RFE. This completes the proof.
\end{proof}

\section{Proof of Corollary~\ref{corollary:R0_less_1_sm}}\label{appendix:proof_R0_less_1_sm}
\begin{corollary}\label{corollary:R0_less_1_sm}
If the basic reproduction number defined in \eqref{eq:R0_def_sm} satisfies
\begin{align}\label{eq:R0_lt_1_sm}
  \mathcal R_0 \;<\; 1,
\end{align}
then the rumor-free equilibrium $\mathcal P^*$ is locally asymptotically stable for the fractional-order system \eqref{eq:main_sm} (for every $\alpha\in(0,1]$).
\end{corollary}

\begin{proof}
Recall the next-generation matrix
\begin{align}\label{eq:K_def_sm}
  \mathcal K \;=\; \frac{1}{\mu}\,\operatorname{diag}(p^{US*}_i\beta^U + p^{AS*}_i\beta^A)\,\mathbf B,
\end{align}
so that the Jacobian block at the RFE appearing in Theorem~\ref{theorem:stability_sm} can be written as
\begin{align}\label{eq:J_muK_sm}
  \mathcal J \;=\; \operatorname{diag}(p^{US*}_i\beta^U + p^{AS*}_i\beta^A)\,\mathbf B - \mu\mathbf I_N
  \;=\; \mu\big(\mathcal K - \mathbf I_N\big).
\end{align}

Let $\sigma$ be any eigenvalue of $\mathcal K$. The hypothesis $\mathcal R_0=\rho(\mathcal K)<1$ implies $|\sigma|<1$ for every eigenvalue $\sigma$. For any complex number $\sigma$ with $|\sigma|<1$ we have
\begin{align}\label{eq:real_part_negative_sm}
  \Re(\sigma-1) \le |\sigma|-1 < 0,
\end{align}
hence $\sigma-1$ lies strictly in the open left half-plane. Consequently every eigenvalue of $\mathcal J$, which takes the form
\begin{align}\label{eq:lambda_relation_sm}
  \lambda \;=\; \mu(\sigma-1),
\end{align}
has strictly negative real part. In particular the argument of each $\lambda$ satisfies
\begin{align}\label{eq:arg_bound_sm}
  |\arg(\lambda)| \;>\; \frac{\pi}{2}.
\end{align}

Since $\alpha\in(0,1]$ it follows that $\dfrac{\alpha\pi}{2}\le\dfrac{\pi}{2}$, and therefore the stronger relation \eqref{eq:arg_bound_sm} implies the Matignon spectral condition of Theorem~\ref{theorem:stability_sm}:
\begin{align}\label{eq:matignon_satisfied_sm}
  |\arg(\lambda)| \;>\; \frac{\alpha\pi}{2}\quad\text{for every eigenvalue }\lambda\text{ of }\mathcal J.
\end{align}
By Theorem~\ref{theorem:stability_sm} the RFE is thus locally asymptotically stable. This completes the proof.
\end{proof}
\section{Proof of Theorem~\ref{theorem:endemic_existence_sm}}\label{appendix:proof_endemic_existence_sm}
\begin{theorem}\label{theorem:endemic_existence_sm}
(Existence and uniqueness of the rumor-endemic equilibrium)  
Assume that $\mathbf A$ and $\mathbf B$ are nonnegative zero-diagonal irreducible matrices and that $0<\beta^A<\beta^U$. Then, for the fractional-order UAU-SIS system \eqref{eq:main_sm}, there exists a unique rumor-endemic equilibrium $\mathcal P^{**}$ with $\mathcal P^{AI**}\!>\!0$ if and only if the basic reproduction number satisfies
\begin{align}\label{eq:R0_gt_1_sm}
  \mathcal R_0 = \rho(\mathcal K) > 1.
\end{align}
Conversely, if $\mathcal R_0 \le 1$, the only equilibrium is the rumor-free equilibrium $\mathcal P^*$.
\end{theorem}
\begin{proof}
Throughout this proof we write
\begin{align}
p_i^{U**}&:=p_i^{US**},\qquad
p_i^{I**}:=p_i^{AI**},\\
p_i^{A**}&:=p_i^{AS**}+p_i^{AI**},
\end{align}
and
\begin{align}
M_i^{A**}:=\sum_j a_{ij}p_j^{A**},\qquad
M_i^{I**}:=\sum_j b_{ij}p_j^{I**}.
\end{align}
Since the model has no unaware infected state, \(p_i^{U**}+p_i^{A**}=1\) and \(p_i^{AS**}=p_i^{A**}-p_i^{I**}\).

At an equilibrium the US and the AI components satisfy
\begin{align}
\big(\lambda M_i^{A**}+\beta^U M_i^{I**}\big)p_i^{U**}
&=\delta\big(p_i^{A**}-p_i^{I**}\big), \label{eq:sm4_us}\\
\mu p_i^{I**}
&=M_i^{I**}\big(\beta^U p_i^{U**}+\beta^A p_i^{AS**}\big). \label{eq:sm4_ai}
\end{align}

\emph{The rumor-free equilibrium.} If \(p_i^{I**}=0\) for every \(i\), then \(M_i^{I**}=0\), and \eqref{eq:sm4_us} reduces to \(\lambda M_i^{A**}p_i^{U**}=\delta p_i^{A**}\). Together with \(p_i^{U**}+p_i^{A**}=1\), this gives \(p_i^{A**}=\lambda M_i^{A**}/(\delta+\lambda M_i^{A**})\), so that
\begin{align}
M_i^{A**}=\sum_j a_{ij}\frac{\lambda M_j^{A**}}{\delta+\lambda M_j^{A**}}.
\end{align}
The right-hand side is a continuous, monotone, strictly sublinear, bounded map on \(\mathbb R_+^N\); by the classical fixed-point theorem for monotone strictly sublinear maps, it has a unique nonnegative fixed point \(M^{A*}\), which is positive componentwise if \(\lambda\rho(\mathbf A)>\delta\) and is zero otherwise. Hence the rumor-free equilibrium is unique, with
\begin{align}
p_i^{U*}=\frac{\delta}{\delta+\lambda M_i^{A*}},\qquad
p_i^{A*}=\frac{\lambda M_i^{A*}}{\delta+\lambda M_i^{A*}},\qquad
p_i^{I*}=0.
\end{align}
The next-generation matrix is
\begin{align}
K_{ij}:=\frac1\mu\big(\beta^U p_i^{U*}+\beta^A p_i^{A*}\big)b_{ij},
\end{align}
and \(\mathcal R_0=\rho(K)\). Since \(\mathbf B\) is irreducible and \(\beta^U p_i^{U*}+\beta^A p_i^{A*}>0\), the matrix \(K\) is nonnegative and irreducible.

\emph{The awareness layer for a fixed infection vector.} Fix \(t\in[0,1]^N\), put \(M_i^{I**}(t):=\sum_j b_{ij}t_j\), and consider the self-consistent equation for the awareness pressure,
\begin{align}
M_i^{A**}
=\sum_j a_{ij}\,
\frac{\lambda M_j^{A**}+\beta^U M_j^{I**}(t)+\delta t_j}
{\delta+\lambda M_j^{A**}+\beta^U M_j^{I**}(t)}. \label{eq:sm4_aware}
\end{align}
For fixed \(t\), the right-hand side of \eqref{eq:sm4_aware} is a continuous map of \(M^{A**}\) that sends the box \(\prod_{i=1}^N[0,d_i]\), \(d_i:=\sum_j a_{ij}\), into itself; hence it has a fixed point by Brouwer's theorem. For \(t>0\) this fixed point is unique. Indeed, let \(M,M'\) be two fixed points. Since \(t>0\) and \(\mathbf B\) is irreducible, both are positive. Set \(m:=\min_i M_i/M'_i<1\) and choose \(k\) with \(M_k=mM'_k\). The function
\begin{align}
g_j(s):=\frac{\lambda s+\beta^U M_j^{I**}(t)+\delta t_j}
{\delta+\lambda s+\beta^U M_j^{I**}(t)},
\end{align}
is strictly increasing and strictly concave in \(s\ge0\), so
\begin{align}
M_k&=\sum_j a_{kj}g_j(M_j)\\
&\ge \sum_j a_{kj}g_j(mM'_j)\\
&>m\sum_j a_{kj}g_j(M'_j)=mM'_k=M_k,
\end{align}
a contradiction. Hence \eqref{eq:sm4_aware} has a unique solution \(M^{A**}(t)\) for \(t>0\), and we set \(M^{A**}(0):=M^{A*}\). By compactness and uniqueness, \(M^{A**}(t)\) is continuous in \(t\in(0,1]^N\), and it is nondecreasing in each component of \(t\) (the right-hand side of \eqref{eq:sm4_aware} is increasing in \(t\), and the unique fixed point is obtained as the increasing limit of the iteration starting from zero). Moreover, \(M^{A**}(t)\to M^{A*}\) as \(t\to0\): any cluster point is a fixed point of the limiting equation, hence belongs to \(\{0,M^{A*}\}\), and the alternative \(0\) is excluded whenever \(\lambda\rho(\mathbf A)>\delta\) by multiplying the linearized fixed-point equation near \(0\) by a positive left eigenvector of \(\mathbf A\). Consequently,
\begin{align}
p_i^{U**}(t)&:=\frac{\delta(1-t_i)}
{\delta+\lambda M_i^{A**}(t)+\beta^U M_i^{I**}(t)},\\
p_i^{A**}(t)&:=1-p_i^{U**}(t),
\end{align}
are continuous and nondecreasing in \(t\), and satisfy \(p_i^{A**}(t)\ge t_i\).

\emph{Reduction to a fixed-point equation.} Solving \eqref{eq:sm4_ai} for \(p_i^{I**}\) gives
\begin{align}
\big(\mu+\beta^A M_i^{I**}\big)p_i^{I**}
=M_i^{I**}\big(\beta^U p_i^{U**}+\beta^A p_i^{A**}\big).
\end{align}
Define the map \(\Phi:[0,1]^N\to\mathbb R_+^N\) by
\begin{align}
\Phi_i(t)
:=\frac{M_i^{I**}(t)\big(\beta^U p_i^{U**}(t)+\beta^A p_i^{A**}(t)\big)}
{\mu+\beta^A M_i^{I**}(t)}.
\end{align}
Then \(t\in[0,1]^N\) is an equilibrium infection vector if and only if \(t=\Phi(t)\).

We collect three properties of \(\Phi\). First, \(\Phi(0)=0\), and as \(t\to0\), \(M^{A**}(t)\to M^{A*}\), whence
\begin{align}
\Phi(t)=Kt+o(\|t\|). \label{eq:sm4_lin}
\end{align}
Second, if \(t_i=1\), then \(p_i^{A**}(t)\ge t_i=1\), so \(p_i^{A**}(t)=1\), and therefore
\begin{align}
\Phi_i(t)=\frac{M_i^{I**}(t)\beta^A}{\mu+\beta^A M_i^{I**}(t)}<1.
\end{align}
Third, \(\Phi\) is strictly sublinear: since \(p^{A**}(t)\) is nondecreasing in \(t\) and \(\beta^U>\beta^A\), for \(0<\theta<1\) and \(t>0\),
\begin{align}
\Phi_i(\theta t)
&\ge \frac{\theta M_i^{I**}(t)\big(\beta^U p_i^{U**}(t)+\beta^A p_i^{A**}(t)\big)}
{\mu+\beta^A\theta M_i^{I**}(t)} \notag\\
&=\theta\Phi_i(t)\,\frac{\mu+\beta^A M_i^{I**}(t)}{\mu+\beta^A\theta M_i^{I**}(t)}
>\theta\Phi_i(t). \label{eq:sm4_sub}
\end{align}

\emph{Existence for \(\mathcal R_0>1\).} Let \(v>0\) and \(w>0\) be the right and left Perron vectors of \(K\), so that \(Kv=\rho v\) and \(w^\top K=\rho w^\top\) with \(\rho=\mathcal R_0>1\). For \(\theta>0\) define
\begin{align}
T_\theta(t):=\min\big(\Phi(t)+\theta\mathbf 1,\mathbf 1\big),\qquad t\in[0,1]^N,
\end{align}
where \(\min\) is taken componentwise. By Brouwer's theorem there exists \(t_\theta\in[0,1]^N\) with \(T_\theta(t_\theta)=t_\theta\). Let \(\theta_n\downarrow0\) and assume \(t_{\theta_n}\to t_*\) along a subsequence. Then \(t_*=\min(\Phi(t_*),\mathbf 1)\). If \(t_{*,i}=1\) for some \(i\), then \(\Phi_i(t_*)<1\), contradicting \(t_{*,i}=\min(\Phi_i(t_*),1)=1\). Hence \(t_*\in[0,1)^N\) and \(t_*=\Phi(t_*)\).

It remains to rule out \(t_*=0\). If \(t_{\theta_n}\to0\), then for all large \(n\), \(\Phi(t_{\theta_n})+\theta_n\mathbf 1<\mathbf 1\), so \(t_{\theta_n}=\Phi(t_{\theta_n})+\theta_n\mathbf 1\). Using \eqref{eq:sm4_lin},
\begin{align}
w^\top t_{\theta_n}
=\rho\,w^\top t_{\theta_n}+o\big(\|t_{\theta_n}\|\big)+\theta_n w^\top\mathbf 1,
\end{align}
that is,
\begin{align}
(1-\rho)w^\top t_{\theta_n}
=o\big(\|t_{\theta_n}\|\big)+\theta_n w^\top\mathbf 1.
\end{align}
Since the left-hand side is at most \(-(\rho-1)(\min_i w_i)\|t_{\theta_n}\|_1\) and the \(o(\|t_{\theta_n}\|)\) term can be absorbed into half of this negative bound, the equality would force \(\theta_n w^\top\mathbf 1<0\) for all sufficiently large \(n\), contradicting \(\theta_n>0\). Thus \(t_*>0\), and \(t_*=\Phi(t_*)\) is a rumor-endemic equilibrium.

\emph{Nonexistence for \(\mathcal R_0\le1\).} Suppose \(t>0\) and \(t=\Phi(t)\). For \(0<s\le1\) put \(h(s):=\Phi(st)/s\). By \eqref{eq:sm4_sub}, \(h\) is strictly decreasing and \(h(s)>h(1)=t\) for \(s<1\). By \eqref{eq:sm4_lin}, \(h(s)\to Kt\) as \(s\to0\). Since \(h\) is decreasing,
\begin{align}
Kt=\lim_{s\to0}h(s)\ge h(s)>t,
\end{align}
componentwise, so \(Kt>t\). But then \(\rho(K)\ge\min_i(Kt)_i/t_i>1\), contradicting \(\mathcal R_0=\rho(K)\le1\). Hence no positive equilibrium exists when \(\mathcal R_0\le1\).

\emph{Uniqueness of the positive equilibrium.} It remains to show that there is at most one equilibrium with \(p_i^{I**}>0\) for every \(i\). Consider the vector field in the variables \((p^{I**},p^{A**})\):
\begin{align}
f^I_i
&:=M_i^{I**}\big(\beta^U p_i^{U**}+\beta^A(p_i^{A**}-p_i^{I**})\big)-\mu p_i^{I**},\\
f^A_i
&:=\big(1-p_i^{A**}\big)\big(\lambda M_i^{A**}+\beta^U M_i^{I**}\big)-\delta p_i^{A**}+\delta p_i^{I**},
\end{align}
where \(p_i^{U**}:=1-p_i^{A**}\). On the convex set \(\Omega:=\{(p^I,p^A)\in\mathbb R_+^{2N}:0\le p_i^I\le p_i^A\le1\}\), for \(j\ne i\),
\begin{align}
\frac{\partial f^I_i}{\partial p^I_j}
&=b_{ij}\big(\beta^U p_i^{U**}+\beta^A(p_i^{A**}-p_i^{I**})\big)\ge0,\\
\frac{\partial f^I_i}{\partial p^A_j}&=0,\\
\frac{\partial f^A_i}{\partial p^I_j}
&=(1-p_i^{A**})\beta^U b_{ij}\ge0,\\
\frac{\partial f^A_i}{\partial p^A_j}
&=(1-p_i^{A**})\lambda a_{ij}\ge0,
\end{align}
so the system is cooperative on \(\Omega\). It is also strictly sublinear: for \(0<\theta<1\) and \(z=(p^I,p^A)>0\),
\begin{align}
f^I_i(\theta z)-\theta f^I_i(z)
&=\theta(1-\theta)M_i^{I**}\big((\beta^U-\beta^A)p_i^A\notag\\
&\qquad +\beta^A p_i^I\big)>0,\\
f^A_i(\theta z)-\theta f^A_i(z)
&=\theta(1-\theta)p_i^A\big(\lambda M_i^{A**}\notag\\
&\qquad +\beta^U M_i^{I**}\big)>0.
\end{align}
Clearly \(f(0)=0\).

We claim that a vector field with \(f(0)=0\) that is cooperative and strictly sublinear on a convex cone has at most one positive zero. Indeed, let \(z,z'\) be two positive zeros. If they are comparable, say \(z\le z'\) and \(z\ne z'\), let \(\theta:=\max\{s\ge0:s z'\le z\}\). Then \(\theta<1\), and \(z_k=\theta z'_k\) for some \(k\). Cooperativity gives \(f_k(z)\ge f_k(\theta z')\), and strict sublinearity gives \(f_k(\theta z')>\theta f_k(z')=0\), contradicting \(f_k(z)=0\). If \(z,z'\) are incomparable, let \(\theta:=\max\{s\ge0:s z\le z'\ \text{and}\ s z'\le z\}\); then \(\theta\in(0,1)\) and, say, \(z'_k=\theta z_k\) for some \(k\), and the same argument gives \(0=f_k(z')\ge f_k(\theta z)>\theta f_k(z)=0\), a contradiction. Hence at most one positive equilibrium exists.

Finally, no equilibrium can have a zero infection component unless it is the rumor-free equilibrium: if \(p_i^{I**}=0\) for some \(i\) at an equilibrium, then \eqref{eq:sm4_ai} and the strict positivity of \(\beta^U p_i^{U**}+\beta^A p_i^{AS**}\) imply \(M_i^{I**}=\sum_j b_{ij}p_j^{I**}=0\), and irreducibility of \(\mathbf B\) forces \(p^{I**}=0\). Thus any positive equilibrium has \(p_i^{I**}>0\) for all \(i\), and the claim above applies.

Combining the existence, nonexistence, and uniqueness statements completes the proof.
\end{proof}
\section{Proof of Theorem~\ref{theorem:endemic_stability_compact_sm}}\label{appendix:proof_endemic_stability_compact_sm}
\begin{theorem}\label{theorem:endemic_stability_compact_sm}
  (Local asymptotic stability of the rumor-endemic equilibrium)
Let $\mathcal P^{**}=(\mathcal P^{US**},\mathcal P^{AS**},\mathcal P^{AI**})$ be a rumor-endemic equilibrium of system~\eqref{eq:main_sm} with $\mathcal P^{AI**}\ge0$. 
If all eigenvalues $\lambda$ of the Jacobian matrix $\mathcal J^{*}$, defined by~\eqref{eq:Jstar_compact_text_sm}--\eqref{eq:Bmn_compact_text_sm}, satisfy
\begin{align}\label{eq:matignon_endemic_compact_sm}
  \big|\arg(\lambda)\big| > \frac{\alpha\pi}{2},
\end{align}
then the rumor-endemic equilibrium $\mathcal P^{**}$ is locally asymptotically stable.
\end{theorem}

\begin{proof}
The Jacobian matrix $\mathcal J^{*}$ at the equilibrium $\mathcal P^{**}$ is obtained by linearizing system~\eqref{eq:main_component_sm}, resulting in the compact representation given in equations~\eqref{eq:Jstar_compact_text_sm}--\eqref{eq:Bmn_compact_text_sm}. Therefore, the perturbation $\Delta\mathcal P(t) = \mathcal P(t) - \mathcal P^{**}$ satisfies the linearized Caputo system
\begin{align}\label{eq:linearized_endemic_final_sm}
  {}^C_0\mathcal{D}_t^\alpha\,\Delta\mathcal P(t) = \mathcal J^{*}\,\Delta\mathcal P(t).
\end{align}
By Lemma~\ref{lemma:matignon_sm}, the equilibrium $\mathcal P^{**}$ is locally asymptotically stable whenever every eigenvalue $\lambda$ of $\mathcal J^{*}$ satisfies \(|\arg(\lambda)|>\alpha\pi/2\). This completes the proof.
\end{proof}

\ifCLASSOPTIONcaptionsoff
  \newpage
\fi

\section{Block-Adaptive FFT Solver Implementation Details}
\label{sec:solver_details}

This section provides the complete implementation details of the Block-Adaptive FFT solver referenced in Section~III-A of the main manuscript. The solver handles Caputo FODEs of the form
\begin{align}\label{eq:caputo_FODE_general_sm}
  {}^C_0\mathcal D_t^{\alpha} \mathcal P(t) = F(\mathcal P(t)), 
  \qquad \mathcal P(0) = \mathcal P_0.
\end{align}

The pseudocode of the solver is given in Algorithm~\ref{alg:caputo_solver_sm}.

\begin{algorithm}[!h]
    \caption{Block-Adaptive FFT Solver for Caputo FODEs}
    \label{alg:caputo_solver_sm}
    \renewcommand{\algorithmicrequire}{\textbf{Input:}}
    \renewcommand{\algorithmicensure}{\textbf{Output:}}
    
    \begin{algorithmic}[1]
        \REQUIRE $F$, $\alpha$, $y_0$, $T$, $N$, cfg
        \ENSURE $t$, $y$
        
        \STATE Initialize time grid, solution array, and history
        \STATE $M \gets \text{cfg.block\_steps}$, $n_0 \gets 0$
        
        \WHILE{$n_0 < N$}
            \STATE $M_{\text{cur}} \gets \min(M, N - n_0)$
            \STATE Compute convolution weights $b$
            
            \IF{$n_0 > 0$}
                \STATE $H \gets \text{FFTConv}(b, \text{history})[n_0:n_0+M_{\text{cur}}]$
            \ELSE
                \STATE $H \gets 0$
            \ENDIF
            
            \STATE Predictor: $y_{\text{pred}} \gets y_0 + \text{factor} \cdot (H + \text{FFTConv}(b_{\text{loc}}, F(\text{const})))$
            
            \FOR{$it = 1$ to $\text{max\_iter}$}
                \STATE Corrector: $y_{\text{corr}} \gets y_0 + \text{factor} \cdot (H + \text{FFTConv}(b_{\text{loc}}, F(y)))$
                \STATE Damped update: $y_{\text{new}} \gets y + \omega \cdot (y_{\text{corr}} - y)$
                \STATE Check convergence and adjust $\omega$
            \ENDFOR
            
            \IF{converged}
                \STATE Update solution and history
                \STATE Adapt $M$ based on performance
                \STATE $n_0 \gets n_0 + M_{\text{cur}}$
            \ELSE
                \STATE Reduce $M$ and retry
            \ENDIF
        \ENDWHILE
        
        \RETURN $t$, $y$
    \end{algorithmic}
\end{algorithm}

The solver uses a blockwise adaptive memory truncation strategy: each block of $M_{\text{cur}}$ steps is advanced using FFT-based convolution for the history term, and a predictor--corrector iteration with damping is applied within each block to ensure convergence. The numerical treatment of convolution integrals in this solver follows the convolution quadrature framework of Lubich~\cite{lubich1988convolutionI,lubich1988convolutionII}, and the predictor--corrector iteration is based on the Adams-type method for fractional differential equations proposed by Diethelm et al.~\cite{diethelm2002predictor}.

Specifically, we discretize the Caputo derivative using the L1-type scheme on a uniform grid $t_n = n h$ with $h = T/N$. The numerical solution satisfies $y(t_n) = y_0 + \frac{h^\alpha}{\Gamma(\alpha+1)} \sum_{j=0}^{n-1} b_{n-j} F(y(t_j), t_j)$, where the convolution weights are $b_k = (k+1)^\alpha - k^\alpha$. The history sum is accelerated by FFT via the function \texttt{FFTConv}, reducing the cost per block from $O(M_{\text{cur}}^2)$ to $O(M_{\text{cur}} \log M_{\text{cur}})$. 

Inside each block, a predictor--corrector iteration is performed. The predictor assumes a constant right-hand side equal to the value at the block start, giving an initial guess $y_{\text{pred}}$. The corrector recomputes the right-hand side using the current $y$ and updates via a damped Picard iteration: $y_{\text{new}} = y_{\text{old}} + \omega\,(y_{\text{fp}} - y_{\text{old}})$. The damping factor $\omega$ is initialized to $0.7$ and is reduced by a factor of $0.8$ whenever the error does not decrease. Convergence is declared when the maximum absolute change between iterations falls below a tolerance $\text{tol}=10^{-4}$ (or $10^{-10}$ in high-precision tests). If convergence is not achieved after $80$ iterations, the block size is reduced and the block is recomputed. 

The block size $M_{\text{cur}}$ is adapted dynamically: it is increased by a factor of $1.5$ (up to a prescribed maximum) when the current block converges quickly and is decreased by a factor of $0.5$ (down to a prescribed minimum) when convergence is slow or fails. The default configuration uses an initial block of $2048$ steps with bounds $128\le M_{\text{cur}}\le 8192$, whereas the production runs reported in this work start from a smaller block of $16$ steps (with an upper bound of $2048$) so as to guarantee convergence for the network sizes considered. This mechanism maintains high efficiency while guaranteeing numerical stability.

The overall computational complexity is approximately $O(d N \log N)$, where $d = 3N$ is the total number of state variables, because each FFT operation costs $O(M_{\text{cur}} \log M_{\text{cur}})$ and the total number of steps $N$ is covered by a sequence of blocks. For the system sizes considered in this work ($N=250$--$1000$), the solver typically completes a simulation in a few minutes on a standard workstation.

\section{Parameter Sensitivity Analysis}
\label{sec:sensitivity_sm}

This section provides a supplementary parameter sensitivity analysis that complements the numerical study of Section~III of the main manuscript. To provide a systematic and quantitative comparison of the relative influence of each model parameter on the long-term rumor prevalence, we conduct a tornado-style sensitivity analysis synthesizing the simulation data from Section~III-C of the main manuscript and Section~12.3 of this Supplementary Material. For each parameter, the steady-state infection density $\rho^{AI}$ is evaluated at the minimum and maximum values within the scanned range while holding all other parameters fixed at their baseline values ($\alpha=0.8$, $\beta^U=0.6$, $\lambda=0.3$, $\delta=0.5$, $m_A=4$, $\mu=1.0$, $\gamma=1/3$). The resulting tornado plot is shown in Fig.~S1.

\begin{figure}[htbp]
    \centering
    \includegraphics[width=0.8\linewidth]{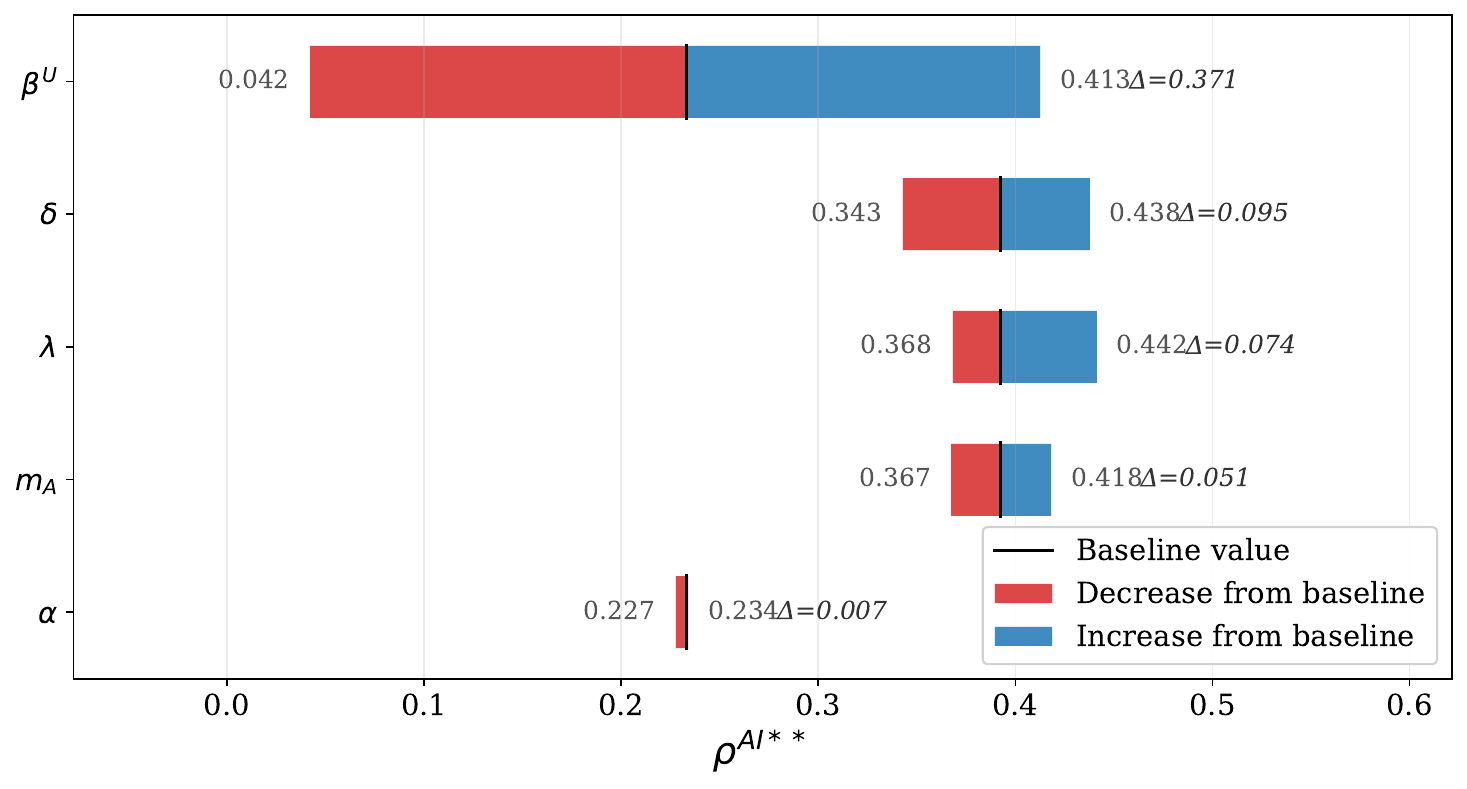}
    \caption{
        \textbf{Parameter sensitivity tornado plot.}
        For each of the five key parameters---$\alpha$ (fractional order), $\beta^U$ (infection rate for unaware individuals), $\lambda$ (awareness spreading rate), $\delta$ (forgetting rate), and $m_A$ (average degree of the awareness layer)---the horizontal bar spans the range of steady-state infection density $\rho^{AI}$ resulting from varying that parameter across its full scanned interval while keeping all other parameters at their baseline values. The black vertical tick within each bar marks the $\rho^{AI}$ value at the baseline parameter setting. The bar is colored red to the left of the baseline (decrease from baseline) and blue to the right (increase from baseline). The scanned ranges and data sources are as follows: $\alpha\in\{0.5,0.6,0.7,0.8,0.9,0.95\}$ and $\beta^U\in\{0.3,0.4,0.5,0.6,0.8,1.0\}$ from the fractional-order simulations on BA networks with $N=1000$ (Section~III-C); $\lambda\in[0.1,1.0]$, $\delta\in[0.1,1.0]$, and $m_A\in\{2,4,8\}$ from the integer-order mean-field computations on BA networks with $N=250$ (Section~12.3 of this Supplementary Material). The parameter with the largest span exerts the strongest influence on the steady-state infection density. The span $\Delta$ is reported to the right of each bar.
    }
    \label{fig:tornado_sm}
\end{figure}

The tornado plot quantitatively confirms several observations already suggested by the individual experiments. The infection rate $\beta^U$ is the dominant determinant of the steady-state rumor prevalence, with its influence ($\Delta=0.371$) exceeding that of all other parameters combined. In stark contrast, the fractional order $\alpha$ exerts a negligible influence on the equilibrium ($\Delta=0.007$), reinforcing the theoretical conclusion that memory effects shape the transient journey of the outbreak but not its final destination. Among the awareness-layer parameters, the forgetting rate $\delta$ ($\Delta=0.095$) and the awareness spreading rate $\lambda$ ($\Delta=0.074$) exhibit moderate but non-negligible impact, with $\delta$ being slightly more influential than $\lambda$, suggesting that sustained retention of corrective information is at least as important as its rapid dissemination. The awareness-layer degree $m_A$ has the smallest effect among the varying parameters ($\Delta=0.051$), consistent with the observation in Section~12.3 that network topology provides a secondary modulation compared with the dynamical parameters of the awareness process. These quantitative rankings provide practical guidance for rumor management: interventions targeting the infection rate (e.g., reducing contact frequency or improving individual resistance to misinformation) yield the largest reductions in final rumor prevalence, while awareness campaigns benefit from strategies that enhance both the reach ($\lambda$) and the durability ($\delta$) of corrective information.

\section{Global Stability Analysis}
\label{sec:global_stability}

This section analyzes the global asymptotic stability (GAS) of the rumor-free equilibrium (RFE) of the fractional-order UAU-SIS system \eqref{eq:main_sm}. Throughout, $p_i^{US}$, $p_i^{AS}$, and $p_i^{AI}$ denote the marginal probabilities that node $i$ is unaware--susceptible, aware--susceptible, and aware--infected, respectively, and
\begin{align}\label{eq:sm_fields}
  M_i^{A}=\sum_{j=1}^{N}a_{ij}\big(p_j^{AS}+p_j^{AI}\big),\qquad
  M_i^{I}=\sum_{j=1}^{N}b_{ij}\,p_j^{AI},
\end{align}
so that the componentwise system reads
\begin{align}\label{eq:sm_x}
  {}^C_0\mathcal D_t^{\alpha} p_i^{AI}
  &= M_i^I\big(\beta^U p_i^{US} + \beta^A p_i^{AS}\big) - \mu p_i^{AI},\\[2pt]
  {}^C_0\mathcal D_t^{\alpha} p_i^{AS}
  &= \lambda M_i^A p_i^{US} - \big(\delta + \beta^A M_i^I\big) p_i^{AS} + \mu p_i^{AI}. \label{eq:sm_y}
\end{align}
The feasible region $\Omega = \{\mathcal P\in\mathbb R^{N\times 3}_+ : p_i^{US}+p_i^{AS}+p_i^{AI}=1,\ \forall i\}$ is positively invariant by Theorem~\ref{theorem:positivity_boundedness_sm}. We write $\rho(\cdot)$ for the spectral radius and set
\begin{align}\label{eq:sm_barR0}
  \bar{\mathcal R}_0 := \frac{\beta^U \rho(\mathbf B)}{\mu}.
\end{align}
Since $\mathrm{diag}(p_i^{US*}\beta^U + p_i^{AS*}\beta^A)\le \beta^U\mathbf I_N$ componentwise and the spectral radius is monotone with respect to the componentwise order on nonnegative matrices, we have $\mathcal R_0\le\bar{\mathcal R}_0$.

\subsection{Fractional Comparison Lemmas}

The proofs below use a positivity/comparison principle for fractional linear systems whose coefficient matrix is Metzler, i.e., all off-diagonal entries are nonnegative.

\begin{lemma}[Positivity with nonnegative forcing]\label{lem:sm_metzler}
Let $M$ be a Metzler matrix and let $h\colon[0,\infty)\to\mathbb R_+^n$ be continuous. If $w\in C^1$ satisfies
\begin{align}\label{eq:sm_metzler_eq}
  {}^C_0\mathcal D_t^{\alpha} w(t) = M w(t) + h(t),\qquad w(0)\ge 0,
\end{align}
then $w(t)\ge 0$ componentwise for every $t\ge 0$.
\end{lemma}

\begin{proof}
For $\varepsilon>0$, let $w^{\varepsilon}$ solve \eqref{eq:sm_metzler_eq} with $h$ replaced by $h+\varepsilon\mathbf 1$ and initial value $w(0)+\varepsilon\mathbf 1$, where $\mathbf 1=(1,\dots,1)^\top$. We first show that $w^{\varepsilon}$ is nonnegative on every interval $[0,T]$. Suppose not, and let
\begin{align}
  \sigma=\inf\big\{t\in(0,T]:\ \min_{i}w_i^{\varepsilon}(t)<0\big\}.
\end{align}
By continuity, $w^{\varepsilon}(\sigma)\ge 0$, and there exists a component $i$ with $w_i^{\varepsilon}(\sigma)=0$ while $w_i^{\varepsilon}(t)<0$ for $t$ arbitrarily close to $\sigma$ from the right. Since $w_i^{\varepsilon}$ is $C^1$ on $(0,\infty)$ and $w_i^{\varepsilon}(0)>0$, it is not identically zero on $[0,\sigma]$. As $w_i^{\varepsilon}\ge 0$ on $[0,\sigma]$ and $w_i^{\varepsilon}(\sigma)=0$, the fractional integration-by-parts identity gives
\begin{multline}\label{eq:sm_ipbp}
  {}^C_0\mathcal D_t^{\alpha}w_i^{\varepsilon}(\sigma)
  = -\frac{1}{\Gamma(1-\alpha)}\Bigg[\frac{w_i^{\varepsilon}(0)}{\sigma^{\alpha}}\\
  + \alpha\int_0^{\sigma}\frac{w_i^{\varepsilon}(s)}{(\sigma-s)^{\alpha+1}}\,ds\Bigg] < 0,
\end{multline}
the boundary term at $s=\sigma$ vanishing because $w_i^{\varepsilon}(\sigma)=0$ and $w_i^{\varepsilon}(s)=O(\sigma-s)$. On the other hand, at $t=\sigma$ the right-hand side of the $i$-th component of \eqref{eq:sm_metzler_eq} equals
\begin{align}\label{eq:sm_rhs}
  \sum_{j\ne i} m_{ij} w_j^{\varepsilon}(\sigma) + h_i(\sigma) + \varepsilon \ge \varepsilon > 0,
\end{align}
since $m_{ij}\ge 0$ for $j\ne i$, $w_j^{\varepsilon}(\sigma)\ge 0$, and $w_i^{\varepsilon}(\sigma)=0$. This contradicts \eqref{eq:sm_ipbp}. Hence $w^{\varepsilon}\ge 0$ on $[0,T]$ for every $T>0$. The fractional Gr\"onwall inequality gives $w^{\varepsilon}\to w$ uniformly on compact intervals as $\varepsilon\to 0$, and therefore $w(t)\ge 0$ for all $t\ge 0$.
\end{proof}

\begin{corollary}[Fractional comparison]\label{cor:sm_comparison}
Let $M$ be Metzler. If $x\in C^1$ satisfies ${}^C_0\mathcal D_t^{\alpha} x \le M x$ componentwise for all $t\ge 0$, and if $z$ solves ${}^C_0\mathcal D_t^{\alpha} z = M z$ with $z(0)\ge x(0)$, then $x(t)\le z(t)$ for all $t\ge 0$.
\end{corollary}

\begin{proof}
Let $w=z-x$. Then
\begin{align}
  {}^C_0\mathcal D_t^{\alpha} w = M w + h,
  \qquad h := M x - {}^C_0\mathcal D_t^{\alpha} x \ge 0,
\end{align}
and $w(0)\ge 0$; Lemma~\ref{lem:sm_metzler} applies.
\end{proof}

\begin{lemma}[Stability of the linear fractional system]\label{lem:sm_linear}
If every eigenvalue $\lambda$ of $M$ satisfies $\mathrm{Re}(\lambda)<0$, then the zero solution of ${}^C_0\mathcal D_t^{\alpha}z = Mz$ is asymptotically stable for every $\alpha\in(0,1)$, i.e., $z(t)\to 0$ for every initial condition.
\end{lemma}

\begin{proof}
For every eigenvalue $\lambda$ with $\mathrm{Re}(\lambda)<0$ we have $|\arg(\lambda)|>\pi/2>\alpha\pi/2$; the Matignon-type criterion \cite{brandibur2021stability, li2010stability} yields asymptotic stability of the linear fractional system.
\end{proof}

\begin{lemma}[Scalar relaxation inequality]\label{lem:sm_scalar}
Let $\gamma>0$ and let $Q\in C([0,\infty);\mathbb R_+)$ satisfy $Q(t)\to 0$ as $t\to\infty$. If $s\in C^1$ satisfies
\begin{align}\label{eq:sm_relax}
  {}^C_0\mathcal D_t^{\alpha}s(t) \le -\gamma s(t) + Q(t),
\end{align}
then $\limsup_{t\to\infty}s(t)\le 0$. In particular, if $Q\equiv 0$, then $s(t)\le s(0)E_\alpha(-\gamma t^\alpha)\to 0$.
\end{lemma}

\begin{proof}
Let $S$ solve ${}^C_0\mathcal D_t^{\alpha}S = -\gamma S + Q$ with $S(0)=s(0)$. Applying Corollary~\ref{cor:sm_comparison} to the scalar Metzler matrix $M=-\gamma$ gives $s(t)\le S(t)$. The variation-of-constants formula yields
\begin{align}\label{eq:sm_S}
  S(t) = s(0)E_\alpha(-\gamma t^\alpha) \notag\\
  + \int_0^t (t-u)^{\alpha-1}E_{\alpha,\alpha}\big(-\gamma(t-u)^\alpha\big)\,Q(u)\,du.
\end{align}
The kernel $k(u)=u^{\alpha-1}E_{\alpha,\alpha}(-\gamma u^\alpha)$ is nonnegative, because the differentiation identity $\frac{d}{du}E_\alpha(-\gamma u^\alpha)=-\gamma k(u)$ and the complete monotonicity of $E_\alpha(-x)$ on $[0,\infty)$ for $0<\alpha\le 1$ \cite{pollard1948completely} imply $k(u)\ge 0$; moreover $\int_0^\infty k(u)\,du = \lim_{t\to\infty}(1-E_\alpha(-\gamma t^\alpha))/\gamma = 1/\gamma$. Since $Q$ is bounded and tends to zero, splitting the integral in \eqref{eq:sm_S} at $t/2$ gives, for all sufficiently large $t$,
\begin{align}
  \int_0^{t/2}k(t-u)Q(u)\,du
  &\le \|Q\|_\infty\int_{t/2}^{\infty}k \to 0,\\
  \int_{t/2}^{t}k(t-u)Q(u)\,du
  &\le \frac{1}{\gamma}\sup_{u\ge t/2}Q(u) \to 0,
\end{align}
while $s(0)E_\alpha(-\gamma t^\alpha)\to 0$. Hence $S(t)\to 0$, and so $\limsup_{t\to\infty}s(t)\le 0$.
\end{proof}

\subsection{Global Asymptotic Stability of the Rumor-Free Equilibrium}

\begin{theorem}[Global asymptotic stability of the rumor-free equilibrium]\label{thm:sm_G1}
Assume that $\mathbf A$ and $\mathbf B$ are nonnegative zero-diagonal irreducible adjacency matrices and that $0<\beta^A<\beta^U$. Define $\bar{\mathcal R}_0$ by \eqref{eq:sm_barR0}. If
\begin{align}\label{eq:sm_G1_cond}
  \bar{\mathcal R}_0 < 1
  \qquad\text{and}\qquad
  \lambda\rho(\mathbf A) < \delta,
\end{align}
then, for every $\alpha\in(0,1)$ and every initial condition $\mathcal P(0)\in\Omega$, the solution of \eqref{eq:main_sm} converges to the rumor-free equilibrium $\mathcal P^{*}=(\mathbf 1,\mathbf 0,\mathbf 0)$; i.e., $p_i^{US}(t)\to 1$, $p_i^{AS}(t)\to 0$, and $p_i^{AI}(t)\to 0$ for every $i$. Thus $\mathcal P^{*}$ is globally asymptotically stable.
\end{theorem}

\begin{proof}
We first establish that the rumor dies out. Since $p_i^{US}+p_i^{AS}=1-p_i^{AI}\le 1$ and $\beta^A\le\beta^U$, equation \eqref{eq:sm_x} gives
\begin{align}\label{eq:sm_xbound}
  {}^C_0\mathcal D_t^{\alpha}p_i^{AI}
  &\le \beta^U M_i^I\big(p_i^{US}+p_i^{AS}\big) - \mu p_i^{AI} \notag\\
  &\le \beta^U\sum_{j=1}^{N}b_{ij}p_j^{AI} - \mu p_i^{AI},
\end{align}
i.e., ${}^C_0\mathcal D_t^{\alpha}p^{AI}\le M p^{AI}$ with the Metzler matrix $M=\beta^U\mathbf B-\mu\mathbf I_N$. Let $z$ solve ${}^C_0\mathcal D_t^{\alpha}z = Mz$ with $z(0)=p^{AI}(0)$. Corollary~\ref{cor:sm_comparison} yields $0\le p^{AI}(t)\le z(t)$. Every eigenvalue $\xi$ of $M$ is of the form $\xi=\beta^U\lambda-\mu$ for an eigenvalue $\lambda$ of $\mathbf B$, and satisfies $\mathrm{Re}(\xi)=\beta^U\mathrm{Re}(\lambda)-\mu\le\beta^U|\lambda|-\mu\le\beta^U\rho(\mathbf B)-\mu<0$, so Lemma~\ref{lem:sm_linear} gives $z(t)\to 0$; hence $p^{AI}(t)\to 0$.

Next, we show that the awareness layer also dies out. From \eqref{eq:sm_y} and $p_i^{US}\le 1$,
\begin{align}\label{eq:sm_ybound}
  {}^C_0\mathcal D_t^{\alpha}p_i^{AS}
  &\le \lambda\Big(\sum_{j}a_{ij}p_j^{AS}+\sum_{j}a_{ij}p_j^{AI}\Big) - \delta p_i^{AS} + \mu p_i^{AI}\notag\\
  &= \big((\lambda\mathbf A-\delta\mathbf I_N)p^{AS}\big)_i + q_i(t),
\end{align}
where $q_i(t)=\lambda\sum_j a_{ij}p_j^{AI}(t)+\mu p_i^{AI}(t)\ge 0$ and $q_i(t)\to 0$ by the previous step. Let $v>0$ be the left Perron vector of the irreducible nonnegative matrix $\mathbf A$, normalized by $v^\top\mathbf A=\rho(\mathbf A)v^\top$, and set $s(t)=\sum_i v_i p_i^{AS}(t)$. Multiplying \eqref{eq:sm_ybound} by $v_i$ and summing over $i$ gives
\begin{align}\label{eq:sm_sbound}
  {}^C_0\mathcal D_t^{\alpha}s(t)
  &\le \big(\lambda\rho(\mathbf A)-\delta\big)s(t) + Q(t),\notag\\
  Q(t)&:=\sum_{i=1}^{N}v_i\,q_i(t)\ge 0,
\end{align}
with $Q(t)\to 0$. Since $\gamma:=\delta-\lambda\rho(\mathbf A)>0$, Lemma~\ref{lem:sm_scalar} yields $s(t)\to 0$, and because $v_i>0$ for every $i$, we obtain $p_i^{AS}(t)\le s(t)/v_i\to 0$ for every $i$.

Finally, $p_i^{US}(t)=1-p_i^{AS}(t)-p_i^{AI}(t)\to 1$, so every solution converges to $\mathcal P^{*}=(\mathbf 1,\mathbf 0,\mathbf 0)$. This point is indeed the RFE: under $\lambda\rho(\mathbf A)<\delta$ the single-layer awareness system has the unique equilibrium $p^{AS}=0$. Indeed, if $y\in[0,1]^N$ satisfies $\lambda(\mathbf A y)_i(1-y_i)=\delta y_i$ for all $i$, then $y_i<1$ and $\lambda(\mathbf A y)_i=\delta y_i/(1-y_i)\ge\delta y_i$; multiplying $\lambda\mathbf A y\ge\delta y$ by $v^\top$ gives $(\lambda\rho(\mathbf A)-\delta)v^\top y\ge 0$, which forces $y=0$. Since $\mathcal R_0\le\bar{\mathcal R}_0<1$, the local asymptotic stability of $\mathcal P^{*}$ follows from Corollary~1 of the main manuscript; global attractivity together with local stability gives global asymptotic stability.
\end{proof}

\begin{remark}[Scope and open problems]\label{rem:sm_scope}
The result above is a sufficient-condition theorem, and we deliberately do not claim more than we prove. Several points deserve emphasis.
\begin{itemize}
\item Under $\bar{\mathcal R}_0<1$ alone (with no condition on $\lambda\rho(\mathbf A)$), the first part of the proof already shows that the rumor dies out: $p^{AI}(t)\to 0$ for every initial condition. This is the fractional analogue of the classical disease-free global result of Lajmanovich--Yorke \cite{lajmanovich1976deterministic} for integer-order network SIS models, restricted to the subcritical case.
\item When the awareness layer is supercritical ($\lambda\rho(\mathbf A)>\delta$), the RFE carries a nontrivial awareness distribution $p^{AS*}>0$ solving $\lambda(\mathbf A p^{AS*})_i(1-p^{AS*}_i)=\delta p^{AS*}_i$. A rigorous proof that $p^{AS}(t)\to p^{AS*}$ after the rumor has vanished would require an asymptotic-autonomy (limit-system) principle for Caputo fractional systems, which is a genuinely open technical problem: the classical integer-order argument \cite{lajmanovich1976deterministic} relies on the semigroup property of the flow, which Caputo systems lack. We therefore leave this case as an open refinement and support it numerically in Section~\ref{sec:sm_gs_num} (see Fig.~\ref{fig:sm_gs}).
\item The sharp condition $\mathcal R_0<1$ (instead of $\bar{\mathcal R}_0<1$) for the rumor-free equilibrium remains open; the single-layer fractional network SIS analogue in the disease-free case is established in \cite{liu2019fractional}. For the rumor-endemic equilibrium in the regime $\mathcal R_0>1$, a rigorous global asymptotic stability theorem under a verifiable quadratic dissipativity condition is established in Section~\ref{sec:sm_endemic} (Theorem~\ref{thm:sm_G2}). On the BA network used in Section~III-C of the main manuscript, $\bar{\mathcal R}_0=1$ at $\beta^U\approx 0.126$ whereas $\mathcal R_0=1$ at $\beta^U\approx 0.289$; the sufficient condition of Theorem~\ref{thm:sm_G1} therefore covers approximately $43\%$ of the subcritical interval, the gap being the price of the crude bound $\beta^U p_i^{US}+\beta^A p_i^{AS}\le\beta^U$.
\end{itemize}
\end{remark}

\subsection{Global Asymptotic Stability of the Rumor-Endemic Equilibrium}
\label{sec:sm_endemic}

We now turn to the rumor-endemic regime $\mathcal R_0>1$, in which Theorem~4 of the main manuscript guarantees the existence and uniqueness of the rumor-endemic equilibrium $\mathcal P^{**}=(p^{US**},p^{AS**},p^{AI**})$ with $p^{AI**}>0$. Throughout this subsection we assume in addition that $p^{US**},p^{AS**}>0$ componentwise, which is the case for the parameter regime of the main manuscript. For $\mathcal P\in\Omega$ write
\begin{align}\label{eq:sm_dev}
  \Delta x_i&=p_i^{AI}-p_i^{AI**},\notag\\
  \Delta y_i&=\big(p_i^{AS}+p_i^{AI}\big)-\big(p_i^{AS**}+p_i^{AI**}\big),
\end{align}
where $\Delta x_i$ is the deviation of the infected fraction and $\Delta y_i$ that of the total aware fraction. The componentwise vector fields of the system \eqref{eq:sm_x}--\eqref{eq:sm_y} are denoted by $G_i(\mathcal P)$ and $H_i(\mathcal P)$:
\begin{align}\label{eq:sm_GH}
  G_i(\mathcal P)&=M_i^I\big(\beta^U p_i^{US}+\beta^A p_i^{AS}\big)-\mu p_i^{AI},\\
  H_i(\mathcal P)&=\lambda M_i^A p_i^{US}-\delta p_i^{AS}+\beta^U M_i^I p_i^{US},
\end{align}
so that $G_i(\mathcal P^{**})=H_i(\mathcal P^{**})=0$ for every $i$; here $H_i$ is the right-hand side of the equation satisfied by the total aware fraction $p_i^{AS}+p_i^{AI}$, obtained by adding the second and third equations of \eqref{eq:main_component_sm}.

\begin{theorem}[Global asymptotic stability of the rumor-endemic equilibrium under a quadratic dissipativity condition]\label{thm:sm_G2}
Let $\mathbf A$ and $\mathbf B$ be nonnegative zero-diagonal irreducible adjacency matrices, let $0<\beta^A<\beta^U$ and $\mathcal R_0>1$, and let $\mathcal P^{**}$ be the unique rumor-endemic equilibrium with $p^{US**},p^{AS**},p^{AI**}>0$. Suppose there exist positive numbers $d_i$, $e_i$ and a constant $c>0$ such that, for every $\mathcal P\in\Omega$,
\begin{align}\label{eq:sm_diss}
  2\sum_{i=1}^{N}d_i\,\Delta x_i\,G_i(\mathcal P)
  +2\sum_{i=1}^{N}e_i\,\Delta y_i\,H_i(\mathcal P) \notag\\
  \le -c\Big(\sum_{i=1}^{N}\Delta x_i^2+\sum_{i=1}^{N}\Delta y_i^2\Big),
\end{align}
where $\Delta x_i$ and $\Delta y_i$ are defined in \eqref{eq:sm_dev}. Then, for every $\alpha\in(0,1)$ and every initial condition in $\Omega$, the solution of \eqref{eq:main_sm} converges to $\mathcal P^{**}$; i.e., $\mathcal P^{**}$ is globally asymptotically stable.
\end{theorem}

\begin{proof}
Consider the positive definite quadratic function
\begin{align}\label{eq:sm_V2}
  V(\mathcal P)=\sum_{i=1}^{N}d_i\,\Delta x_i^2+\sum_{i=1}^{N}e_i\,\Delta y_i^2.
\end{align}
By the fractional inequality for quadratic functions ${}^C_0\mathcal D_t^{\alpha}z^2\le 2z\,{}^C_0\mathcal D_t^{\alpha}z$ \cite{aguila2014lyapunov} applied componentwise, and since the Caputo derivative of a constant vanishes,
\begin{align}\label{eq:sm_dV2}
  {}^C_0\mathcal D_t^{\alpha}V
  \le 2\sum_{i=1}^{N}d_i\,\Delta x_i\,G_i(\mathcal P)
    +2\sum_{i=1}^{N}e_i\,\Delta y_i\,H_i(\mathcal P) \notag\\
  \le -c\Big(\sum_{i=1}^{N}\Delta x_i^2+\sum_{i=1}^{N}\Delta y_i^2\Big),
\end{align}
where the last inequality is \eqref{eq:sm_diss}. Let $\bar d=\max_i d_i$ and $\bar e=\max_i e_i$; then $\sum_i\Delta x_i^2+\sum_i\Delta y_i^2\ge V/\max(\bar d,\bar e)$, so that ${}^C_0\mathcal D_t^{\alpha}V\le -\big(c/\max(\bar d,\bar e)\big)V$. Lemma~\ref{lem:sm_scalar} (with $Q\equiv 0$) yields
\begin{align}
  V(t)\le V(0)E_\alpha\Big(-\frac{c}{\max(\bar d,\bar e)}\,t^\alpha\Big)\to 0,
\end{align}
hence $\Delta x(t)\to 0$ and $\Delta y(t)\to 0$, i.e., $p^{AI}(t)\to p^{AI**}$ and $p^{AS}(t)+p^{AI}(t)\to p^{AS**}+p^{AI**}$, which gives $p^{AS}(t)\to p^{AS**}$ and consequently $p^{US}(t)=1-p^{AS}(t)-p^{AI}(t)\to p^{US**}$. Global asymptotic stability follows.
\end{proof}

\begin{remark}[On the dissipativity condition \eqref{eq:sm_diss}]\label{rem:sm_diss}
Condition \eqref{eq:sm_diss} is a structural quadratic-dissipativity condition of Lajmanovich--Yorke type, expressed directly in terms of the vector fields of the coupled system; for a given network it can be checked either analytically or by direct evaluation on the compact feasible set $\Omega$. For the BA network of the main manuscript, with $d_i=e_i=1$, $\beta^U=0.6$ ($\mathcal R_0=2.07$), $\beta^A=\gamma\beta^U$ with $\gamma=1/3$, $\lambda=0.3$, $\delta=0.5$ and $\mu=1.0$, we evaluated the ratio of the left-hand side of \eqref{eq:sm_diss} to $\sum_i(\Delta x_i^2+\Delta y_i^2)$ over $1.4\times10^5$ sampled feasible states---including uniform samples of $\Omega$, small perturbations around $\mathcal P^{**}$, states on the boundary faces of $\Omega$, and local maximization of the ratio starting from the best sampled states. The largest value found is approximately $-0.19$, attained on the boundary face $p^{AI}\equiv0$ on which the infection has died out, so that \eqref{eq:sm_diss} is satisfied numerically with $c\approx0.19$ for this network (Section~\ref{sec:sm_gs_num}). Since finite sampling cannot verify the inequality over the entire continuous feasible region, Theorem~\ref{thm:sm_G2} is presented as a conditional global-stability criterion, and we do not claim that the sampling alone constitutes a proof of its applicability to this network.
\end{remark}

\subsection{Numerical Verification of Global Convergence}
\label{sec:sm_gs_num}

Complementing the theoretical results of Section~\ref{sec:sm_endemic}, we verify the global convergence behavior directly by simulation. On the BA network of Section~12.3 of this Supplementary Material ($N=250$, $m_A=4$, $m_B=2$), we simulate the system from twelve markedly different initial conditions (infection fractions ranging from $1\%$ to $30\%$ and awareness fractions from $1\%$ to $60\%$, with independent random perturbations). The results are summarized in Table~\ref{tab:sm_gs} and illustrated in Fig.~\ref{fig:sm_gs}.

\begin{figure}[htbp]
\centering
\subfigure[$\mathcal R_0>1$: convergence to the rumor-endemic equilibrium]{
    \includegraphics[width=0.45\linewidth]{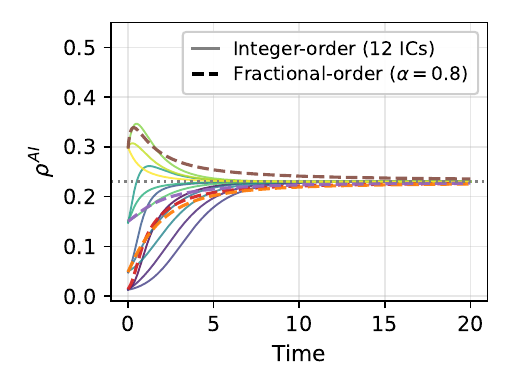}
    \label{fig:sm_gs_endemic}
}
\hfill
\subfigure[$\mathcal R_0<1$: convergence to the rumor-free equilibrium]{
    \includegraphics[width=0.45\linewidth]{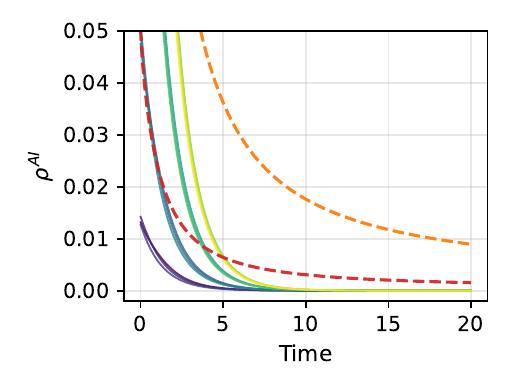}
    \label{fig:sm_gs_free}
}
\caption{\textbf{Numerical verification of global convergence on the BA network ($N=250$, $m_A=4$, $m_B=2$, $\lambda=0.3$, $\delta=0.5$, $\mu=1.0$, $\gamma=1/3$).} Temporal evolution of the infection density $\rho^{AI}$ from twelve markedly different initial conditions (infection fractions from $1\%$ to $30\%$ and awareness fractions from $1\%$ to $60\%$, with independent random perturbations). (a) In the endemic regime ($\beta^U=0.6$, $\mathcal R_0=2.07>1$, left), all twelve integer-order trajectories (solid curves) converge to the same steady-state value $\rho^{AI**}=0.2301$ (gray dotted line), with a spread below $5\times10^{-6}$ across initial conditions; the fractional-order trajectories ($\alpha=0.8$, dashed curves) converge to the same value up to the tolerance of the fractional solver, with final values in $[0.2254,0.2355]$ over the finite horizon $T=20$. (b) In the rumor-free regime ($\beta^U=0.1$, $\mathcal R_0=0.35<1$, right), all trajectories decay to $\rho^{AI}\approx0$ for every initial condition, for both the integer-order system (down to $10^{-7}$) and the fractional-order system (down to $0.0015$--$0.0089$ within the finite horizon, the residual reflecting the slower subdiffusive convergence). Solid curves show the integer-order solutions from the twelve initial conditions, and dashed curves show the fractional-order ($\alpha=0.8$) solutions from representative initial conditions. These results provide numerical evidence that the rumor-free and rumor-endemic equilibria are globally attracting within their respective regimes $\mathcal R_0<1$ and $\mathcal R_0>1$.}
\label{fig:sm_gs}
\end{figure}

\begin{table}[htbp]
\centering
\caption{Numerical verification of global convergence on the BA network ($N=250$).}
\label{tab:sm_gs}
\resizebox{\columnwidth}{!}{%
\begin{tabular}{lcc}
\toprule
 & $\beta^U=0.6$ ($\mathcal R_0=2.07>1$) & $\beta^U=0.1$ ($\mathcal R_0=0.35<1$)\\
\midrule
Final $\rho^{AI}$ (int.) & $0.23008$ & $\sim10^{-7}$\\
Final $\rho^{AI}$ (frac., $\alpha=0.8$) & $0.2254$--$0.2355$ & $0.0015$--$0.0089$ (within $T=20$)\\
12-IC spread (int.) & $<5\times10^{-6}$ & $<10^{-6}$\\
\bottomrule
\end{tabular}%
}
\end{table}

In the endemic regime, all twelve integer-order trajectories converge to the same steady-state value $\rho^{AI**}=0.2301$ with a spread below $5\times10^{-6}$; the fractional-order trajectories ($\alpha=0.8$) converge to the same value up to the tolerance of the fractional solver, the slightly larger spread ($0.2254$--$0.2355$) reflecting the slower memory-induced convergence over the finite horizon $T=20$. In the rumor-free regime, all trajectories decay to zero for every initial condition, for both the integer-order system (down to $10^{-7}$) and the fractional-order system (down to $0.0015$--$0.0089$ within the simulation horizon, with the residual attributable to the subdiffusive transient slowdown). These results constitute numerical evidence---not a proof---that the rumor-free and rumor-endemic equilibria are globally attracting within their respective regimes $\mathcal R_0<1$ and $\mathcal R_0>1$. They also provide numerical evidence in support of the dissipativity condition \eqref{eq:sm_diss} reported in Remark~\ref{rem:sm_diss}: for this network, evaluation of the ratio of the left-hand side of \eqref{eq:sm_diss} to $\sum_i(\Delta x_i^2+\Delta y_i^2)$ over $1.4\times10^5$ sampled feasible states, combined with local maximization of the ratio, yields a maximum value of approximately $-0.19$ with $d_i=e_i=1$, attained on the boundary face $p^{AI}\equiv0$. Since finite sampling cannot verify the inequality over the entire feasible region, Theorem~\ref{thm:sm_G2} is presented as a conditional criterion, and we do not claim that the sampling alone constitutes a proof of its applicability. The rigorous global statements established in this section are Theorem~\ref{thm:sm_G1} (GAS of the RFE under \eqref{eq:sm_G1_cond}) and Theorem~\ref{thm:sm_G2} (GAS of the rumor-endemic equilibrium under \eqref{eq:sm_diss}); the sharp-threshold cases are discussed in Remark~\ref{rem:sm_scope}.

\section{Additional Numerical Experiments}
\label{sec:sm_additional}

This section collects the numerical experiments that complement those reported in Section~III of the main manuscript: the convergence test towards the integer-order limit, the bifurcation diagram of the coupled system, the influence of the awareness layer parameters and connectivity, and the role of the two-layer network architecture.

\subsection{Convergence to the Integer-Order Limit}
\label{sec:sm_convergence}

To further validate the consistency of the fractional-order formulation, we perform a convergence test as the fractional order approaches the integer limit ($\alpha \to 1$). Using the baseline parameter configuration ($\lambda=0.3$, $\beta^U=0.6$, $\beta^A=\gamma\beta^U$ with $\gamma=1/3$, $\delta=0.5$, $\mu=1.0$), we compare the fractional dynamics for $\alpha = 0.5, 0.6, 0.7, 0.8, 0.9, 0.95$ against the integer-order solution ($\alpha=1$). The time evolution of the infection density $\rho^{AI}$ is shown in Fig.~\ref{fig:convergence}, where the fractional curves gradually approach the integer-order reference as $\alpha$ increases. The inset displays the root-mean-square error (RMSE) computed over the entire transient phase ($t>0.1$), which decreases monotonically with $\alpha$, approaching zero at $\alpha=1$. This quantitatively confirms that the fractional model recovers the classical integer-order dynamics in the limit $\alpha\to1$, thereby ensuring that the memory effects introduced by the Caputo derivative do not introduce spurious discrepancies but rather produce a smooth interpolation between Markovian and strongly non-Markovian regimes. In summary, the close agreement between the individual-based mean-field approximation (IBMFA) solutions and the mechanism-based Monte Carlo simulations reported in the main manuscript validates the mean-field closure for the integer-order setting, and since the fractional model reduces to this validated integer-order model as $\alpha\to1$, the convergence test provides a semi-analytical consistency check of the fractional memory extension.

\begin{figure}[!ht]
    \centering
    \includegraphics[width=0.85\linewidth]{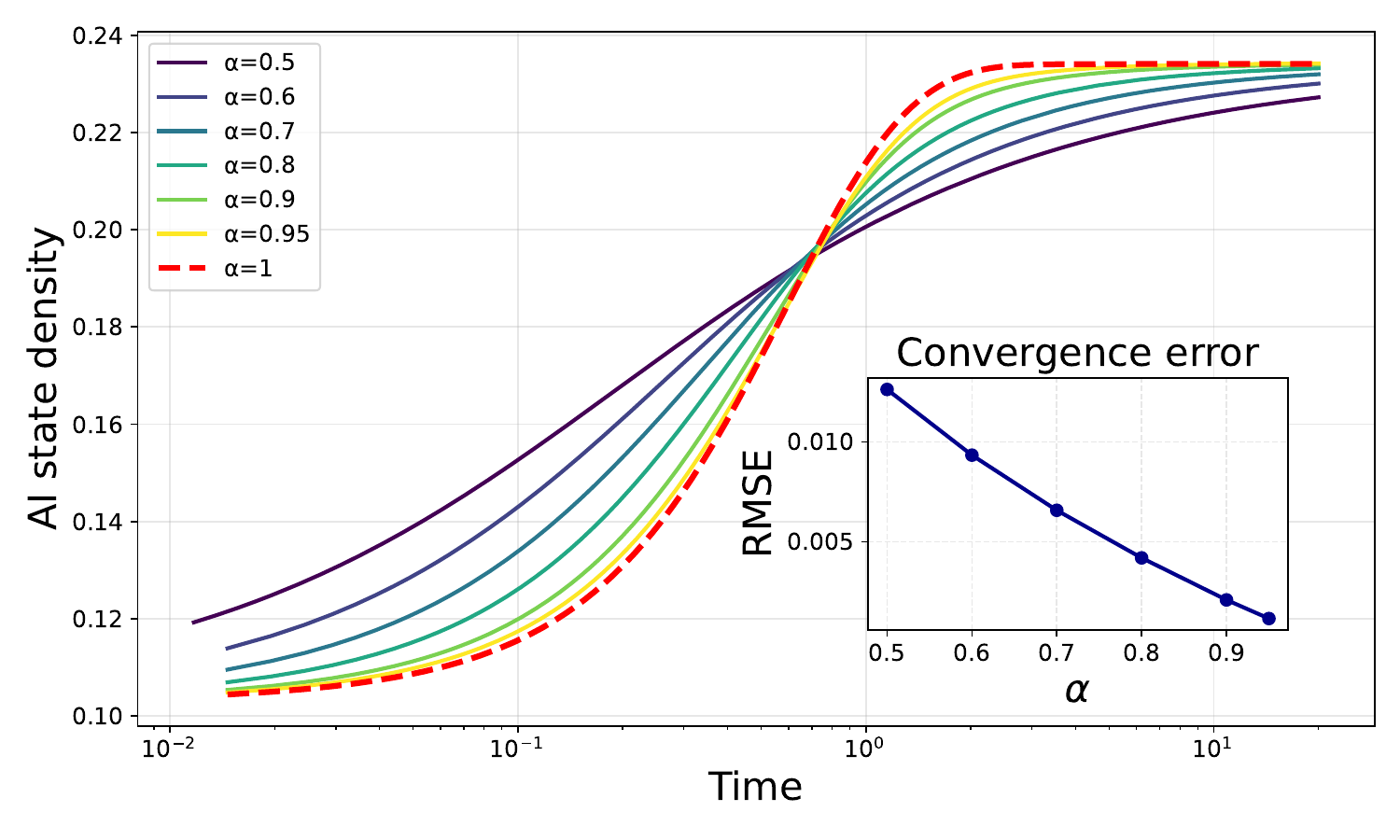}
    \vspace{-4mm}
    \caption{
    \textbf{Convergence of the fractional-order model to the integer-order limit as $\alpha\to1$.}
    Main panel: Temporal evolution of infection density $\rho^{AI}$ for decreasing fractional orders $\alpha$ (colored curves) and the integer-order solution ($\alpha=1$, red dashed) under the Baseline parameter configuration. The curves with $\alpha=0.95$ are nearly indistinguishable from the integer-order reference, demonstrating that the fractional dynamics continuously approach the Markovian limit. The inset shows the root-mean-square error (RMSE) computed over $t>0.1$ as a function of $\alpha$, which monotonically decreases and tends to zero when $\alpha\to1$. This provides strong numerical evidence that the fractional-order model is a consistent extension of the classical UAU-SIS framework.
    }
    \vspace{-8mm}
    \label{fig:convergence}
\end{figure}

\subsection{Bifurcation Analysis and Threshold Validation}
\label{sec:sm_bifurcation}

To further corroborate the threshold-governed nature of the coupled awareness--rumor system and directly visualize the transition from the rumor-free to the rumor-endemic regime, we present a bifurcation diagram in Fig.~\ref{fig:bifurcation}. Using the integer-order mean-field solver, we densely scan $\beta^U$ over the range $[0.05, 2.0]$ on a BA network with $N=250$, $m_A=4$, and $m_B=2$, while keeping other parameters at their baseline values ($\lambda=0.3$, $\delta=0.5$, $\mu=1.0$, $\gamma=1/3$). For each $\beta^U$, the steady-state infection density $\rho^{AI}$ is recorded after $T=20$. Fractional-order simulations with $\alpha=0.6$ and $\alpha=0.8$ are additionally performed at selected $\beta^U$ values on the same network to verify the $\alpha$-independence of the equilibrium. The inset of Fig.~\ref{fig:bifurcation} magnifies the vicinity of $\mathcal{R}_0\approx 1$, where the infection density transitions from near-zero to positive values.

\begin{figure}[htbp]
    \centering
    \includegraphics[width=0.85\linewidth]{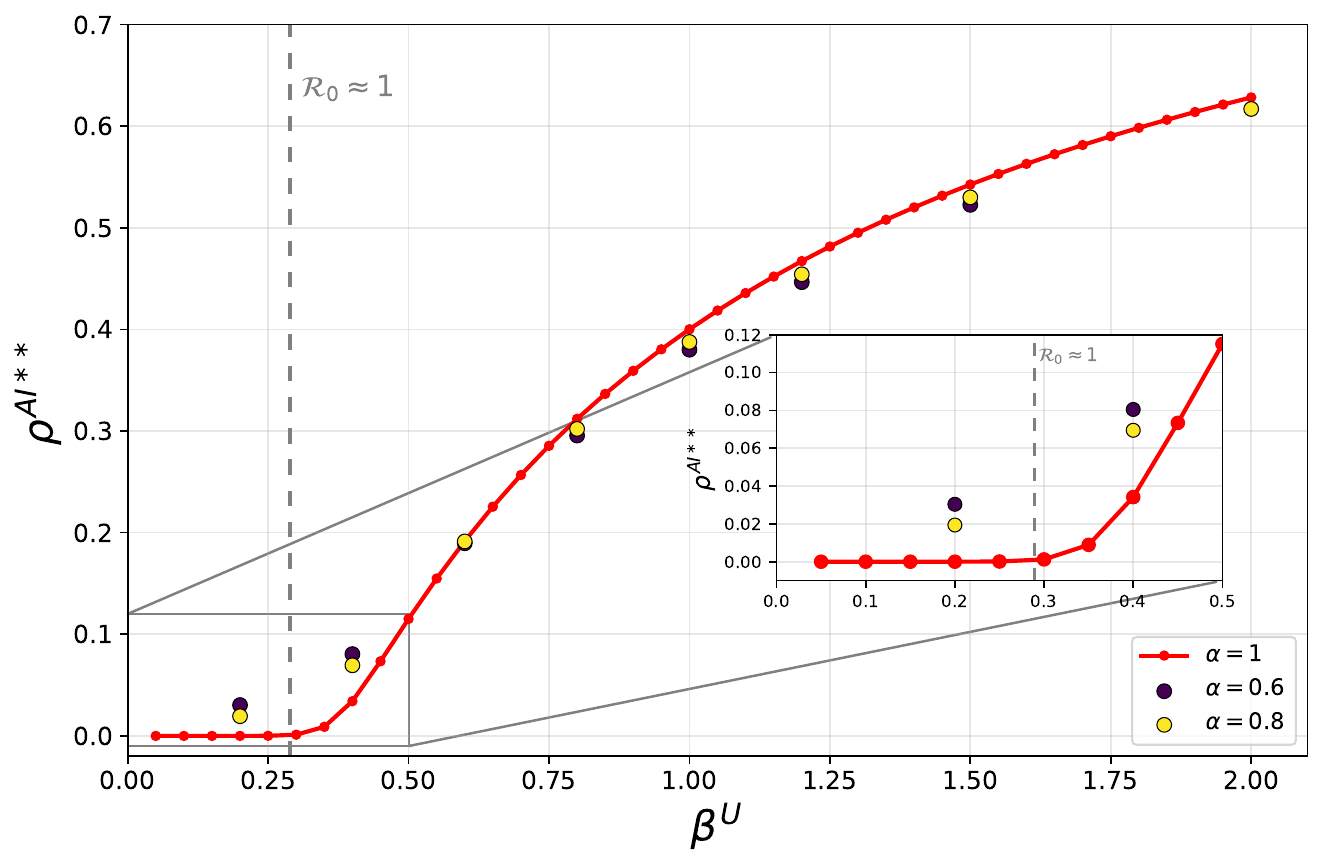}
    \vspace{-5mm}
    \caption{
        \textbf{Bifurcation diagram of the fractional-order UAU-SIS model.}
        Steady-state infection density $\rho^{AI}$ as a function of the infection rate $\beta^U$ for unaware individuals.
        The red curve with circular markers corresponds to the integer-order solution ($\alpha=1$), while the colored scatter points represent fractional-order results with $\alpha=0.6$ and $\alpha=0.8$ on the same BA network ($N=250$, $m_A=4$, $m_B=2$).
        The gray dashed line marks the location of $\mathcal{R}_0=1$ ($\beta^{U,c}\approx 0.289$), below which the system remains in the rumor-free state ($\rho^{AI}\approx 0$) and above which the rumor-endemic equilibrium emerges with $\rho^{AI}$ monotonically increasing with $\beta^U$.
        The inset magnifies the critical region near the threshold. In the endemic regime the fractional-order results agree closely with the integer-order curve; near the threshold the fractional-order points lie slightly above it, a finite-horizon artifact of the slower subdiffusive convergence ($T=20$), while the long-term equilibrium remains $\alpha$-independent (Remark~1 of the main manuscript).
        The zoom area is indicated by the gray rectangular frame and connecting lines in the main panel.
    }
    \vspace{-4mm}
    \label{fig:bifurcation}
\end{figure}

The bifurcation diagram quantitatively confirms several theoretical and numerical findings established earlier. First, the system undergoes a clear threshold transition at $\mathcal{R}_0\approx 1$: for $\beta^U\lesssim 0.25$, the integer-order infection density remains vanishingly small (numerically below $10^{-4}$), while for $\beta^U\gtrsim 0.3$, $\rho^{AI}$ rises continuously with increasing $\beta^U$, consistent with the transcritical nature of the rumor-free to rumor-endemic bifurcation predicted by Theorem~\ref{theorem:endemic_existence_sm}. The threshold value observed here, $\beta^{U,c}\approx0.289$, corresponds precisely to $\mathcal{R}_0=1$: the reproduction numbers computed on the same network are $\mathcal{R}_0=0.35$, $0.69$, $1.04$, $2.07$, and $5.18$ at $\beta^U=0.1$, $0.2$, $0.3$, $0.6$, and $1.5$, respectively, fully consistent with the integer-order bifurcation diagram. Second, in the endemic regime the fractional-order data points ($\alpha=0.6$ and $\alpha=0.8$) agree with the integer-order curve and approach the same threshold; the small positive deviations visible near the threshold (e.g., at $\beta^U=0.2$) reflect the slower subdiffusive convergence within the finite simulation horizon $T=20$, and converge to the same $\alpha$-independent equilibrium as the horizon increases, in agreement with Remark~1. Third, the smooth, monotonic growth of $\rho^{AI}$ with $\beta^U$ beyond the threshold indicates the absence of bistability or discontinuous transitions, which is consistent with the unique-equilibrium structure of the coupled system. Together, these observations reinforce the conclusion that while $\alpha$ governs the transient memory effects and the convergence rate, the ultimate fate of rumor propagation is dictated by the classical reproductive potential encoded in the infection rate $\beta^U$ and the network structure.

\subsection{Impact of the Awareness Layer on Rumor Dynamics}
\label{sec:sm_exp3}
To quantitatively assess the significance of the consciousness layer in the coupled UAU-SIS model, 
we conduct a systematic investigation of how awareness dynamics parameters and network topology 
influence the final rumor prevalence. This simulation aims to validate that the consciousness layer 
is not merely an auxiliary component but plays a crucial role in modulating rumor spread through 
information dissemination and behavioral changes.

We fix the fractional order at $\alpha = 0.8$ and employ a conservative infection rate 
$\beta^U = 0.4$ to ensure observable rumor activity. The network consists of $N = 250$ nodes 
with the risk layer generated using $m_{B} = 2$. We systematically vary three key 
parameters of the consciousness layer: the awareness spreading rate $\lambda \in [0.1, 1.0]$ 
(in 10 evenly spaced steps), the forgetting rate $\delta \in [0.1, 1.0]$ (likewise in 10 steps), and the network connectivity 
$m_{A} \in \{2, 4, 8\}$. For each parameter combination, we measure the 
steady-state infection density after $T = 20.0$ time units using our fractional-order IBMFA solver. The simulation results, which comprehensively characterize the influence of consciousness layer 
properties on rumor outcomes, are presented in Fig.~\ref{fig:exp3_consciousness_impact}.

\begin{figure}[htbp]
  \centering
  \subfigure[Three-dimensional visualization of infection density as functions of awareness parameters]{
    \includegraphics[width=0.8\linewidth]{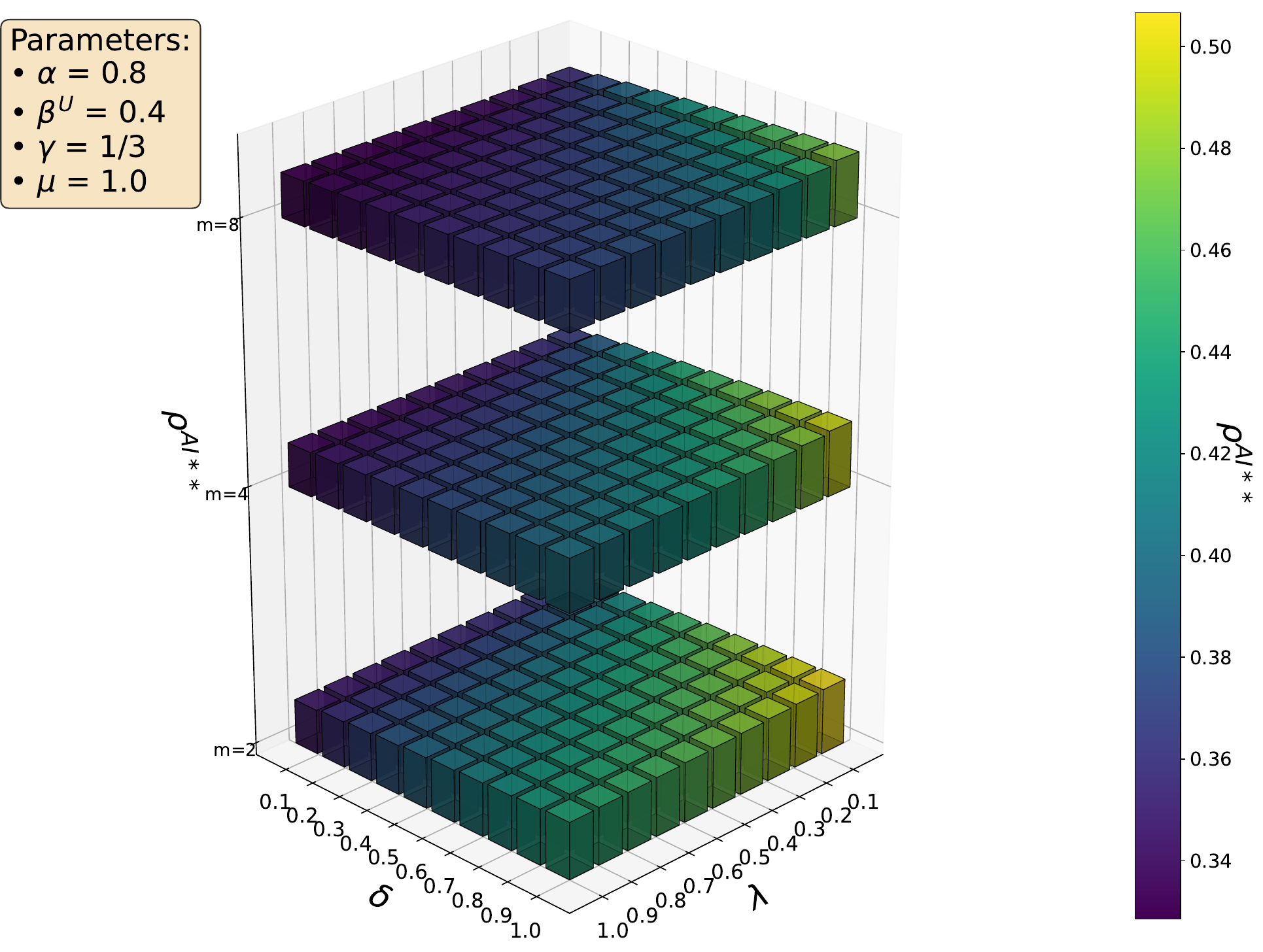}
    \label{fig:exp3a}
  }
  \subfigure[Average infection density across different consciousness layer connectivities]{
    \includegraphics[width=0.6\linewidth]{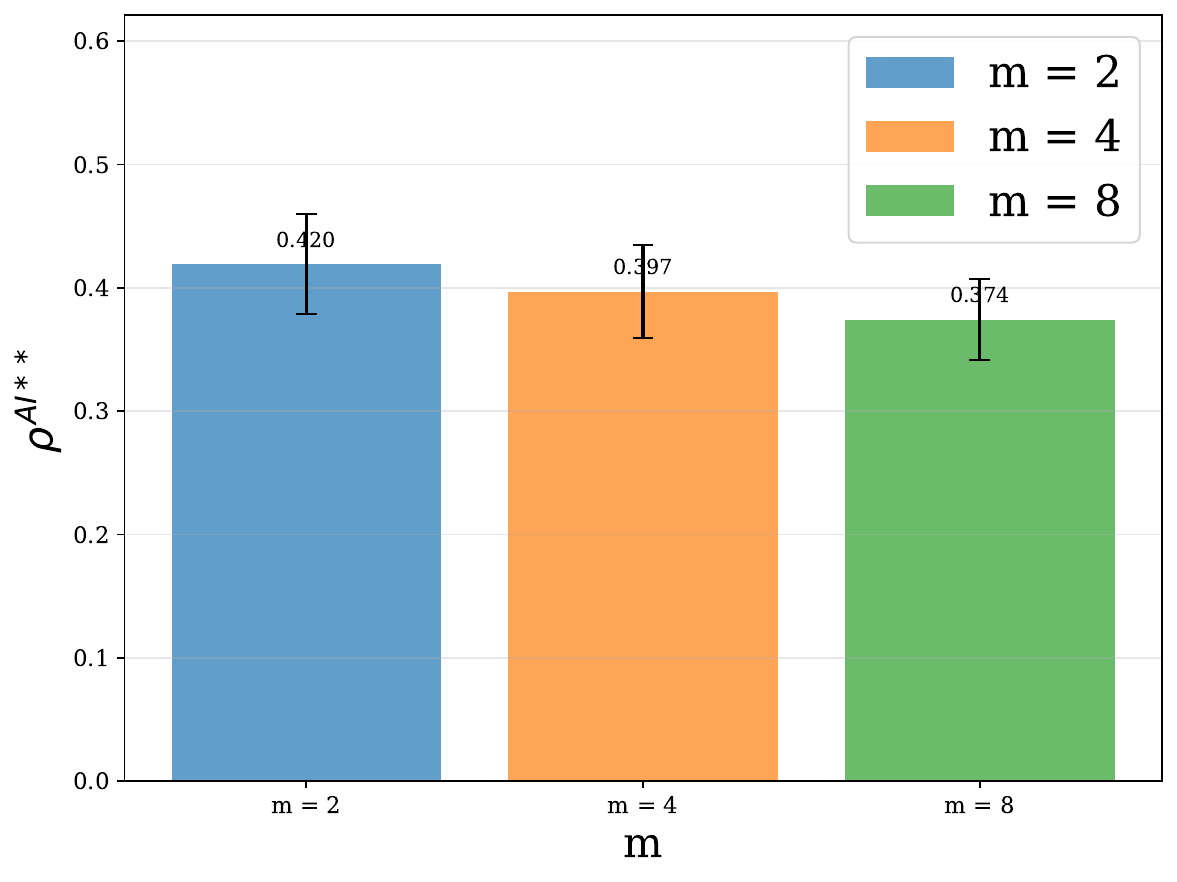}
    \label{fig:exp3b}
  }
  \vspace{-3mm}
  \caption{
    \textbf{Impact of consciousness layer parameters on rumor propagation outcomes.}
    (a) presents a three-dimensional multi-layer visualization of final infection densities as functions of awareness rate $\lambda$ and forgetting rate $\delta$ for different consciousness layer connectivities ($m = 2, 4, 8$). Each bar represents the steady-state infection density (AI fraction) for a specific parameter combination, with height and color indicating the infection level (see colorbar). The visualization demonstrates how optimal combinations of high $\lambda$ and low $\delta$ effectively suppress rumor spread across all network configurations.
    (b) shows the average final infection density for different consciousness layer degrees $m$, computed over all $\lambda$-$\delta$ parameter combinations. Error bars indicate standard deviation. The results clearly show that increased connectivity in the awareness layer ($m = 8$) leads to significantly lower rumor prevalence compared to sparser networks ($m = 2$), highlighting the crucial role of information dissemination topology in rumor control. Together, these findings provide strong empirical evidence for the importance of the consciousness layer in the coupled UAU-SIS model.
}
\vspace{-2mm}
  \label{fig:exp3_consciousness_impact}
\end{figure}

The analysis of the three-dimensional visualization in Fig.~\ref{fig:exp3a} reveals a consistent pattern across all consciousness layer configurations: higher forgetting rates $\delta$ coupled with lower corrective information spreading rates $\lambda$ consistently lead to elevated steady-state rumor prevalence, with observed increases in the range of 20\%--30\% between optimal and worst-case parameter combinations. This observed relationship aligns perfectly with theoretical expectations: rapid forgetting of factual information and sluggish correction propagation naturally diminish the proportion of informed individuals in the population, thereby facilitating more extensive rumor propagation. The systematic variation in rumor outcomes with information parameters provides compelling evidence that the characteristics of corrective information dissemination significantly influence rumor spreading dynamics.

This finding leads to important implications for rumor management strategies in knowledgeable communities. The parameter space defined by $(\lambda, \delta)$ essentially represents different information campaign strategies: aggressive debunking information spreading with persistent memory ($\lambda \rightarrow 1, \delta \rightarrow 0$) versus passive correction dissemination with quick forgetting ($\lambda \rightarrow 0, \delta \rightarrow 1$). The consistent 20\%--30\% variation in rumor prevalence across this parameter space underscores the meaningful impact of strategic information management in responding to rumor outbreaks. Furthermore, the robust emergence of this pattern across different network connectivities suggests that while the network structure modulates the quantitative impact, the fundamental relationship between information dynamics and rumor outcomes remains consistently observable.

Beyond the parameter effects observed in Fig.~\ref{fig:exp3a}, the three-dimensional visualization also reveals distinct infection density distributions for different $m$, indicating that the topological structure of the consciousness layer itself significantly influences rumor propagation dynamics. This structural effect is more explicitly quantified in Fig.~\ref{fig:exp3b}, which clearly demonstrates that higher connectivity in the awareness layer ($m = 8$) leads to a substantially lower average rumor prevalence compared to sparser networks ($m = 2$).

The underlying mechanism for this topological dependence can be attributed to the enhanced efficiency of debunking information dissemination in better-connected networks. Higher values of $m$ increase the average degree of the consciousness layer, facilitating more rapid and extensive spread of corrective information throughout the highly knowledgeable community. This creates a larger cohort of informed individuals who exhibit a reduced susceptibility to believing and spreading rumors, thereby impeding the propagation process in the rumor layer. The observed reduction in rumor prevalence with increasing $m$ highlights the crucial role of network connectivity in determining the effectiveness of information-based interventions in academic and research environments.

From a practical perspective, these findings emphasize the importance of building well-connected communication networks for factual information dissemination during rumor outbreaks in knowledgeable communities. The structural properties of information channels are shown to be as critical as the content parameters ($\lambda$, $\delta$) in determining rumor propagation outcomes. This dual dependence on both topological and dynamical aspects of the consciousness layer reinforces the necessity of the multilayer network framework for accurately modeling coupled information-rumor systems and designing effective debunking strategies in educational and research institutions and high-tech companies.
\vspace{-3mm}
\subsection{Impact of Network Topology on Rumor Propagation}
\label{sec:sm_exp4}
Building upon the findings from the previous simulations, we now investigate how the topological 
characteristics of both network layers collectively influence rumor propagation dynamics in the 
coupled UAU-SIS model. While the previous subsection examined the role of connectivity within 
a fixed network family (Barabási-Albert), this simulation systematically explores the effects of 
different network architectures on both the information and risk layers.

We consider three basic network types for each layer: Erdős-Rényi (ER) random graphs, 
Barabási-Albert (BA) scale-free networks, and regular graphs. These topologies correspond to structures widely observed in real social and information systems: scale-free degree distributions characterize many online social platforms and communication networks, arising naturally through the mechanisms of network growth and preferential attachment \cite{barabasi1999emergence}; random-graph-type structures approximate weakly structured or homogenized communication channels; and regular, strongly clustered structures resemble fixed organizational or acquaintance circles. The spreading of rumors over ER and scale-free network models has been analyzed theoretically \cite{nekovee2007theory}, which showed that scale-free topologies substantially enhance rumor propagation. This design creates a comprehensive 
3×3 parameter space encompassing all possible combinations of information and risk layer topologies. 
To isolate topological effects from connectivity differences, all networks are generated to maintain 
consistent average degrees: approximately 4 for the information layer and 2 for the risk layer. 
Other model parameters are fixed at $\alpha = 0.8$, $\lambda = 0.1$, $\beta^U = 1.5$, 
$\delta = 0.5$, and $\mu = 1.0$, with a population size of $N = 500$ nodes. The comparative results across all nine network combinations, which reveal the complex interplay 
between information dissemination structure and rumor propagation pathways, are presented in 
Fig.~\ref{fig:exp4_network_comparison}.

\begin{figure*}[htbp]
  \centering
  \includegraphics[width=\linewidth]{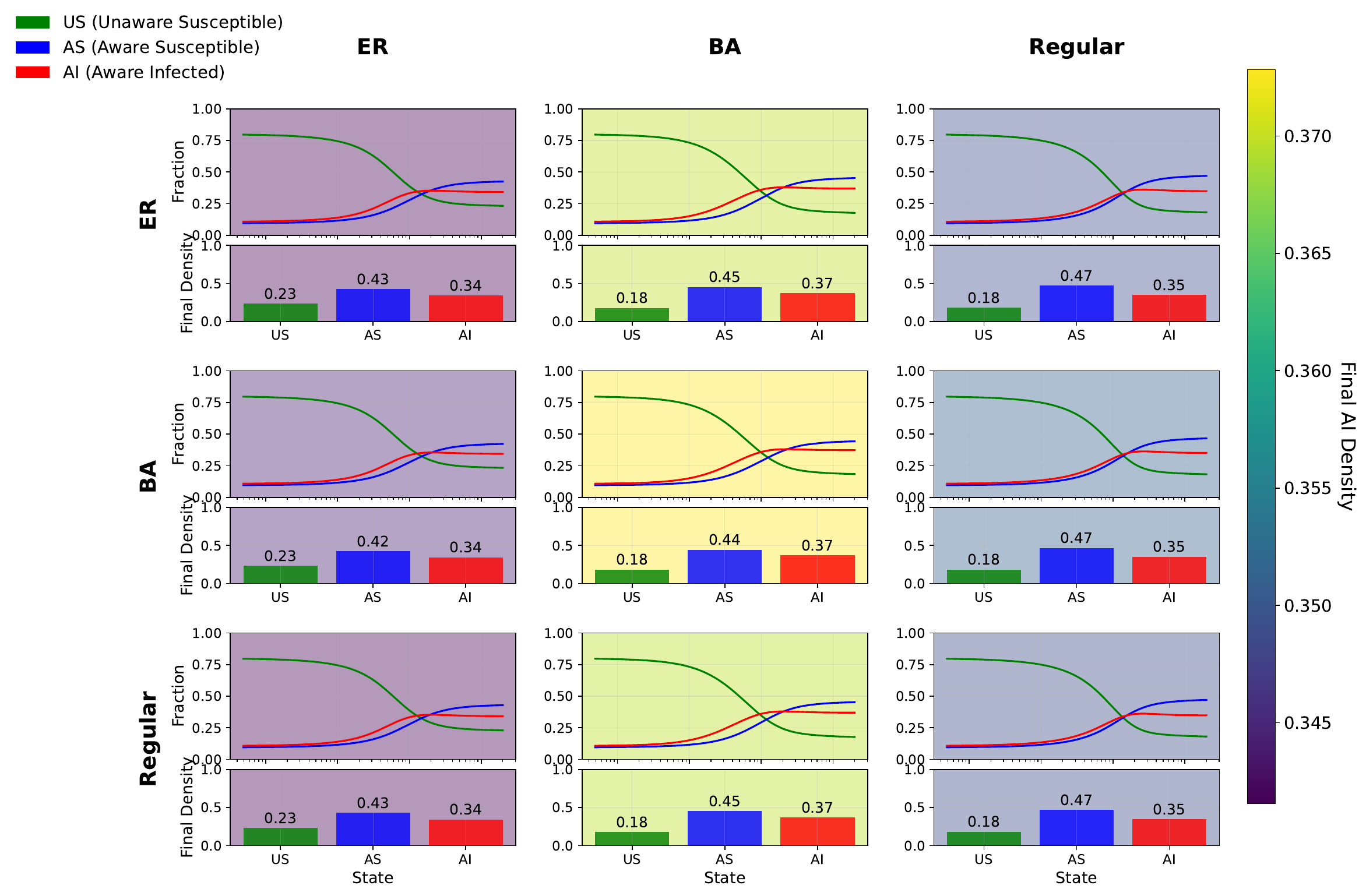}
  \vspace{-10mm}
  \caption{
    \textbf{Comparison of rumor propagation dynamics across different network topology combinations.}
    The 3×3 grid compares nine different combinations of information layer (rows: ER, BA, Regular) and risk layer (columns: ER, BA, Regular) network topologies. Each cell contains two panels: the upper panel shows the temporal evolution of state densities (US, AS, AI) on a logarithmic time axis, while the lower panel displays the final steady-state densities as a bar chart. The background color of each cell corresponds to the final infection density (AI state), with the colorbar providing the reference scale. All networks were constructed to maintain approximate average degrees of 4 for the information layer and 2 for the risk layer.
  }
  \vspace{-6mm}
  \label{fig:exp4_network_comparison}
\end{figure*}

Inspection of Fig.~\ref{fig:exp4_network_comparison} reveals several key patterns that illuminate the complex role of network topology in coupled rumor propagation systems. First, the consistent state proportions (US, AS, AI) within each column indicate that information layer topology exerts relatively limited influence on the final rumor propagation outcomes. This phenomenon can be understood through the model's awareness diffusion mechanism: awareness spreading depends on the proportion of aware neighbors $M_i^A(t)$ in the information layer network (matrix $\mathbf{A}$). Since all information layer networks are constrained to the same average degree (approximately 4), topological variations, such as degree heterogeneity in BA networks, may be masked by the average connectivity. Furthermore, the relatively low awareness spreading rate ($\lambda = 0.1$) likely limits the manifestation of topological effects in awareness diffusion. Consequently, despite differing information layer topologies, awareness propagation dynamics converges to similar patterns under identical risk layer conditions, resulting in consistent state distributions.

In contrast, the risk layer topology demonstrates a pronounced impact on rumor propagation outcomes. Specifically, BA networks exhibit the highest infection density (AI state), ER networks the lowest, and regular networks intermediate values. This pattern primarily stems from degree distribution heterogeneity in the risk layer: highly connected nodes (hubs) in BA networks act as super-spreaders that accelerate rumor diffusion, as infection propagation depends on the proportion of infected neighbors $M_i^I(t)$ in the risk layer network (matrix $\mathbf{B}$). In BA networks, $M_i^I(t)$ grows more rapidly due to the presence of highly connected nodes that become infected early and propagate rumors efficiently. The randomness of ER networks leads to more uniform but slower propagation, while the regular structure of regular networks results in an intermediate propagation speed. Thus, the risk layer topology directly determines rumor severity by modulating infection dynamics.

Interestingly, the risk layer topology also indirectly influences awareness distribution patterns. For instance, regular networks display the highest AS density (aware but uninfected), while ER networks show the lowest, and BA networks intermediate. This phenomenon can be explained through the model's state coupling mechanism: infected nodes immediately convert to the aware state ($\kappa \to \infty$), but the proportion of aware nodes is constrained by infection dynamics. In BA networks, rapid infection causes more nodes to convert to the AI state, reducing accumulation in the AS state; whereas in regular networks, slower infection allows more nodes to maintain the AS state. BA networks exhibit intermediate behavior. Therefore, the risk layer topology regulates the awareness state distribution through the infection process, highlighting the dynamic coupling between the two network layers.

These findings underscore the dominant role of the risk layer topology in rumor propagation, while the information layer effects appear relatively subdued under the current parameterization. This insight suggests that rumor management in highly knowledgeable communities should prioritize the structural characteristics of risk networks (such as degree distribution), for example by identifying and isolating highly connected nodes to suppress rumor diffusion. Simultaneously, the coupled mechanism reveals that intervention strategies must consider the interactions between both network layers to achieve effective rumor control.

\end{document}